\documentclass[acmsmall]{acmart}

\startPage{1}

\acmJournal{PACMPL}
\acmVolume{1}
\acmNumber{XXXX} 
\acmArticle{1}
\acmYear{2025}
\acmMonth{1}
\acmDOI{} 

\copyrightyear{2025}
\setcopyright{rightsretained}

\usepackage{mathtools}
\usepackage[capitalise]{cleveref}
\usepackage{tikz}
\usetikzlibrary{calc}
\usepackage{xifthen}
\usepackage{makecell}
\usepackage{stmaryrd}
\usepackage{xspace}
\usepackage{subcaption}
\usepackage{xstring} 

\usepackage{refcount}

\floatstyle{boxed}
\restylefloat{figure}

\newcommand\R{\mathbb R}
\newcommand\Z{\mathbb Z}
\newcommand\coloneqqq{\mathbin{\raisebox{0.04em}{::}\hspace{-0.25em}=}}

\newcommand\code[1]{\texttt{#1}}

\newcommand\hordead{\code{horde-ad}\xspace}
\newcommand\BOT{\textsc{bot}\xspace}
\newcommand\botnamelc{bulk-operation transformation\xspace}
\newcommand\botnametc{Bulk-Operation Transformation\xspace}

\newcommand\implementationLink{\url{https://hackage.haskell.org/package/horde-ad}}

\newcommand\previoussection[2]{\pgfmathparse{\getrefnumber{#1}==\value{section}-1}\ifnum \pgfmathresult>0\relax #2\else \cref{#1}\fi}
\makeatletter
\newcommand\previoussubsection[2]{%
    \StrBefore{\getrefnumber{#1}.}{.}[\@previoussubsectionSection]%
    \StrBehind{\getrefnumber{#1}.0}{.}[\@previoussubsectionSubsection]%
    \pgfmathparse{\@previoussubsectionSection==\value{section}&&\@previoussubsectionSubsection==\value{subsection}-1}%
    \ifnum \pgfmathresult>0\relax #2\else \cref{#1}\fi
}
\makeatother
\newcommand\nextsection[2]{\pgfmathparse{\getrefnumber{#1}==\value{section}+1}\ifnum \pgfmathresult>0\relax #2\else \cref{#1}\fi}

\newcommand\cc[1]{\textsf{\textcolor{blue!40!black}{#1}}}  
\newcommand\clam[1]{\lambda #1.\ }  
\newcommand\cletnosp{\textbf{\cc{let}}}
\newcommand\clet{\cletnosp\ }
\newcommand\cinnosp{\textbf{\cc{in}}}
\newcommand\cin{\cinnosp\ }
\newcommand\cletin[1]{\clet #1\ \cin}
\newcommand\ctranspx[2]{{#1}_{#2}}
\newcommand\ctransp[1]{\ctranspx{\cc{tr}}{#1}}

\newcommand\tyin{\tau_{\mathrm{in}}}

\newcommand\tyout{\tau_{\mathrm{out}}}

\newcommand\hslam[1]{\textrm{\textbackslash} #1 \to}
\newcommand\hsklcomp{\mathbin{\texttt{>=>}}}
\newcommand\name[1]{\text{#1}}
\newcommand\var[1]{\mathit{#1}}
\newcommand\keyw[1]{\mathbf{#1}}
\newcommand\Array{\name{Array}}

\newcommand\DArray{\name{DArray}}
\newcommand\Float{\name{Float}}
\newcommand\Map{\name{Map}}
\newcommand\Int{\name{Int}}
\newcommand\Bool{\name{Bool}}
\newcommand\Ix{\name{Ix}}
\newcommand\IZ{\name{IZ}}
\newcommand\IS{:::}
\newcommand\ASTIx{\name{ASTIx}}
\newcommand\ASTIxFun{\name{ASTIxFun}}
\newcommand\Sh{\name{Sh}}
\newcommand\ShZ{\name{SZ}}
\newcommand\ShS{\mathbin{\text{:\$:}}}
\newcommand\tVarName{\name{VarName}}
\newcommand\tDelta{\textsf{Delta}}

\newcommand\tSymDelta{\textsf{SymDelta}}
\newcommand\tDVarName{\textsf{DVarName}}
\newcommand\tID{\textsf{ID}}
\newcommand\tASTID{\mathrm{ID}_{\mathrm{AST}}}
\newcommand\genid{\underline{\textsc{GenID}}}
\newcommand\mgenid{\mathrm{genID}}
\newcommand\mgenastid{\mathrm{genID}_{\mathrm{AST}}}
\newcommand\tIdGen{\mathrm{IdGen}}
\newcommand\del[1]{\textsf{#1}}
\newcommand\eval{\mathit{eval}}
\newcommand\backprop{\mathit{backprop}}
\newcommand\reversePass{\mathit{reversePass}}
\newcommand\unshare{\mathit{unshare}}
\newcommand\shareToLet{\mathit{shareToLet}}
\newcommand\stackLets{\mathit{stackLets}}
\newcommand\id{\mathsf{id}}
\newcommand\EState{\name{ES}}
\newcommand\shapeDelta{\name{shapeDelta}}
\newcommand\cTensor{\ensuremath{\name{BaseTensor}}\xspace}
\newcommand\tAST{\name{AST}}
\newcommand\tADVal{\name{ADVal}}
\newcommand\tASTVarName{\name{ASTVarName}}
\newcommand\interpret{\mathit{interpret}}
\newcommand\interpretIx{\mathit{interIx}}
\newcommand\interpretIxFun{\mathit{interIxFun}}
\newcommand\tDMap{\name{DMap}}
\newcommand\tDMaptwo{\name{DMap}_2}
\newcommand\dmEmpty{\name{DMap}.\name{empty}}
\newcommand\dmInsert{\name{DMap}.\name{insert}}
\newcommand\dmInsertWith{\name{DMap}.\name{insertWith}}
\newcommand\dmLookup{\name{DMap}.\name{lookup}}
\newcommand\dmIndex{\mathbin{\name{DMap}.!}}
\newcommand\dmMaxViewWithKey{\name{DMap}.\name{maxViewWithKey}}
\newcommand\dmDelete{\name{DMap}.\name{delete}}
\newcommand\dmtwoEmpty{\name{DMap}_2.\name{empty}}
\newcommand\dmtwoInsert{\name{DMap}_2.\name{insert}}
\newcommand\dmtwoLookup{\name{DMap}_2.\name{lookup}}
\newcommand\arrlit[1]{\ensuremath{[ #1 ]}}
\newcommand\ixvec[1]{\ensuremath{[ #1 ]}}
\newcommand\shvec[1]{\ensuremath{[ #1 ]}}
\newcommand\ravel[1]{\ensuremath{[ #1 ]}}
\newcommand\vecadd{\mathbin{\vec+}}
\newcommand\vecmul{\mathbin{\vec\cdot}}

\newcommand\TDN{\name{TDN}}
\newcommand\mergeshare{\mathbin{{\cup}\scalebox{0.7}{\raisebox{0.15em}{$\downarrow$}}}}
\newcommand\ndD{D'}

\newcommand\ccomment[1]{\textit{\textcolor{gray}{--- #1}}}

\newcommand\leadcomma{\mathllap{,\mkern2mu}}
\newcommand\tightunderbrace[2]{\parbox[t]{\widthof{\(#1\)}}{\(#1\)\vspace{-0.7em}\\\upbracefill\vspace{-0.4em}\\\centering{\scriptsize\(#2\)}}}

\newenvironment{sizeddisplay}[1]
  {\par\nopagebreak\setlength\lineskip{0pt}#1\noindent\ignorespaces}
  {\nopagebreak\ignorespacesafterend}

\newcommand\TS[1]{}\newcommand\MK[1]{}\newcommand\simon[1]{}\newcommand\TODO[1]{}\newcommand\TODOfootnote[1]{}

\newenvironment{ietarray}{\begin{array}{@{}r@{\quad}c@{\quad}l@{}}}{\end{array}}
\newcommand\ietrule[3]{
    #1 &
    \ifthenelse{\isempty{#3}}
        {{\leadsto}}
        {
         \thinmuskip=0mu\relax\medmuskip=0mu\relax\thickmuskip=0mu\relax
         \underset{\mathclap{\textit{(#3)}}}{{\leadsto}}}
    & #2
}
\newcommand\multilineT[1]{\begin{array}[t]{@{}l@{}} #1 \end{array}}
\newcommand\multilineM[1]{\begin{array}{@{}l@{}} #1 \end{array}}

\newcommand\jisindex[3]{#1 \vdash #2 \text{ is an } #3\text{-dim.\ index}}
\newcommand\jisnumeric[1]{#1 \text{ numeric}}
\newcommand\infrule[2]{\frac{\begin{array}{@{}c@{}} #1 \end{array}}{\begin{array}{@{}c@{}} #2 \end{array}}} 
\newcommand\precsep{\qquad} 
\newcommand\precdots{\;\; \ldots \;\;} 
\newcommand\rulesep{\qquad} 

\begin{document}

\title{Dual-Numbers Reverse AD for Functional Array Languages}
\titlenote{%
    This technical report documents our progress in applying the dual-numbers construction to reverse AD for array programs.
    While there is material for more than a single paper here, we are still unsatisfied with the elegance of the algorithm design, and we do not yet know if these algorithms can be improved (without extensive special-casing) to reach performance competitive with popular machine learning frameworks.
    If you would like to collaborate with us in improving the algorithms or the presentation as one or more papers, please reach out!
}
\subtitle{Technical Report, July 2025}

\author{Tom Smeding}
\orcid{0000-0002-4986-6820}
\affiliation{%
    \department{Department of Information and Computing Sciences}
    \institution{Utrecht University}
    \city{Utrecht}
    \country{The Netherlands}
}
\email{t.j.smeding@uu.nl}
\authornote{Equal contribution}

\author{Miko\l aj Konarski}
\orcid{0009-0008-7585-7590}
\affiliation{%
    \institution{Well-Typed}
    \country{UK}
}
\email{mikolaj@well-typed.com}
\authornotemark[2]

\author{Simon Peyton Jones}
\orcid{0000-0002-6085-1435}
\affiliation{
    \institution{Epic Games}
    \country{USA}
}
\email{simon.peytonjones@gmail.com}

\author{Andrew Fitzgibbon}
\orcid{0000-0002-9839-660X}
\affiliation{
    \institution{Graphcore}
    \country{UK}
}
\email{awf@graphcore.ai}

\begin{abstract}
The standard dual-numbers construction works well for forward-mode automatic differentiation (AD) and is attractive due to its simplicity; recently, it also has been adapted to reverse-mode AD, but practical performance, especially on array programs, leaves a lot to be desired.
In this paper we introduce first-class support for multidimensional arrays in dual-numbers reverse-mode AD with little to no performance overhead.
The algorithm consists of three loosely-coupled components: a semantics-preserving vectorisation code transformation (the \emph{bulk-operation transform} or \BOT), a fairly straightforward lifting of the basic dual-numbers reverse AD algorithm to a mostly first-order array language, and symbolic interpretation to achieve an end-to-end compilation pipeline.
Unfortunately, we lose some of the nice generalisable aspects of dual-numbers AD in the process, most importantly support for higher-order code.

We do support some higher-order array combinators, but only a carefully-chosen set: `build' (elementwise array construction), `gather' and `scatter'.
In return, the \BOT can eliminate the \emph{essential} (for AD) higher-orderness of the input program, meaning that AD gets essentially presented with a first-order program.
This allows the naive trick of lifting dual numbers to ``dual arrays'' to work without much modification.
\end{abstract}

\ccsdesc[500]{Mathematics of computing~Automatic differentiation}
\ccsdesc[300]{Software and its engineering~Functional languages}

\keywords{automatic differentiation, functional programming, array programming}

\maketitle

\section{Introduction}

The classical dual-numbers technique for forward-mode automatic differentiation (AD) is remarkably simple, yet generalises to very expressive, higher-order languages without effort.
The technique can be extended to reverse-mode AD as well~\cite{2022-krawiec-dualrev,2023-smeding-dualrev}, introducing a little more complexity but nevertheless still scaling well to expressive languages.
Furthermore, dual-numbers reverse AD is simple enough that it admits a full correctness proof (e.g.~\cite{2020-ad-gluing,2022-krawiec-dualrev,2022-ad-logical-relations}).
However, it also has a big problem in practice: it is unacceptably inefficient for array programs, because
it differentiates each scalar operation individually.
As far as we know, all published reverse AD algorithms so far \emph{either} have a proof of correctness,
\emph{or} are fast --- but never both at the same time.

In this paper, we improve this situation by fixing the performance problem of dual-numbers reverse AD for array programs.
The algorithm retains a simple core that is mostly unchanged from naive dual-numbers reverse AD, a correctness proof for which would be a straightforward extension of the proofs referenced above.\footnote{The algorithm also includes pre- and post-processing stages, which are rewrite systems that fairly directly follow from the equational theory on arrays and are thus semantics-preserving.}
However, in lifting the algorithm to arrays, we unfortunately lose the effortless generality of the naive approach: in particular, we can no longer support higher-order input code, nor full dynamic control flow, and we have to limit the primitive higher-order array operations.
Nevertheless, the language still allows programmers to use element-wise computation (rather than
forcing them to use bulk operations) and is amply expressive for many array applications.

While we present the AD algorithm as source-to-source transformations in this paper, our Haskell implementation (\hordead, see \cref{sec:implementation}) is based on type class instantiations, not unlike the very pretty presentation of dual-numbers forward AD in~\cite{2009-elliott-typeclass-ad}.
We further reuse this type class infrastructure to implement staging and symbolic execution in a single instance.
An in-depth description of these topics can be found in appendices linked from \cref{sec:implementation}.


The main contributions of this paper are as follows:
\begin{itemize}
\item A concise, new presentation of dual-numbers reverse AD, incorporating the lessons of \cite{2022-krawiec-dualrev,2023-smeding-dualrev} (\cref{sec:ad-naive}).
\item An aggressive vectorisation transform (the \emph{\botnamelc}), described in \cref{sec:bot}, that transforms the programs we want to \emph{write} (which use element-at-a-time computation) into the programs we want to \emph{differentiate} (which use bulk operations only).
\item A source-to-source transformation that lifts the ``dual-number'' approach to AD
  to ``dual arrays'' (\cref{sec:ad-dual-arrays}).
\item An analysis on the structure of the output of the algorithm, that allows us to make the differentiation algorithm fully symbolic (compile-time), eliminating all runtime overhead of the differentiation algorithm over the actual gradient code (\cref{sec:restaging}).
\item The type class implementation of the algorithms, detailed in \cref{sec:implementation} and thereafter.
\end{itemize}

\section{Background: Dual-Numbers Reverse-Mode AD}\label{sec:ad-naive}

The reverse AD algorithm used in this paper is mostly an extension of the scalar-level algorithm described in the literature~\cite{2022-krawiec-dualrev,2023-smeding-dualrev}, but we make a few small changes in its implementation details.
In this section we describe the scalar-level algorithm that we build on; we lift this algorithm to arrays in \cref{sec:ad-dual-arrays}.

\subsection{Input language}\label{sec:ad-naive-input-language}

Assume a simply-typed lambda calculus with products and ground types $\R$, $\Int$ with their standard primitive operations:\footnote{We use ``$\R$'' to denote the type of floating-point numbers in use by the program, e.g.\ \texttt{double}.}
\[ \begin{array}{@{}r@{\;}r@{\;}l@{}}
    \sigma, \tau &\coloneqqq& \R \mid \Int \mid (\sigma, \tau) \mid \sigma \to \tau \\
    s, t, u, v &\coloneqqq& r \mid k \mid x \mid \mathbf{let}\ x = s\ \mathbf{in}\ t\mid (s, t) \mid \name{fst}\ t \mid \name{snd}\ t \\
    &\mid& \lambda x.\ t \mid s\ t \mid \mathbf{if}\ t > 0\ \mathbf{then}\ u\ \mathbf{else}\ v \mid \mathit{op}(t_1, \ldots, t_n)
\end{array} \]
where $r$ stands for a scalar constant (i.e.\ of type $\R$), $k$ stands for an integral constant, $x$ is a variable reference, and $\mathit{op}$ stands for any arithmetic operation such as $+_\R$, $\times_\R$, $+_{\Int}$, $\name{round} : \R \to \Int$, etc.
The word ``scalar'' will always refer to a \emph{real scalar}, i.e.\ a value of type $\R$, in this paper.
The let-binding construct can be recursive without problems, as long as we are only asked to differentiate input programs that do in fact terminate.
We assume call-by-value evaluation semantics in this paper.
Polymorphism and various further language extensions could be supported by the naive algorithm described in this section, but since the extension to arrays in the main contribution of this paper does not easily support such extensions, we refrain from over-generalising here.

While this language is higher-order (it has full lambda abstraction and application), the top-level program being differentiated (i.e.\ the model to be trained, or the function to be optimised, etc.) must have \emph{first-order} input and output types.
That is to say: the input program can use function values internally as much as it likes, but we do not define what it means to take a derivative \emph{with respect to} a function, or the derivative of a function value with respect to something else.
Hence, the top-level program to differentiate must have a type of the form $\sigma_1 \to \cdots \to \sigma_n \to \tau$, where none of $\sigma_1,\ldots,\sigma_n,\tau$ mention the function arrow `$\to$'.
In fact, we will assume for simplicity that the type is $\tyin \to \R$ for first-order $\tyin$.
The restriction to a single input is without loss of generality because the language supports pairs; we further restrict the output to to a single scalar because (1) this is by far the most common case in applications of reverse AD, (2) it simplifies the wrappers (interfaces) around the core algorithm, and (3) generalisation to more general, yet still first-order, output types $\tyout$ is straightforward (see \cref{app:output-type-generalisation}).

\subsection{Code transformation} \label{sec:transform}

Automatic differentiation is performed by a source-to-source program transformation $D$,
given in \cref{fig:naive-dn-ad-terms}.  A source term $t : \tau$ is transformed to a dual-number
target term $D[t]:D[\tau]$.  Every real number of type $\R$ in $t$ becomes a dual number $(\R,\tDelta)$ in $D[t]$;
the type translation is given in \cref{subfig:naive-dn-ad-types}.  The type $\tDelta$
(defined in  \cref{subfig:naive-dn-ad-delta}) describes the derivative of the number it is
paired with: see \cref{sec:ad-naive-delta}.

Both transformations are quite simple, recursing over the structure of terms
and types respectively.
The only place where the transformation ``does something'' is where the program directly manipulates scalars;
see the upper group of rules in \cref{fig:naive-dn-ad-terms}.
In a sense, the source program is seen as nothing more than some procedure that once in a while performs some computation on scalars, and these computations are all that we are interested in.\footnote{This is what allows the algorithm to accept higher-order code without any work.}
Similarly, the type transformation (\cref{subfig:naive-dn-ad-types}) mostly just recurses over the structure of the type, except for scalars $\R$, which are mapped to a \emph{pair} of a scalar and a value of type $\tDelta$ (\cref{subfig:naive-dn-ad-delta}, explained in \cref{sec:ad-naive-delta}).

\begin{figure}
    \begin{subfigure}{0.48\textwidth}
    \begin{center}
    \( \multilineT{
        D[\R] = (\R, \tDelta) \\
        D[\Int] = \Int \\
        D[(\sigma, \tau)] = (D[\sigma], D[\tau]) \\
        D[\sigma \to \tau] = D[\sigma] \to D[\tau]
    } \)
    \end{center}
    \caption{\label{subfig:naive-dn-ad-types}
        The type transformation.
    }
    \end{subfigure}
    \begin{subfigure}{0.48\textwidth}
    \begin{center}
    \( \mathbf{data}\ \tDelta \begin{array}[t]{@{\ }c@{\ }l}
        = & \del{Zero} \\
        | & \del{Input}\ \tDVarName \\
        | & \del{Add}\ \tDelta\ \tDelta \\
        | & \del{Scale}\ \R\ \tDelta
    \end{array} \)
    \end{center}
    \caption{\label{subfig:naive-dn-ad-delta}
        Defunctionalised forward derivatives.
    }
    \end{subfigure}
    \caption{\label{fig:naive-dn-ad}
        Types for naive, scalar-level dual-numbers reverse AD; presentation after~\cite{2022-krawiec-dualrev}.
    }
\end{figure}

\begin{figure}
    \begin{gather*}
    x_1 : \tau_1, \ldots, x_n : \tau_n \vdash t : \tau \quad \leadsto \quad x_1 : D[\tau_1], \ldots, x_n : D[\tau_n] \vdash D[t] : D[\tau] \\
    \begin{array}{@{}l@{}}
        \ccomment{Primitive operations on real numbers include derivative computations:} \\[0.3em]
        \begin{array}{@{}r@{\;}c@{\;}l@{}}
            D[r]
                &=& (r, \del{Zero}) \\
            D[\sin t]
                &=& \multilineT{
                    \mathbf{let}\ (x, d) = D[t] \\
                    \mathbf{in}\ (\sin x, \del{Scale}\ (\cos x)\ d)
                } \\
            D[t_1 +_{\R} t_2]
                &=& \multilineT{
                    \mathbf{let}\ (x_1, d_1) = D[t_1]; (x_2, d_2) = D[t_2] \\
                    \mathbf{in}\ (x_1 +_{\R} x_2, \del{Add}\ d_1\ d_2)
                } \\
            D[t_1 \times_{\R} t_2]
                &=& \multilineT{
                    \mathbf{let}\ (x_1, d_1) = D[t_1]; (x_2, d_2) = D[t_2] \\
                    \mathbf{in}\ (x_1 \times_{\R} x_2, \del{Add}\ (\del{Scale}\ x_2\ d_1)\ (\del{Scale}\ x_1\ d_2))
                }
            \vspace{0.3em}
        \end{array} \\
        \ccomment{In general for $t_1 : \R, \ldots, t_n : \R$ and $\mathit{op}(t_1, \ldots, t_n) : \R$:} \\[0.3em]
        \begin{array}{@{}r@{\;}c@{\;}l@{}}
            D[\mathit{op}(t_1, \ldots, t_n)]
                &=& \multilineT{
                    \mathbf{let}\ (x_1, d_1) = D[t_1]; \ldots; (x_n, d_n) = D[t_n] \\
                    \mathbf{in}\ (\multilineT{
                        \mathit{op}(x_1, \ldots, x_n) \\
                        \leadcomma \del{Add}\ (\del{Scale}\ \bigl(\frac{\partial \mathit{op}(x_1,\ldots,x_n)}{x_1}\bigr)\ d_1)\ (\del{Add}\ \ldots\ (\del{Scale}\ \bigl(\frac{\partial \mathit{op}(x_1,\ldots,x_n)}{x_n}\bigr)\ d_n)))
                    }
                }
            \vspace{0.3em}
        \end{array} \\
        \ccomment{Everything else is structure-preserving:} \\[0.3em]
        \begin{array}{@{}r@{\;}c@{\;}l@{\qquad}r@{\;}c@{\;}l@{}}
            D[x] &=& x
                & D[k : \Z] &=& k \\
            D[(s, t)] &=& (D[s], D[t])
                & D[s \times_{\Z} t] &=& D[s] \times_{\Z} D[t] \\
            D[\name{fst}\ t] &=& \name{fst}\ D[t]
                & D[\mathbf{let}\ x = s\ \mathbf{in}\ t] &=& \mathbf{let}\ x = D[s]\ \mathbf{in}\ D[t] \\
            D[\lambda x.\ t] &=& \lambda x.\ D[t]
                & D[\mathbf{if}\ t > 0\ \mathbf{then}\ u\ \mathbf{else}\ v] &=& \mathbf{if}\ \name{fst}\ D[t] > 0\ \mathbf{then}\ D[u]\ \mathbf{else}\ v \\
            D[s\ t] &=& D[s]\ D[t] \\
        \end{array}
    \end{array}
    \end{gather*}
    \caption{\label{fig:naive-dn-ad-terms}
        Selected rules from the code transformation for dual-numbers reverse AD, complementing \cref{fig:naive-dn-ad}.
    }
\end{figure}

When looking at the term transformation in \cref{fig:naive-dn-ad-terms}, one should note that the first component of an $(\R, \tDelta)$ pair is always equal to the $\R$ value that would have been computed in the source program; it is
often called the \emph{primal}.
Thus, every intermediate value computed by the source program is also computed by the transformed program, and in the same order.

\paragraph{Larger target language}
The code transformation $D[-]$ from \cref{fig:naive-dn-ad-terms} maps from the small input language from \cref{sec:ad-naive-input-language} into a larger output language, that also includes the $\tDelta$ type and its constructors.
This is typical for automatic differentiation: a particular term language is not necessarily closed under differentiation in the first place (e.g.\ $\times_\R$ begets addition; `$\log$' begets division; `$\arcsin$' begets `$\name{sqrt}$'), especially reverse-mode differentiation (where e.g.\ product constructors beget product projections and vice-versa).
Furthermore, with dual-numbers reverse-mode AD, which looks a lot like tracing AD~\cite{2023-smeding-dualrev}, there must be some way to represent this trace (i.e.\ the $\tDelta$ data type), which is something that did not exist in the source program.

In the rest of this paper, the differentiating code transformation will continue to map into a larger language (its type system and semantics will always be clear from context).
This means that the algorithm in this paper cannot be used directly for higher derivatives via iterated differentiation; we leave computation of higher derivatives by computing a longer Taylor series prefix (see e.g.~\cite{2022-higher-order-derivs}) to future work.

\subsection[Delta terms]{$\tDelta$ terms}\label{sec:ad-naive-delta}

\newcommand{\dir}{\varepsilon}

In this section we focus on the mysterious $\tDelta$ that is the second component of
each dual number.  As we shall see, it represents the derivative of the first (primal) component.
First, we need to establish some notation.  Given a function\footnote{For now we consider only
  functions of type $\R^n \to \R$.  Generalising to richer input and output types is straightforward, but
  adds a lot of notational clutter.}
$f : \R^n \to \R$, the vector
\[
    \nabla_v\, f \;=\; \left( \frac{\partial f(v)}{\partial v_1}, \ldots, \frac{\partial f(v)}{\partial v_n} \right)
\]
is the vector of partial derivatives of $f$ at input $v:\R^n$, with respect to each of its $n$ inputs $v = (v_1 \ldots v_n)$.  That is, the $i$'th component of the vector $\nabla_v\, f$ describes how the output varies as you vary $v_i$.
$\nabla f$ is the derivative of $f$ in the following sense. For any small vector $\dir:\R^n$,
\begin{equation}
    (\nabla_v\ f) \odot \dir \; \approx \; f(v+\dir) - f(v)
    \label{eq:finite-difference}
\end{equation}
That is, the dot-product $\odot$ of $\nabla_v\, f$ with any small vector $\dir$ multiplies each partial
derivative of $f$ by the corresponding component of $\dir$, and adds the results, to give (approximately)
the difference between $f(v+\dir)$ and $f(v)$.
In other words, $\dir \mapsto (\nabla_v\ f) \odot \dir$ is the best local linear approximation of $f$ at $v$, also known as a \emph{directional derivative}.

The goal of AD is to compute $\nabla_v\, f$.  We do so in two steps. First, we use the transformation
in \cref{sec:transform} to transform the function $f : \R^n \to \R$ into $f' : (\R, \tDelta)^n \to (\R, \tDelta)$,\footnote{This $f'$ implements the forward derivative, also known as \emph{total derivative}, of $f$.}
and then we use an \emph{evaluator} to convert the $\tDelta$ to an actual derivative.
The correctness criterion for the transformation is this:
\[ \begin{array}{ll}
    \text{if} & f'((v_1,\del{Input}\ 1), \ldots, (v_n,\del{Input}\ n)) = (r, d) \\
    \text{then} & r = f(v_1, \ldots, v_n)\\
    \text{and}  & \eval_0\ d\ \dir = (\nabla_v\, f) \odot \dir
\end{array} \]
Here $\del{Input}\ 1, \ldots, \del{Input}\ n$ are values in the $\tDelta$ type (\cref{subfig:naive-dn-ad-delta});
each is paired with its corresponding input value, $v = (v_1 \ldots v_n)$.  The primal result of $f'$ is $r$, and will
be equal to $f(v)$.  The second component of the result is $d:\tDelta$.  This data structure
describes, or represents, the derivative of $f$.  More precisely, we can \emph{evaluate} $d$
in a direction described by $\dir$ to get $(\nabla_v\, f) \odot \dir$, see \cref{eq:finite-difference} above.

The evaluation function $\eval_0$ could not be more straightforward: it just interprets $\del{Add}$ as addition,
$\del{Scale}$ as multiplication, and so on:
\[ \begin{array}{@{}l@{\;}l@{\;}c@{\;}l@{}}
\multicolumn{4}{@{}l@{}}{\eval_0 :: \tDelta \to \R^n \to \R} \\
\eval_0\ \del{Zero} & \dir &=& 0 \\
\eval_0\ (\del{Input}\ i) & \dir &=& \text{(the $i$'th component of $\dir$)} \\
\eval_0\ (\del{Add}\ d_1\ d_2) & \dir &=& \eval_0\ d_1\ \dir + \eval_0\ d_2\ \dir \\
\eval_0\ (\del{Scale}\ r\ d) & \dir &=& r \cdot \eval_0\ d\ \dir
\end{array} \]

\subsection{Efficient gradients 1: a single pass}\label{sec:ad-naive-reverse-mode}


Remember that our driving goal is to compute $\nabla_v\, f$.  How can we do that, given $f'$?  One obvious way is
to call $\eval_0$ $n$ times, like this:
\begin{equation}
    \begin{array}{@{}l@{\hspace{1mm}}l}  
        \nabla_v\, f & = (
            \eval_0\ d\ (1, 0, \ldots, 0),
            \ldots,
            \eval_0\ d\ (0, 0, \ldots, 1))
        \\
     & \text{where}\\
     & \hspace{5mm} (r,d) = f' ((v_1,\del{Input}\ 1),\ldots, (v_n,\del{Input}\ n)) \\
    \end{array}
    \label{eq:grad-many-eval}
\end{equation}
But, following \cite{2022-krawiec-dualrev,2023-smeding-dualrev}, a natural optimisation is to
write a new evaluator $\eval_1$, that computes those $n$ results simultaneously, thus:
\begin{equation*}
\begin{array}{l@{\ }c@{\ }l}
\mathrlap{\eval_1 :: \tDelta \to \R^n} \\
\eval_1\ \del{Zero} &=& (0, 0, \ldots, 0) \\
\eval_1\ (\del{Input}\ i) &=& \text{(one-hot vector with 1 at position $i$)} \\
\eval_1\ (\del{Add}\ d_1\ d_2) &=& \eval_1\ d_1 + \eval_1\ d_2 \\
\eval_1\ (\del{Scale}\ r\ d) &=& r \cdot \eval_1\ d
\end{array}
\end{equation*}
where $(+)$ and $(\cdot)$ operate elementwise.
With this formulation we only need to call $\eval_1$ once, but it creates, scales, and adds, many $n$-vectors,
which is not efficient if $n$ is large.
Fortunately, it is not difficult to transform $\eval_1$ into an evaluator that transforms a \emph{single} $n$-vector, and
does so in a completely single-threaded way, amenable to mutable in-place updates.
To do so, we apply Cayley transformation (also known as the ``difference list trick''):
\begin{equation*}
\begin{array}{l@{\ }c@{\ }l}
\mathrlap{\eval_2 :: \tDelta \to \R^n \to \R^n} \\
\eval_2\ \del{Zero} &=& \id \\
\eval_2\ (\del{Input}\ \var{name}) &=& \clam{\mathit{tg}} \text{($\mathit{tg}$ with $1$ added to position $\var{name}$)} \\
\eval_2\ (\del{Add}\ d_1\ d_2) &=& \eval_2\ d_2 \circ \eval_2\ d_1 \\
\eval_2\ (\del{Scale}\ r\ d) &=& (r \cdot) \circ \eval_2\ d
\end{array}
\end{equation*}
We have $\eval_1\ d = \eval_2\ d\ (0, \ldots, 0)$.

While this has improved the efficiency of the first three cases, the case for $\del{Scale}$ still does $O(n)$ work to multiply by $r$ elementwise, where we would like it to be $O(1)$ instead.
To address this, we ``push'' the scale factors $r$ inside so that at \del{Input} nodes, we do not add just `1', but the value it would have been after all the scaling factors have been applied:
\begin{equation*}
\begin{array}{l@{\ }c@{\ }l}
\mathrlap{\eval_3 :: \R \to \tDelta \to \R^n \to \R^n} \\
\eval_3\ c\ \del{Zero} &=& \id \\
\eval_3\ c\ (\del{Input}\ \var{name}) &=& \clam{\mathit{tg}} \text{($\mathit{tg}$ with $c$ added to position $\var{name}$)} \\
\eval_3\ c\ (\del{Add}\ d_1\ d_2) &=& \eval_3\ c\ d_2 \circ \eval_3\ c\ d_1 \\
\eval_3\ c\ (\del{Scale}\ r\ d) &=& \eval_3\ (c \cdot r)\ d
\end{array}
\end{equation*}
We have $\eval_2\ d\ v = \eval_3\ 1\ d\ v$.

Readers familiar with reverse AD algorithms will recognise that in rewriting $\eval_2$ to $\eval_3$, the accumulation order of the derivative values has flipped from forward order to reverse order.
Indeed, for the following example $\tDelta$ term:
\begin{equation*}
\del{Scale}\ r_1\ (\del{Scale}\ r_2\ (\ldots\,(\del{Scale}\ r_n\ (\del{Input}\ n))))
\end{equation*}
$\eval_2$ would compute $r_1 \cdot (r_2 \cdot (\ldots \cdot (r_n \cdot 1)))$, whereas $\eval_3$ would compute $(((r_1 \cdot r_2) \cdot \ldots) \cdot r_n) \cdot 1$.
This reversal is expected in a \emph{reverse-mode AD} algorithm, which we need to be able to calculate gradients efficiently.
The name of the ``accumulating parameter'' in $\eval_3$, ``$c$'', is chosen because it contains the incoming \underline{c}otangent in a reverse AD algorithm.

While $\eval_3$ improves computational complexity significantly over $\eval_1$, we do still have a second problem: lost sharing.

\subsection{Efficient gradients 2: respecting sharing}\label{sec:ad-naive-sharing}

Consider the program $x : \R \vdash P_1 : \R$ given by:
\begin{equation*}
    \multilineT{
        \mathbf{let}\ \multilineT{
            x_1 = x +_\R x \\
            x_2 = x_1 +_\R x_1 \\
            \ \vdots \\
            x_n = x_{n-1} +_\R x_{n-1}
        } \\
        \mathbf{in}\ x_n
    }
\end{equation*}
The code transformation will transform this to $x : (\R, \tDelta) \vdash D[P_1] : (\R, \tDelta)$:
\begin{equation*}
    \multilineT{
        \mathbf{let}\ \multilineT{
            x_1 = \mathbf{let}\ (y_1, d_1) = x; (y_2, d_2) = x\ \mathbf{in}\ (y_1 +_\R y_2, \del{Add}\ d_1\ d_2) \\
            x_2 = \mathbf{let}\ (y_1, d_1) = x_1; (y_2, d_2) = x_1\ \mathbf{in}\ (y_1 +_\R y_2, \del{Add}\ d_1\ d_2) \\
            \ \vdots \\
            x_n = \mathbf{let}\ (y_1, d_1) = x_{n-1}; (y_2, d_2) = x_{n-1}\ \mathbf{in}\ (y_1 +_\R y_2, \del{Add}\ d_1\ d_2)
        } \\
        \mathbf{in}\ x_n
    }
\end{equation*}
The in-memory representation of the $\tDelta$ term in $x_n$ looks like this:
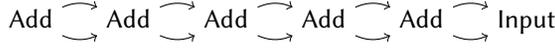
\vspace{0.15em}
\begin{center}
\begin{tikzpicture}[xscale=1.3]
    \node (n) at (0, 0) {$\del{\vphantom{Input}Add}$};
    \node (n-1) at (1, 0) {$\del{\vphantom{Input}Add}$};
    \draw[->] (n) edge[bend left=30] (n-1);
    \draw[->] (n) edge[bend right=30] (n-1);
    \node (dots) at (2, 0) {$\del{\vphantom{Input}Add}$};
    \draw[->] (n-1) edge[bend left=30] (dots);
    \draw[->] (n-1) edge[bend right=30] (dots);
    \node (1) at (3, 0) {$\del{\vphantom{Input}Add}$};
    \draw[->] (dots) edge[bend left=30] (1);
    \draw[->] (dots) edge[bend right=30] (1);
    \node (0) at (4, 0) {$\del{\vphantom{Input}Add}$};
    \draw[->] (1) edge[bend left=30] (0);
    \draw[->] (1) edge[bend right=30] (0);
    \node (x) at (5, 0) {$\del{\rlap{Input}\phantom{Add}}$};
    \draw[->] (0) edge[bend left=30] (x);
    \draw[->] (0) edge[bend right=30] (x);
\end{tikzpicture}
\end{center}
\vspace{-0.15em}
but a plain interpreter of a $\tDelta$ term cannot see this sharing, and thus evaluation ($\eval_0$ for forward AD, and $\eval_3$ for reverse AD) will be \emph{exponential} in $n$.
This is clearly unacceptable.

We follow \citeN{2023-smeding-dualrev} in solving this problem.
The idea is to ``simply'' make the sharing visible: we give every fragment of a $\tDelta$ term that may later be shared a unique name (implemented as a numeric ID).
Concretely, we add an additional constructor `$\del{Share}\ \tID\ \tDelta$' to $\tDelta$:\footnote{In \cite[\S8.2]{2023-smeding-dualrev}, the equivalent of $\tDelta$ is called $\mathrm{Contrib}$, which can be seen as a combination of $\del{Zero}$, $\del{Add}$, $\del{Scale}$ and $\del{Share}$.}
\[ \begin{array}{@{}l@{\;}c@{\;}l@{}}
    \mathbf{data}\ \tDelta
    &=& \textcolor{gray}{
            \del{Zero}
            \mid \del{Input}\ \tDVarName
            \mid \del{Add}\ \tDelta\ \tDelta
            \mid \del{Scale}\ \R\ \tDelta
        } \\
    &\mid& \del{Share}\ \tID\ \tDelta \\[0.3em]
    \multicolumn{3}{@{}l@{}}{\mathbf{type}\ \tID = \Int \quad\ccomment{for semantic clarity}}
\end{array} \]
and we give unique IDs to all potentially shareable $\tDelta$ (sub)terms by wrapping them in a $\del{Share}$ constructor with that unique ID.
We have two invariants on those IDs:

\vspace{0.4em}
\textbf{Invariant 1}.
In a $\tDelta$ term `$\del{Share}\ i\ d$', all $\tID$s appearing inside $d$ are strictly smaller than $i$.

\textbf{Invariant 2}.
Given two $\tDelta$ terms `$\del{Share}\ i\ d_1$' and `$\del{Share}\ j\ d_2$', if $i = j$ then $d_1$ and $d_2$ live at the same address in memory.
\vspace{0.4em}

In practice, the converse of invariant 2 is also true, but for soundness we need only the invariant as stated.

Together, these invariants ensure that the sharing structure is soundly represented and acyclic (i.e.,\ the $\tDelta$ term represents a directed acyclic graph (DAG)), and furthermore that it is helpful for ensuring that $\eval$ can avoid evaluating any part of the $\tDelta$ term more than once.
The trick is that the new $\eval$ will backpropagate through $\tDelta$ subterms in strict decreasing order of ID.
Whenever it encounters a $\del{Share}$ node, $\eval$ will save the $r$ value to be backpropagated into that part of the $\tDelta$ term; if the same node is encountered multiple times, the $r$ values are added.
Backpropagation into a node is resumed only when it is next-in-line: its ID is the highest among the $\tDelta$ subterms still to process.
See the discussion of the sharing-aware evaluator in \cref{sec:ad-naive-evaluator} for how this is implemented.

\paragraph{Lifting to monadic code}
To be able to generate these unique IDs, we lift the right-hand side of $D[-]$ to monadic code in a state monad with a single $\Int$ as state.

\begin{figure}
\begin{gather*}
    x_1 : \tau_1, \ldots, x_n : \tau_n \vdash t : \tau \quad \leadsto \quad x_1 : D[\tau_1], \ldots, x_n : D[\tau_n] \vdash D[t] : \textcolor{red}{\tIdGen}\ D[\tau] \\
    \begin{array}{@{}l@{}}
        \ccomment{Updated type transformation:} \\[0.3em]
        \begin{array}{@{}l@{\quad}l@{}}
            \textcolor{gray}{D[\R] = (\R, \tDelta)}
                & \textcolor{gray}{D[(\sigma, \tau)] = (D[\sigma], D[\tau])}\\
            \textcolor{gray}{D[\Int] = \Int}
                & D[\sigma \to \tau] = D[\sigma] \to \textcolor{red}{\tIdGen}\ D[\tau]
            \vspace{0.3em}
        \end{array} \\
        \ccomment{The ID generation monad:} \\[0.3em]
        \begin{array}{@{}l@{}}
            \mathbf{newtype}\ \tIdGen\ a = \tIdGen\ (\mathrm{State}\ \Int\ a) \\
            \mgenid :: \tIdGen\ \tID \qquad\ccomment{Recall that $\tID = \Int$.}
            \vspace{0.3em}
        \end{array} \\
        \ccomment{Primitive operations on real numbers:} \\[0.3em]
        \begin{array}{@{}r@{\;}c@{\;}l@{}}
            D[r]
                &=& \mathbf{return}\ (r, \del{Zero}) \\
            D[\sin t]
                &=& \mathbf{do}\ \multilineT{
                    (x, d) \leftarrow D[t] \\
                    \textcolor{red}{\var{id} \leftarrow \mgenid} \\
                    \mathbf{return}\ (\sin x, \textcolor{red}{\del{Share}\ \var{id}\ (}\del{Scale}\ (\cos x)\ d\textcolor{red}{)})
                } \\
            D[t_1 +_{\R} t_2]
                &=& \mathbf{do}\ \multilineT{
                    (x_1, d_1) \leftarrow D[t_1]; (x_2, d_2) \leftarrow D[t_2] \\
                    \textcolor{red}{\var{id} \leftarrow \mgenid} \\
                    \mathbf{return}\ (x_1 +_{\R} x_2, \textcolor{red}{\del{Share}\ \var{id}\ (}\del{Add}\ d_1\ d_2\textcolor{red}{)})
                } \\
            D[t_1 \times_{\R} t_2]
                &=& \mathbf{do}\ \multilineT{
                    (x_1, d_1) \leftarrow D[t_1]; (x_2, d_2) \leftarrow D[t_2] \\
                    \textcolor{red}{\var{id} \leftarrow \mgenid} \\
                    \mathbf{return}\ (x_1 \times_{\R} x_2, \textcolor{red}{\del{Share}\ \var{id}\ (}\del{Add}\ (\del{Scale}\ x_2\ d_1)\ (\del{Scale}\ x_1\ d_2)\textcolor{red}{)})
                } \\
            \textit{etc.\rlap{\ the other primitive operations on $\R$}}
            \vspace{0.3em}
        \end{array} \\
        \ccomment{Other operations are monadically lifted:} \\[0.3em]
        \begin{array}{@{}r@{\;}c@{\;}l@{}}
            D[k] &=& \mathbf{return}\ k \\
            D[x] &=& \mathbf{return}\ x \\
            D[\mathbf{let}\ x = s\ \mathbf{in}\ t] &=& \mathbf{do}\ x \leftarrow D[s]; D[t] \\
            D[(s, t)] &=& \mathbf{do}\ x \leftarrow D[s]; y \leftarrow D[t]; \mathbf{return}\ (x, y) \\
            D[\name{fst}\ t] &=& \mathbf{do}\ x \leftarrow D[t]; \mathbf{return}\ (\name{fst}\ x) \\
            D[\name{snd}\ t] &=& \mathbf{do}\ x \leftarrow D[t]; \mathbf{return}\ (\name{snd}\ x) \\
            D[\lambda x.\ t] &=& \mathbf{return}\ (\lambda x.\ D[t]) \\
            D[s\ t] &=& \mathbf{do}\ f \leftarrow D[s]; x \leftarrow D[t]; f\ x \\
            D[\mathbf{if}\ t_1\ \mathbf{then}\ t_2\ \mathbf{else}\ t_3] &=& \mathbf{do}\ x \leftarrow D[t_1]; \mathbf{if}\ x\ \mathbf{then}\ D[t_2]\ \mathbf{else}\ D[t_3] \\
            D[s \times_{\Int} t] &=& \mathbf{do}\ x \leftarrow D[s]; y \leftarrow D[t]; \mathbf{return}\ (x \times_{\Int} y) \\
            \textit{etc.\rlap{\ the other primitive operations on $\Int$ and $\Bool$}}
        \end{array}
    \end{array}
\end{gather*}
\caption{\label{fig:dn-ad-terms}
    Updated rules for the dual-numbers reverse AD code transformation to properly handle sharing.
    Compare \cref{fig:naive-dn-ad-terms}; the added text (apart from the lifting to monadic code) is highlighted in \textcolor{red}{red}.
}
\end{figure}

The updated code transformation is given in \cref{fig:dn-ad-terms}.
Note that for all language constructs except primitive operations, this monadic lifting is done very systematically, as functional programmers (especially in Haskell) are well used to.
The only wrinkle is that because our language so far supports user-written functions, and the bodies of those functions get differentiated too, those differentiated functions also become effectful.
This results in the updated rule $D[\sigma \to \tau] = D[\sigma] \to \tIdGen\ D[\tau]$ in \cref{fig:dn-ad-terms}, as well as the fact that in $D[s\ t]$, there is no `$\mathbf{return}$' around the result of the call $f\ x$.

For primitive operations, we use `$\mgenid$', the (only) monad method, to generate unique, monotonically increasing $\tID$ values --- the monotonicity allows us to preserve invariant 1.
Note that we only give an ID to the full $\tDelta$ value returned by the code for each primitive operation; for example, there is no $\del{Share}$ node wrapping the $\del{Scale}$ constructors in the $\tDelta$ value for $(\times_\R)$.
We can leave these out because these $\del{Scale}$ nodes can never be shared: they are used exactly once, namely in the containing $\del{Add}$, which itself \emph{does} have an ID.
Furthermore, we elide the $\del{Share}$ node around $\del{Zero}$ in the right-hand side of $D[r]$, because while that $\del{Zero}$ may well be used multiple times, $\eval_3$ of $\del{Zero}$ is very cheap and not worth deduplicating.

This scheme does generally result in ``too many'' IDs: contrary to the somewhat contrived $P_1$ from the beginning of this subsection, in practice far from all $\tDelta$ terms that we give a unique ID to will actually be shared.
But we certainly have \emph{enough} $\del{Share}$ nodes: every $\tDelta$ node that could benefit from being evaluated only once, gets a unique ID.

This way of recording sharing inside $\tDelta$ terms using IDs is quite different from the usual way of notating shared subterms in an expression language: let bindings.\footnote{
    \citeN{2022-krawiec-dualrev} instead choose to make $\tDelta$ a proper, traditional term language with let bindings; this requires a more complicated monad.
    In short, the monad becomes additionally a writer monad of $(\tID, \tDelta)$ pairs, and $\tDelta$ is augmented with two constructors: $\del{Var}\ \tID$ and $\del{Let}\ \tID\ \tDelta$.
    In $D[-]$ for primitive operations, instead of simply naming the returned $\tDelta$ term $d$, the pair $(\var{id}, d)$ is emitted in the writer monad and we return only $\del{Var}\ \var{id}$.
    Before evaluation, the writer log is stacked in chronological order (which is also increasing $\tID$ order) and set as a stack of $\del{Let}$ bindings on top of the final $\tDelta$ term from the program.
    We use our global sharing approach instead because we will need it in \cref{sec:restaging}.
}
In contrast with let bindings, where the shared term is available only in a limited scope (the `in' part of the let binding), with our $\del{Share}$-based approach, the shared term is ``available'' everywhere apart from inside the shared term itself.
For this reason, we call this $\del{Share}$-based approach \emph{global sharing}; this idea will be used again in \cref{sec:restaging}.

\subsection{Evaluator}\label{sec:ad-naive-evaluator}

\begin{figure}
\begin{subfigure}{0.24\textwidth}
    \centering
    \begin{tikzpicture}[xscale=0.75, yscale=0.75]
        \node (x0) at (1, -1.5) {$x$};
        \node (x1) at (1, -0.25) {$\exp$};
        \node (x3) at (2, 1.25) {$\log$};
        \node (x5) at (3, 2.75) {$\cos$};
        \node (x7) at (1, 2.75) {$\sin$};
        \node (x11) at (2, 4.5) {$/$};
        \node (x13) at (1, 6.25) {$+$};
        \draw[->] (x0) -- (x1);
        \draw[->] (x1) -- (x3);
        \draw[->] (x3) -- (x5);
        \draw[->] (x3) -- (x7);
        \draw[->] (x5) -- (x11);
        \draw[->] (x7) -- (x11);
        \draw[->] (x11) -- (x13);
        \draw[->] (x1) to[bend left=25] (x13);
    \end{tikzpicture}
    \caption{\label{subfig:eval-share-visits-source}
        Source program
    }
\end{subfigure}
\begin{subfigure}{0.24\textwidth}
    \centering
    \newcommand\nSh{$\del{Share}$}
    \newcommand\nSc{$\vphantom{t}\smash{\del{Scale}}$}
    \newcommand\nAd{$\vphantom{t}\smash{\del{Add}}$}
    \begin{tikzpicture}[xscale=0.75, yscale=0.75]
        \tikzset{every node/.style={inner sep=1pt,font=\small}}
        \node (d0) at (1, -1.5) {$\del{Input}$};
        \node (d1p) at (1, -0.5) {$\nSc$};
        \node (d1) at (1, 0) {$\nSh\ 1$};
        \node (d2) at (2, 1) {\nSc};
        \node (d3) at (2, 1.5) {\nSh\ 2};
        \node (d4) at (3, 2.5) {\nSc};
        \node (d5) at (3, 3) {\nSh\ 4};
        \node (d6) at (1, 2.5) {\nSc};
        \node (d7) at (1, 3) {\nSh\ 3};
        \node (d8) at (2.5, 4) {\nSc};
        \node (d9) at (1.5, 4) {\nSc};
        \node (d10) at (2, 4.5) {\nAd};
        \node (d11) at (2, 5) {\nSh\ 5};
        \node (d12) at (1, 6) {\nAd};
        \node (d13) at (1, 6.5) {\nSh\ 6};
        \draw[->] (d1p) -- (d0);
        \draw[->] (d1) -- (d1p);
        \draw[->] (d2) -- (d1);
        \draw[->] (d3) -- (d2);
        \draw[->] (d4) -- (d3);
        \draw[->] (d5) -- (d4);
        \draw[->] (d6) -- (d3);
        \draw[->] (d7) -- (d6);
        \draw[->] (d8) -- (d5);
        \draw[->] (d9) -- (d7);
        \draw[->] ($(d10.south east) - (0.15, 0)$) -- (d8);
        \draw[->] ($(d10.south west) + (0.15, 0)$) -- (d9);
        \draw[->] (d11) -- (d10);
        \draw[->] (d12) to[bend right=25] (d1);
        \draw[->] (d12) -- (d11);
        \draw[->] (d13) -- (d12);
    \end{tikzpicture}
    \caption{\label{subfig:eval-share-visits-delta}
        $\tDelta$ term
    }
\end{subfigure}
\begin{subfigure}{0.24\textwidth}
    \centering
    \begin{tikzpicture}[xscale=0.75, yscale=0.75]
        \node (x0) at (1, -1.5) {3, 8, 12};
        \node (x1) at (1, -0.25) {2, 7, 11};
        \node (x3) at (2, 1.25) {6, 10};
        \node (x5) at (3, 2.75) {9};
        \node (x7) at (1, 2.75) {5};
        \node (x11) at (2, 4.5) {4};
        \node (x13) at (1, 6.25) {1};
        \draw[color=gray!50!white] (x0) -- (x1);
        \draw[color=gray!50!white] (x1) -- (x3);
        \draw[color=gray!50!white] (x3) -- (x5);
        \draw[color=gray!50!white] (x3) -- (x7);
        \draw[color=gray!50!white] (x5) -- (x11);
        \draw[color=gray!50!white] (x7) -- (x11);
        \draw[color=gray!50!white] (x11) -- (x13);
        \draw[color=gray!50!white] (x1) to[bend left=25] (x13);
    \end{tikzpicture}
    \caption{\label{subfig:eval-share-visits-order-naive}
        Naive visit order
    }
\end{subfigure}
\begin{subfigure}{0.24\textwidth}
    \centering
    \begin{tikzpicture}[xscale=0.75, yscale=0.75]
        \node (x0) at (1, -1.5) {7};
        \node (x1) at (1, -0.25) {6};
        \node (x3) at (2, 1.25) {5};
        \node (x5) at (3, 2.75) {3};
        \node (x7) at (1, 2.75) {4};
        \node (x11) at (2, 4.5) {2};
        \node (x13) at (1, 6.25) {1};
        \draw[color=gray!50!white] (x0) -- (x1);
        \draw[color=gray!50!white] (x1) -- (x3);
        \draw[color=gray!50!white] (x3) -- (x5);
        \draw[color=gray!50!white] (x3) -- (x7);
        \draw[color=gray!50!white] (x5) -- (x11);
        \draw[color=gray!50!white] (x7) -- (x11);
        \draw[color=gray!50!white] (x11) -- (x13);
        \draw[color=gray!50!white] (x1) to[bend left=25] (x13);
    \end{tikzpicture}
    \caption{\label{subfig:eval-share-visits-order-aware}
        Sharing-aware order
    }
\end{subfigure}
\caption{\label{fig:eval-share-visits}
    $\eval_3$ of \cref{sec:ad-naive-reverse-mode} visits $\tDelta$ nodes as often as there are paths to them; the sharing-aware evaluator of \cref{sec:ad-naive-evaluator}, $\reversePass_4$, visits $\del{Share}$-wrapped nodes only once.
    From left to right: the source expression $\mathbf{let}\ y = \exp x\ \mathbf{in}\ \mathbf{let}\ z = \log y\ \mathbf{in}\ y + \frac{\sin z}{\cos z}$ with input $x$ (arrows in the direction of execution); the resulting $\tDelta$ term (arrows pointing to subterms; IDs assume the $\sin$ was executed before the $\cos$); the order in which $\eval_3$ visits the nodes of \cref{subfig:eval-share-visits-delta} if $\del{Share}$ nodes were removed; the order in which $\reversePass_4$ visits the nodes of \cref{subfig:eval-share-visits-delta}.
    The gray lines in \cref{subfig:eval-share-visits-order-naive,subfig:eval-share-visits-order-aware} are only to visualise the relation to \cref{subfig:eval-share-visits-source,subfig:eval-share-visits-delta}.
}
\end{figure}
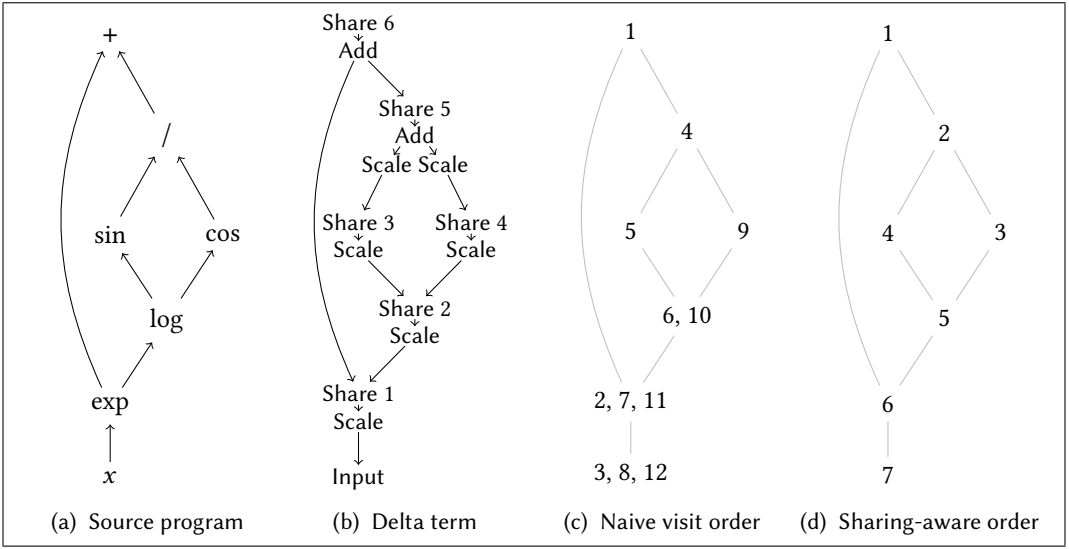

The evaluator presented in this section improves over $\eval_3$ from \cref{sec:ad-naive-reverse-mode} by using the IDs in $\del{Share}$ nodes to visit every $\tDelta$ node only once (excepting $\del{Zero}$ and $\del{Input}$, which take $O(1)$ time to evaluate anyway).
The resulting evaluation order is shown on a simple example program in \cref{fig:eval-share-visits}.
Note that $\eval_3$ would have visited the `$\exp$' and $x$ nodes in \cref{subfig:eval-share-visits-source} three times, whereas the new evaluator visits them only once.

The improved evaluator does not simply traverse the $\tDelta$ term depth-first, but instead in a mix of breadth-first and depth-first traversal.
Accordingly, the evaluator is restructured into multiple functions:
\begin{enumerate}
\item
    $\reversePass_4$, a (non-recursive) wrapping function that initialises the backpropagation process and calls $\backprop_4$.
\item
    $\backprop_4$, which loops over the found $\tID$s from highest to lowest, evaluating the $\tDelta$ fragment for each in turn using $\eval_4$.
\item
    $\eval_4$, which interprets a $\tDelta$ fragment under a $\del{Share}$ constructor, stopping at any contained $\del{Share}$ constructors and deferring their recursive traversal until $\backprop_4$ decides that their time has come.
\end{enumerate}

\begin{figure}
\begin{align*}
&\begin{array}{@{}l@{}}
    \mathbf{data}\ \EState = \EState \quad \ccomment{evaluation state} \\
    \quad \{ \begin{array}[t]{@{}l@{\quad}l@{}}
        \ \name{grad} :: \Map\ \tDVarName\ \R & \ccomment{input cotangents: will collect final gradient} \\
        \leadcomma\ \name{dfrag} :: \Map\ \tID\ \tDelta & \ccomment{delta fragments} \\
        \leadcomma\ \name{accum} :: \Map\ \tID\ \R\ \} & \ccomment{accumulated node cotangents} \\
    \end{array}
\end{array} \\[0.5em]
&\begin{array}{@{}l@{}}
    \reversePass_4 :: \R \to \tDelta \to \Map\ \tDVarName\ \R \\
    \begin{array}{@{}l@{\ }c@{\ }l@{}}
        \reversePass_4\ c\ d &=& \name{grad}\ (\backprop_4\ (\eval_4\ c\ d\ (\EState\ \{\}\ \{\}\ \{\}))) \\
    \end{array}
\end{array} \\[0.5em]
&\begin{array}{@{}l@{}}
    \backprop_4 :: \EState \to \EState \\
    \begin{array}{@{}l@{\ }c@{\ }l@{}}
        \backprop_4\ s &=&
            \multilineT{
                \mathbf{case}\ \name{Map.maxViewWithKey}\ (\name{accum}\ s)\ \mathbf{of} \\
                \quad \multilineT{
                    \name{Just}\ ((i, c), \mathit{acc}') \to \\
                    \quad \multilineT{
                        \mathbf{let}\ \multilineT{
                            d = \name{dfrag}\ s \mathbin{\name{Map.!}} i \\
                            s' = \eval_4\ c\ d\ (s\ \{\ \name{accum} = acc', \name{dfrag} = \name{Map.delete}\ i\ (\name{dfrag}\ s)\ \})
                        } \\
                        \mathbf{in}\ \backprop_4\ s' \\
                    } \\
                    \name{Nothing} \to s \\
                }
            }
    \end{array}
\end{array} \\[0.5em]
&\begin{array}{@{}l@{}}
    \eval_4 :: \R \to \tDelta \to \EState \to \EState \\
    \begin{array}{@{}l@{\ }c@{\ }l@{}}
        \eval_4\ c\ \del{Zero} &=& \id \\
        \eval_4\ c\ (\del{Input}\ v) &=& \clam{s} s\ \{\ \name{grad} = \name{Map.insertWith}\ (+)\ v\ c\ (\name{grad}\ s)\ \} \\
        \eval_4\ c\ (\del{Add}\ d_1\ d_2) &=& \eval_4\ c\ d_2 \circ \eval_4\ c\ d_1 \\
        \eval_4\ c\ (\del{Scale}\ r\ d) &=& \eval_4\ (c \cdot r)\ d \\
        \eval_4\ c\ (\del{Share}\ i\ d) &=& \clam{s} s\ \{ \multilineT{
            \ \name{dfrag} = \name{Map.insert}\ i\ d\ (\name{dfrag}\ s) \\
            \leadcomma\ \name{accum} = \name{Map.insertWith}\ (+)\ i\ c\ (\name{accum}\ s)\ \} \\
        } \\
    \end{array}
\end{array}
\end{align*}
\caption{\label{fig:ad-efficient-eval}
    The $\tDelta$ evaluator that handles internal sharing in $\tDelta$ terms and has the right time complexity, apart from logarithmic factors due to use of $\Map$.
}
\end{figure}

The result is shown in \cref{fig:ad-efficient-eval}.
To be able to delay recursing below $\del{Share}$ nodes, we need more storage in our evaluation state: in addition to the sparse gradient accumulator in `$\name{grad}$' (which was implicitly a mutable $\R^n$ before in $\eval_3$, but which we make explicitly sparse now using $\Map\ \tDVarName\ \R$), we need storage to save the $\tDelta$ trees that we still need to visit, as well as their accumulated cotangents.
These are stored, respectively, in `$\name{dfrag}$' and `$\name{accum}$' in the evaluation state $\EState$.
The main backpropagation loop is $\backprop_4$, which repeatedly chooses the largest encountered-but-yet-unvisited ID, takes it out of the `$\name{dfrag}$' and `$\name{accum}$' maps, runs $\eval_4$ on this $\tDelta$ fragment, and then continues with the next unvisited ID.
Evaluating a $\tDelta$ fragment proceeds exactly as before in $\eval_3$ (\cref{sec:ad-naive-reverse-mode}), except that on reaching a $\del{Share}$ constructor, evaluation does not recurse but instead saves the contained $\tDelta$ tree, as well as the cotangent $c$ to be backpropagated into that tree, in the evaluation state.
The cotangent is added to the value already in the state, if any.
(This is the add operation in the reverse derivative that comes from sharing in the source program, according to the mantra: ``sharing in the primal becomes addition in the dual''.)
When there are no other unvisited subtrees with higher IDs, $\backprop_4$ will take these values out of the state and backpropagate the summed cotangent contributions down into the saved $\tDelta$ tree by invoking $\eval_4$ again.

At the top level, $\reversePass_4$ initialises the process by starting backpropagation on a state produced by evaluating the topmost fragment of the full $\tDelta$ term resulting from the forward pass.
The final gradient is simply the `$\name{grad}$' field of the state after backpropagation is complete.

In the remainder of the paper, plain use of `$\reversePass$', `$\backprop$' or `$\eval$' will refer to their 4th version in \cref{fig:ad-efficient-eval}.

\subsection{Wrapper}\label{sec:ad-naive-wrapper}

As in \cite{2022-krawiec-dualrev,2023-smeding-dualrev}, we add a wrapper around the algorithm to give it a useful API, and to make explicit what that API \emph{is}, precisely.

\begin{figure}
\begin{gather*}
    \begin{array}{@{}l@{}}
        \name{wrapper}\ \multilineT{
            :: \tAST\ \R \quad\ccomment{with one free variable: $x : \tyin$} \\
            \to \tyin \to \R \to \tyin \\
        } \\
        \name{wrapper}\ t\ \var{inp}\ \var{ctg} = {} \\
        \quad \multilineT{
            \mathbf{let}\ \multilineT{
                \var{inp'} = \name{named}_{\tyin}\ \var{inp} \\
                (\_, d) = \name{runIdGen}\ (\mathbf{let}\ x = \var{inp'}\ \mathbf{in}\ D[t])\ 0 \quad\ccomment{Provide $D[t]$'s free variable $x$} \\
                \var{grad} = \reversePass_4\ \var{ctg}\ d
            } \\
            \mathbf{in}\ \name{reconstruct}_{\tyin}\ \var{grad}\ \var{inp'}
        }
    \end{array} \\
    \begin{array}{@{}l@{}}
        \ccomment{Using the following functions:} \\
        \name{runIdGen} :: \tIdGen\ a \to \Int \to a \\
        \name{named}_{\tyin} :: \tyin \to D[\tyin] \quad\ccomment{for first-order $\tyin$} \\
        \qquad\ccomment{e.g.\ $\multilineT{
            \name{named}_{((\R,\R),\R)}\ ((7,3.1),8) = {} \\
            \qquad (((7, \del{Input}\ 0), (3.1, \del{Input}\ 1)), (8, \del{Input}\ 2))
        }$} \\
        \name{reconstruct}_{\tyin} :: \Map\ \tDVarName\ \R \to D[\tyin] \to \tyin \quad\ccomment{for first-order $\tyin$} \\
        \qquad\ccomment{e.g.\ $\multilineT{
            \name{reconstruct}_{((\R,\R),\R)}\ \var{grad}\ (((7.0, \del{Input}\ 0), (3.1, \del{Input}\ 1)), (8.0, \del{Input}\ 2)) = {} \\
            \qquad ((\var{grad} \mathbin{!} 0, \var{grad} \mathbin{!} 1), \var{grad} \mathbin{!} 2)
        }$}
    \end{array}
\end{gather*}
\caption{\label{fig:dn-ad-wrapper}
    The wrapper for the dual-numbers reverse AD transformation of \cref{fig:dn-ad-terms}.
    The implementations of `$\name{named}$' and `$\name{reconstruct}$' are straightforward but somewhat verbose.
    Technically `$\name{reconstruct}$' does not need the $D[\tyin]$ argument for our language, but we add it for generality, because it would have been necessary had we included e.g.\ coproducts (sum types) in our type system.
}
\end{figure}

An implementation sketch is given in \cref{fig:dn-ad-wrapper}.
Our first input is a term $t$ satisfying the typing $x : \tyin \vdash t : \R$ for some first-order $\tyin$.
(Generalisation of the output type $\R$ would require running the reverse pass once for each output scalar; \cref{app:output-type-generalisation} shows how to do this for the final algorithm of \cref{sec:restaging}.)
We transform this term to $x : D[\tyin] \vdash D[t] : (\R, \tDelta)$, which we can run once we have a $D[\tyin]$.
To obtain such a $D[\tyin]$, we take a $\tyin$-typed input ($\var{inp}$, the point to differentiate at) and pair up each scalar in that structure with an $\del{Input}$ node with a unique $\tDVarName$.
This yields $\var{inp'}$ of type $D[\tyin]$ (recall \cref{subfig:naive-dn-ad-types}).
Thus we can now evaluate the transformed term at the processed input point to obtain the function result ($\R$) and its forward derivative ($\tDelta$).
The function result could be returned as well, but is ignored in \cref{fig:dn-ad-wrapper}.

Then we call $\reversePass_4$ (\cref{fig:ad-efficient-eval}) with the initial $\R$ cotangent (which we get from the user, typically `1') to obtain the gradient in the form of a `$\Map\ \tDVarName\ \R$'.
Finally, replacing all scalar-$\del{Input}$ pairs in $\var{inp'}$ by their corresponding scalar from the gradient map ($\var{grad}$), we obtain the actual gradient of type $\tyin$, which we return.

It is worth noting that to differentiate the same program at multiple input points, almost everything needs to be repeated \emph{for each such point}: only the transformed term $D[t]$ can be cached.
There is never a good opportunity to optimise the calculations of the reverse pass, because the output of the $\tDelta$ interpretation process (in $\reversePass_4$) is already a \emph{gradient}, not a gradient-computing \emph{program} that we can still compile and optimise.
We address this problem as well (in \cref{sec:restaging}), after we have improved support for arrays.

\paragraph{Interpretation of derivatives on discrete types}
In the world of mathematics, if we are given a function
\[ f : \Z^{k_1} \times \R^{n_1} \to \Z^{k_2} \times \R^{n_2} \]
and are asked for its Jacobian (the total derivative), we seem to have to define what it means to take a partial derivative of an integer value with respect to another.
This makes little sense: taking the classical limit definition of a derivative at its word (interpreted in the discrete topology of the integers), such a ``partial derivative'' would be identically 0, so we may as well not report it at all.

Because of this, our transformation joins $\tDelta$ terms only to scalar values, not to discrete values in the input and output of the function to differentiate.
Mathematically, this corresponds to not computing $\mathrm J f$, but instead ignoring the discrete outputs of $f$ and reinterpreting its discrete inputs as global constants:
\[ \tilde f : \R^{n_1} \to \R^{n_2} \]
and then computing $\mathrm J \tilde f : \R^{n_1} \to \R^{n_2 \times n_1}$, a function that produces the Jacobian matrix of $\tilde f$ at any given input point.

Our choice of not computing derivatives for integral values means that the wrapper in \cref{fig:dn-ad-wrapper} has nothing sensible to return for integers in $\tyin$; `$\name{reconstruct}$' has to choose something, and we leave unspecified what it chooses.
(Reasonable options include ``zero'' and ``the input''.)
To reflect the undefinedness of a derivative with respect to an integer value, we could give `$\name{wrapper}$' a more precise type:
\[ \name{wrapper} :: \tAST\ \R \to \tyin \to \R \to T[\tyin] \]
where $T$ is a type function that maps a type to its type of tangents:
\[ T[\R] = \R \qquad T[\Int] = () \qquad T[(\sigma, \tau)] = (T[\sigma], T[\tau]) \]
With this typing, `$\name{wrapper}$' does not need to return nonsensical values.
For convenience, however, we let it return the full $\tyin$ here.

In the rest of this paper, we continue to think about derivatives and Jacobians as if this $T$ is implicitly applied, both in the input and in the output (if the part of the transformation in question supports non-trivial output types).
For example, we say that the function $\name{round} : \R \to \Int$ has trivial derivative (it contains no information), because its Jacobian, being an element of $\R^{0 \times 1}$, is empty.

\section{A language with arrays: the core language}\label{sec:array-language}

The reverse AD algorithm set out in \cref{sec:ad-naive} works fine on functions that manipulate scalars, but real programs that AD is used on typically work with \emph{arrays} of scalars.
This includes machine learning applications such as neural networks, but also probabilistic programming on larger statistical models, most optimisation applications, etc.
Thus, to really count as an AD system, we ought to support arrays.
In what way?

\subsection{Higher-order and first-order array languages}

Abstracting over the precise syntax, there are two prevailing ways to design array programming languages today: higher-order and first-order.
In the former (e.g.\ Futhark~\cite{2017-henriksen-futhark}, XLA\footnote{\url{https://github.com/openxla/xla}}, Accelerate~\cite{acc-2011-cuda}, Dex~\cite{2021-paszke-dex}), there are higher-order array operations (often called \emph{second-order array operations}) with type signatures roughly like the following:
\[ \begin{array}{@{}l@{}}
    \name{build}_n :: \shvec{\tightunderbrace{\Int, \ldots, \Int}{n}} \to (\ixvec{\tightunderbrace{\Int, \ldots, \Int}{n}} \to \tau) \to \Array\ n\ \tau \\
    \qquad \ccomment{Allocate a new $n$-dimensional array of the given shape and fill it with values} \\
    \qquad \ccomment{returned by the given function, called once for each index in the array.} \\
    \name{map} :: (\sigma \to \tau) \to \Array\ n\ \sigma \to \Array\ n\ \tau \\
    \qquad \ccomment{Return a new array of the same size as the input, with every element modified by} \\
    \qquad \ccomment{the given function.} \\
    \name{foldInner} :: (\tau \to \tau \to \tau) \to \tau \to \Array\ n\ \tau \to \Array\ (n - 1)\ \tau \\
    \qquad \ccomment{Inner-dimension reduction: along each inner-dimension vector, reduce with the given} \\
    \qquad \ccomment{function and the given initial value. The function must typically be associative to} \\
    \qquad \ccomment{allow parallel execution.}
\end{array} \]
The `$\Array\ n\ \tau$' type denotes an $n$-dimensional array containing values of type $\tau$.
Typically, many more operations are included as well, such as scans, histograms, stencils, etc.\footnote{In fact, given a (first-order) primitive `$\name{iota} :: \Int \to \Array\ 1\ \Int$' that given $n$ returns the array $\arrlit{0, \ldots, n - 1}$, `$\name{build}$' and `$\name{map}$' are interdefinable, so only one of the two need be primitive.}
These languages are usually not \emph{actually} fully higher-order: even if these types look higher-order, one might not be allowed to create arrays of functions, nor arrays of arrays; sometimes (e.g.\ in Accelerate) even using an array combinator inside a function passed to another array combinator is disallowed.

In the second prevailing style (e.g.\ APL, MatLab, NumPy, TensorFlow), there are first-order operations only, for example:
\[ \begin{array}{@{}ll@{}}
    (+) :: \Array\ n\ \R \to \Array\ n\ \R \to \Array\ n\ \R
        & \ccomment{Add elementwise} \\
    (\times) :: \Array\ n\ \R \to \Array\ n\ \R \to \Array\ n\ \R
        & \ccomment{Multiply elementwise} \\
    \name{sumInner} :: \Array\ n\ \R \to \Array\ (n - 1)\ \R
        & \ccomment{Sum along the inner-dimension vectors} \\
    \name{replicate}_n :: \Int \to \Array\ n\ \tau \to \Array\ (n + 1)\ \tau
        & \ccomment{Add one outer dimension that contains} \\
        & \ccomment{\ \ the input array at every index} \\
    \name{transpose} :: n \geq 2 \Rightarrow \Array\ n\ \tau \to \Array\ n\ \tau & \ccomment{Transpose the outermost 2 dimensions}
\end{array} \]
Of course, arithmetic operations might be overloaded beyond just $\R$, and many more operations are typically available.

The higher-order design is strictly more expressive: the first-order combinators can be defined in terms of the higher-order ones, but not the other way round.
Furthermore, providing these higher-order combinators to the programmer allows them to write more understandable code; this has been eloquently argued by~\citeN{2021-paszke-dex} in their design philosophy of the Dex language.
For example, if one knows that the following code is possible for naive matrix multiplication:\footnote{The `$\name{shape}$' function returns the list of sizes of the dimensions of the array, outermost first.}
\[ \begin{array}{@{}l@{}}
\name{matmat} :: \Array\ 2\ \R \to \Array\ 2\ \R \to \Array\ 2\ \R  \\
\name{matmat}\ a\ b = \\
\quad \keyw{let}\ \begin{array}[t]{@{}l@{}}
        \shvec{k, m} = \name{shape}\ a \\
        \shvec{\_, \mathrlap{n}\hphantom{m}} = \name{shape}\ b
      \end{array} \\
\quad \keyw{in}\ \multilineT{
        \name{build}_2\ \shvec{k, n}\ (\lambda \ixvec{i, j}. \\
        \quad \name{sum}\ (\name{build}_1\ m\ (\lambda p.\ a \mathbin{!} \ixvec{i, p} * b \mathbin{!} \ixvec{p, j}))))
      }
\end{array} \]
then one would certainly not be satisfied with having to write something like the following:
\[ \begin{array}{@{}l@{}}
\name{matmatFirstOrder} :: \Array\ 2\ \R \to \Array\ 2\ \R \to \Array\ 2\ \R  \\
\name{matmatFirstOrder}\ a\ b = \\
\quad \keyw{let}\ \begin{array}[t]{@{}l@{}}
        \shvec{k, m} = \name{shape}\ a \\
        \shvec{\_, \mathrlap{n}\hphantom{m}} = \name{shape}\ b
      \end{array} \\
\quad \keyw{in}\ \name{sumInner}\ (\name{transpose}\ (\name{replicate}\ n\ a) * \name{replicate}\ k\ (\name{transpose}\ b))
\end{array} \]
Both versions of `$\name{matmat}$' do the same thing, but the first is clearly easier to understand and to get correct.

\subsection{Our core language}\label{sec:core-language}

Evidently, we want our array language (the ``core'' language --- used for both input and output of the differentiation algorithm) to be in higher-order style as much as possible.\footnote{On the other hand, it turns out that because first-order bulk operations are much better for efficient reverse AD using the dual-numbers framework, we end up converting `$\name{build}$'-code into bulk operations in \cref{sec:bot}!}
Unfortunately, it turns out that a generic reduction operation, named `$\name{foldInner}$' above, is very difficult to support in a dual-numbers reverse AD framework if one desires fast gradient code, as we do.
Thus, while we do have `$\name{build}$' and all the other operations derivable from it (such as `$\name{map}$'), we have to make do with specific reductions like `$\name{sum}$' and `$\name{maximum}$'.\footnote{
    Our implementation (\cref{sec:implementation}) does experimentally support a general fold operation, restricted to closed combination functions.
}
For conciseness, we include only `$\name{sum}$' explicitly in the language, because adding support for other specific reductions is simple and requires little more than writing down their derivative.

\begin{figure}
\begin{equation*}
\begin{array}{@{}r@{\ }r@{\ }l@{\ \ \ }l@{}}
s, t, u, v & \coloneqqq & c & \text{(constant literal tensors, e.g.\ \arrlit{\arrlit{3,1.2,17},\arrlit{-5,0.4,1}})} \\
& \mid & x \mid \cletin{x = u}{v} & \text{(variables and binding)} \\
& \mid & \cc{cond}\ t\ u\ v & \text{(strict conditionals)} \\
& \mid & \mathit{op}\ u\ v \mid \mathit{op}\ t & \text{(broadcasted (elementwise) binary and unary ops.)} \\
& \mid & \cc{index}\ t\ ix  & \text{(index at a multidimensional position)}\\
& \mid & \cc{sumOuter}\ t & \text{(reduce along the outermost dimension)} \\
& \mid & \cc{gather}\ sh\ t\ (\clam {\var{is}} ix) & \text{(backward permutation / batched array indexing; \cref{sec:core-gather-scatter})} \\
& \mid & \cc{scatter}\ sh\ t\ (\clam {\var{is}} ix) & \text{(forward permutation onto zeros; see \cref{sec:core-gather-scatter})} \\
& \mid & \ravel{t_1,\ldots,t_n} & \text{(combine equal-shaped arrays into one with 1 more dimension)} \\
& \mid & \cc{replicate}\ k\ t & \text{(add an outermost dimension of size $k$ by replicating contents)} \\
& \mid & \ctransp{k_1, \ldots, k_n}\ t & \text{(generalised transposition; $n = \text{rank of $t$}$,} \\
& & & \quad \text{$k_1,\ldots,k_n$ must be a permutation of $0,\ldots,n-1$)} \\
& \mid & \cc{reshape}\ sh\ t & \text{($\text{product of $\var{sh}$} = \text{product of (shape of $t$)}$)} \\
& \mid & \cc{build1}\ k\ (\clam i t) & \text{(construct a new array elementwise)} \\
k & \coloneqqq & 0 \mid 1 \mid 2 \mid \ldots & \text{(a constant (i.e.\ static) natural)} \\
sh & \coloneqqq & \shvec{} \mid k ::: sh & \text{(a constant shape; shorthand: $\shvec{k_1,\ldots,k_n} = k_1 ::: \ldots ::: k_n ::: \shvec{}$)} \\
ix & \coloneqqq & \ixvec{} \mid t ::: ix & \text{(a dynamic index; shorthand: $\ixvec{t_1,\ldots,t_n} = t_1 ::: \ldots ::: t_n ::: \ixvec{}$)} \\
is & \coloneqqq & \ixvec{} \mid x ::: is & \text{(index variables; shorthand: $\ixvec{x_1,\ldots,x_n} = x_1 ::: \ldots ::: x_n ::: \ixvec{}$)}
\end{array}
\end{equation*}
\caption{\label{fig:core-grammar}
    The grammar of the core language.
    We variously use other variable names than $x$ in expressions, especially ``$i$'' for embedded variables of type `$\Array\ \shvec{}\ \Int$'.
}
\end{figure}

\begin{figure}
{\textbf{Types:}\hfill}\vspace{-2pt}
\[ \begin{array}{@{}r@{\;}c@{\;}l@{}}
    \rho &\coloneqqq& \R \mid \Int \mid \Bool \\
    \sigma, \tau &\coloneqqq& \Array\ \var{sh}\ \rho \qquad \rlap{($\var{sh}$ is a list of non-negative integers)}
    \vspace{5pt}
\end{array} \]

{\textbf{Typing rules:}\hfill}\vspace{-4pt}
\begin{gather*}
    \begin{array}{@{}c@{}}
        \fbox{$\jisindex{\Gamma}{ix}{n}$\vphantom{hp}} \\[0.5em]
        \infrule
            {\Gamma \vdash t_1 : \Array\ \shvec{}\ \Int
             \precdots
             \Gamma \vdash t_n : \Array\ \shvec{}\ \Int}
            {\jisindex{\Gamma}{\ixvec{t_1, \ldots, t_n}}{n}}
    \end{array}
    \qquad
    \begin{array}{@{}c@{}}
        \fbox{$\jisnumeric{\rho}$\vphantom{hp}} \\[0.5em]
        \infrule{\strut}{\jisnumeric{\R}} \qquad
        \infrule{\strut}{\jisnumeric{\Int}}
    \end{array}
    \\[0.5em]
    \fbox{$\Gamma \vdash t : \tau$\vphantom{hp}} \\[0.5em]
    \infrule
        {\text{$c$ an array of shape $\var{sh}$ filled with $\rho$s}}
        {\Gamma \vdash c : \Array\ \var{sh}\ \rho}
    \\
    \infrule
        {x : \Array\ \var{sh}\ \rho \in \Gamma}
        {\Gamma \vdash x : \Array\ \var{sh}\ \rho}
    \rulesep
    \infrule
        {\Gamma \vdash v : \Array\ \var{sh}_1\ \rho_1 \precsep
         \Gamma, x : \Array\ \var{sh}_1\ \rho_1 \vdash u : \Array\ \var{sh}_2\ \rho_2}
        {\Gamma \vdash \cletin{x = u}{v} : \Array\ \var{sh}_2\ \rho_2}
\end{gather*}
\[ \begin{array}{@{}l@{\ }l@{}}
    \cc{cond} : \Array\ 0\ \Bool \to \Array\ \var{sh}\ \rho \to \Array\ \var{sh}\ \rho \to \Array\ \var{sh}\ \rho \\
    \mathit{op} : \Array\ \var{sh}\ \rho \to \Array\ \var{sh}\ \rho \to \Array\ \var{sh}\ \rho & \text{(binary arithmetic operations)} \\
    \mathit{op} : \Array\ \var{sh}\ \rho \to \Array\ \var{sh}\ \rho \to \Array\ \var{sh}\ \Bool & \text{(binary comparison operations)} \\
    \mathit{op} : \Array\ \var{sh}\ \rho \to \Array\ \var{sh}\ \rho & \text{(unary arithmetic operations)} \\
    \cc{sumOuter} : \Array\ (k ::: \var{sh})\ \rho \to \Array\ \var{sh}\ \rho & \text{(for numeric $\rho$ (i.e.\ $\R$, $\Int$))}
\end{array} \]
\begin{gather*}
    \infrule
        {\Gamma \vdash t : \Array\ \shvec{k_1, \ldots, k_n}\ \rho \precsep
         \jisindex{\Gamma}{\var{ix}}{m} \precsep
         m \leq n}
        {\Gamma \vdash \cc{index}\ t\ \var{ix} : \Array\ \shvec{k_{m+1}, \ldots, k_n}\ \rho}
    \\
    \infrule
        {\Gamma \vdash t : \Array\ \shvec{k_1, \ldots, k_{m_2}, k_{m_2+1}, \ldots, k_n}\ \rho \\
         \jisindex{\Gamma,i_1,\ldots,i_{m_1} : \Array\ \shvec{}\ \Int}{\var{ix}}{m_2} \precsep
         m_1, m_2 \leq n}
        {\Gamma \vdash \cc{gather}\ \shvec{k'_1, \ldots, k'_{m_1}, k_{m_2+1}, \ldots, k_n}\ t\ (\clam{\ixvec{i_1,\ldots,i_{m_1}}} \var{ix}) : \Array\ \shvec{k'_1, \ldots, k'_{m_1}, k_{m_2+1}, \ldots, k_n}\ \rho}
    \\
    \infrule
        {\Gamma \vdash t : \Array\ \shvec{k_1, \ldots, k_{m_1}, k_{m_1+1}, \ldots, k_n}\ \rho \precsep
         \jisnumeric{\rho} \\
         \jisindex{\Gamma,i_1,\ldots,i_{m_1} : \Array\ \shvec{}\ \Int}{\var{ix}}{m_2} \precsep
         m_1, m_2 \leq n}
        {\Gamma \vdash \cc{scatter}\ \shvec{k'_1, \ldots, k'_{m_2}, k_{m_1+1}, \ldots, k_n}\ t\ (\clam{\ixvec{i_1,\ldots,i_{m_1}}} \var{ix}) : \Array\ \shvec{k'_1, \ldots, k'_{m_2}, k_{m_1+1}, \ldots, k_n}\ \rho}
    \\
    \infrule
        {\Gamma \vdash t_1 : \Array\ \var{sh}\ \rho
         \precdots
         \Gamma \vdash t_n : \Array\ \var{sh}\ \rho}
        {\Gamma \vdash \ravel{t_1, \ldots, t_n} : \Array\ (n ::: \var{sh})\ \rho}
    \\
    \infrule
        {\text{$k$ a constant integer $\geq 0$} \precsep
         \Gamma \vdash t : \Array\ \var{sh}\ \rho}
        {\Gamma \vdash \cc{replicate}\ k\ t : \Array\ (k ::: \var{sh})\ \rho}
    \\
    \infrule
        {\text{$j_1,\ldots,j_m$ is a permutation of $0,\ldots,m-1$} \precsep
         \Gamma \vdash t : \Array\ \shvec{k_1, \ldots, k_n}\ \rho \precsep
         m \leq n}
        {\Gamma \vdash \ctransp{j_1,\ldots,j_m}\ t : \Array\ \shvec{k_{j_1+1}, \ldots, k_{j_m+1}, k_{m+1}, \ldots, k_n}\ \rho}
    \\
    \infrule
        {\Gamma \vdash t : \Array\ \shvec{k_1, \ldots, k_m}\ \rho \precsep
         \prod_{i=1}^m k_i = \prod_{i=1}^n k'_i}
        {\Gamma \vdash \cc{reshape}\ \shvec{k'_1, \ldots, k'_n}\ t : \Array\ \shvec{k'_1, \ldots, k'_n}\ \rho}
    \\
    \infrule
        {\text{$k$ a constant integer $\geq 0$} \precsep
         \Gamma, i : \Array\ \shvec{}\ \Int \vdash t : \Array\ \var{sh}\ \rho}
        {\Gamma \vdash \cc{build1}\ k\ (\clam{i} t) : \Array\ (k ::: \var{sh})\ \rho}
\end{gather*}
\caption{\label{fig:core-typing}
    Typing rules for the core language.
    \TODO{Compact these typing rules somehow?}
}
\end{figure}

The syntax of the core language is given in \cref{fig:core-grammar}, and the type system and typing rules (some abbreviated as pseudo-type-signatures) are given in \cref{fig:core-typing}.
Some notes to clarify parts of \cref{fig:core-grammar,fig:core-typing} that might be unfamiliar or unclear:
\begin{itemize}
\item
    We use the word \emph{rank} to denote the number of dimensions of an array, and by extension, for array-typed terms, the number of dimensions of their output.

\item
    The language is shape-typed: the \emph{shape} (the list of all dimensions' sizes) of an array is reflected on the type-level.
    This results in typing that is stronger than most other array languages.
    For example, a 3-by-2 array of scalars, where ``2'' is the size of the inner dimension, would have type $\Array\ \shvec{3, 2}\ \R$.

\item
    Tuples and nested arrays are unsupported:\footnote{Our implementation fullly supports pairs, and provides some support for nested arrays through \url{https://hackage.haskell.org/package/ox-arrays}. These extensions do not produce interesting algorithmic problems, so this paper excludes them for simplicity.} every expression is of array type, and the only arrays are multidimensional arrays of \emph{element types} (denoted by $\rho$ in \cref{fig:core-typing}).
    Hence, all arrays are \emph{regular}: there are no jagged arrays.
    What one might expect to be scalar subexpressions are really zero-dimensional arrays in our language; see, for example, the ``$\jisindex{\Gamma}{\var{ix}}{n}$'' judgement in \cref{fig:core-typing}, as well as its use in e.g.\ the rule for $\cc{index}$ in the same figure.

\item
    `$\mathit{op}$' stands for an arbitrary unary or binary primitive arithmetic or comparison operator on scalars; we consider these to automatically broadcast to arrays of equal shapes.
    Thus, $s + t$ is valid if $s$ and $t$ are terms producing arrays of the same shapes, and computes their elementwise sum.
    We lump all of these together in a single syntactic element because the differences are immaterial in most of the algorithms in this paper.

\item
    `$\cc{tr}$' is a generalised array transposition: if $a$ is a 4-dimensional array with shape $\shvec{5, 3, 6, 9}$ (with 5 being the outermost dimension and 9 the innermost), then `$\ctransp{3, 0, 1, 2}\ a$' is a 4-dimensional array with shape $\shvec{9, 5, 3, 6}$.
    Also see the typing rule for $\cc{tr}$ in \cref{fig:core-typing}.

\item
    `$\cc{sumOuter}$' reduces elementwise along the \emph{outermost} dimension.
    This is to be dual with $\cc{replicate}$, making its derivative rule in \cref{sec:ad-dual-arrays} more elegant; but note that an e.g.\ inner-dimension sum can be recreated from $\cc{sumOuter}$ by combining it with some transpositions.

\item
    The expressions that make up an index expression $\var{ix}$ cannot have sharing between them in our grammar.
    This is relevant in the index mapping functions passed to $\cc{gather}$ and $\cc{scatter}$ (for their semantics, see below in \cref{sec:core-gather-scatter}).
    This is for simplicity of presentation and not a fundamental limitation.
\end{itemize}

There are a number of peculiarities and restrictions in this core language that result from the ``\botnamelc'' that we will apply to the program in \cref{sec:bot}, before the actual differentiation.
The most important ones are:
\begin{itemize}
\item
    Statically known array shapes only: the size of array dimensions is not allowed to depend on intermediate values computed earlier in the program.
    If one wants to rerun a differentiated program on differently sized arrays, one has to re-differentiate and re-compile the program.
\item
    Conditionals have strict semantics, sometimes called \emph{selections}: `$\cc{cond}\ t\ u\ v$' first evaluates $t$, $u$ and $v$, and subsequently returns either the second or the third argument depending on the value of the first.
\item
    Finally, the language does not support separate top-level functions; all must be a single expression (with possible internal let-bindings, of course).
\end{itemize}
Very roughly, these restrictions exist because we want to be able to eliminate $\cc{build1}$ from the program by ``vectorising'' it into other array operations.
After introducing the \botnamelc, we discuss these restrictions again in \cref{sec:core-language-design-justification}.

\subsection[Semantics of gather and scatter]{Semantics of $\cc{gather}$ and $\cc{scatter}$}\label{sec:core-gather-scatter}

\paragraph{Gather}
As can be inferred from its typing rule in \cref{fig:core-typing}, the `$\cc{build1}$' primitive in the language constructs a $(k + 1)$-dimensional array given a function that maps a single index to a $k$-dimensional array.
Using $\cc{build1}$, we can create a multidimensional $\cc{build}$ operation as a notational shorthand:
\[
    \cc{build}\ \shvec{k_1, \ldots, k_n}\ (\clam{\ixvec{i_1, \ldots, i_n}} t)
    \coloneqq \cc{build1}\ k_1\ (\clam{i_1} \ldots (\cc{build1}\ k_n\ (\clam{i_n} t))\ldots)
\]
Then, the `$\cc{gather}$' primitive is really just a specialisation of this `$\cc{build}$':
\[
    \cc{gather}\ \var{sh}\ a\ (\clam{\var{is}} t)
    = \cc{build}\ \var{sh}\ (\clam{\var{is}} \cc{index}\ a\ t))
\]
We need $\cc{gather}$ explicitly in the language, despite it being expressible using $\cc{build1}$, in order to properly represent the output of the \botnamelc.
This will be discussed in more detail in the following sections.

\paragraph{Scatter}
The `$\cc{scatter}$' operation is the dual of `$\cc{gather}$', and is included in the language not only because it is necessary for histogram-like operations (which cannot be otherwise expressed using the rest of the core language), but also because it forms the reverse derivative of `$\cc{gather}$' (see \cref{sec:ad-dual-arrays-terms}).
In `$\cc{scatter}\ \var{sh}\ t\ (\clam{\ixvec{i_1, \ldots, i_n}} ix)$', the argument $\var{sh}$ gives the shape of the result of the operation, $t$ is the array of input values to be scattered, and the function determines \emph{where} the elements of $t$ are to be written in the output array.
Multiple values sent to the same location are added with $(+)$.
For example, using single-dimensional arrays only, the folowing program (writing flooring integer division as a binary operator $(\mathrm{div})$):
\[
    \cc{scatter}\ \shvec{6}\ \arrlit{1,2,3,4,5,6,7,8,9}\ (\clam{\ixvec{i}} \ixvec{i \mathbin{\mathrm{div}} 2})
\]
returns the array $\arrlit{3,7,11,15,9,0}$.
This result is computed as follows:
\begin{itemize}
\item
    Indices 0 and 1 are both sent to $0 \mathbin{\mathrm{div}} 2 = 1 \mathbin{\mathrm{div}} 2 = 0$, thus the values 1 and 2 are added together to yield 3.
\item
    The last value in the source array (9) is sent to index $8 \mathbin{\mathrm{div}} 2 = 4$, and it is the only element sent to this position; hence the result has 9 at index 4.
\item
    No element is sent to index 5 of the output, hence the result is zero.
\end{itemize}

\section{Naive Extension To Arrays: Unsuccessful}\label{sec:ad-array-fail}

Now that we have an array language to differentiate, let us try to extend the basic dual-numbers reverse AD algorithm from \cref{sec:ad-naive} to our core language from \cref{sec:core-language} in the ``obvious'' way, and see what goes wrong.
The problems that arise will inform the changes and optimisations that we make, eventually resulting in the final algorithm.\footnote{The designs in this section were already suggested in \cite[\S8.3]{2022-krawiec-dualrev}; we discuss them in more detail and improve upon them.}

\subsection[Scalar Dual Numbers: The Delta Explosion Problem]{Scalar Dual Numbers: The $\tDelta$ Explosion Problem}\label{sec:ad-array-fail-scalar}

The promise of dual-numbers AD is that it is extensible to almost any imaginable program construct by just adding more rules to the code transformation $D[-]$ that map over the new constructs in a structure-preserving way.
Let us do this for arrays, seeing an array as little more than a very large product type.
On the type level, we get:
\[ D[\Array\ \var{sh}\ \rho] = \Array\ \var{sh}\ D[\rho] \]
but what of the array operations?
In \cref{fig:dn-ad-terms} we had:
\[ D[s\ t] = \mathbf{do}\ f \leftarrow D[s]; x \leftarrow D[t]; f\ x \]
and our array operations look like functions, so ostensibly we get something like this:\footnote{We are abusing syntax here: e.g.\ `$\cc{build1}\ k\ (\clam i t)$' is a term, but technically `$\cc{build1}$' is not.}
\[ \begin{array}{@{}r@{\;}c@{\;}l@{}}
    D[\cc{build1}\ k\ (\clam i t)] &=& \mathbf{do}\ f \leftarrow D[\cc{build1}]; n \leftarrow D[k]; f\ n\ (\clam i D[t]) \\
    D[\cc{index}\ t\ \ixvec{t_1, \ldots, t_n}] &=& \mathbf{do}\ f \leftarrow D[\cc{index}]; a \leftarrow D[t]; i_1 \leftarrow D[t_1]; \ldots; i_n \leftarrow D[t_n]; f\ a\ \ixvec{i_1, \ldots, i_n} \\
    D[\cc{sumOuter}\ t] &=& \mathbf{do}\ f \leftarrow D[\cc{sumOuter}]; a \leftarrow D[t]; f\ a
\end{array} \]
But then what are $D[\cc{build1}]$, $D[\cc{index}]$, $D[\cc{sumOuter}]$, etc.?

\simon{unconvincing} \TODO{is this better?}
It is worth noting that $\cc{build1}$ and $\cc{index}$ on the one hand, and $\cc{sumOuter}$ on the other hand, are quite different when it comes to differentiation; let us look at $\cc{sumOuter}$ on scalars\footnote{The type system also allows summing arrays of integers, but $D[\Int] = \Int$, so we simply get $D[\cc{sumOuter}_{\Int}] = \mathbf{return}\ \cc{sumOuter}_{\Int}$.} first.
Let us limit ourselves to the 2-dimensional case for notational simplicity (read ``plane'' or ``subarray'' instead of ``row'' for 3 or higher dimensions, respectively).
Then the normal operation of $\cc{sumOuter}$ is to sum the rows of a matrix elementwise, producing a single row.
\[
    \newcommand\smallplus{\text{\small+}}
    \newcommand\pre[1]{\mathllap{\textcolor{lightgray}{#1}}}
    \newcommand\post[1]{\mathrlap{\textcolor{lightgray}{#1}}}
    \begin{array}{@{}l@{}c@{}l@{}}
        \hphantom{[[}\begin{array}{ccc}
            \pre{[[}1\post{,} & 2\post{,} & 3\post{],} \\[-5.5pt]
            \smallplus & \smallplus & \smallplus \\[-4pt]
            \pre{[}4\post{,} & 5\post{,} & 6\post{],} \\[-5.5pt]
            \smallplus & \smallplus & \smallplus \\[-4pt]
            \pre{[}7\post{,} & 8\post{,} & 9\post{]]}
        \end{array}\hphantom{]]}
        & = &
        \hphantom{[}\begin{array}{ccc}
            \pre{[}12\post{,} & 15\post{,} & 18\post{]}
        \end{array}\hphantom{]}
    \end{array}
\]
For the derivative of $\cc{sumOuter}$, instead of an array of scalars we get an array of dual numbers that we need to add:
\begin{equation}
    \newcommand\smallplus{\text{\small+}}
    \newcommand\pre[1]{\mathllap{\textcolor{lightgray}{#1}}}
    \newcommand\post[1]{\mathrlap{\textcolor{lightgray}{#1}}}
    \begin{array}{@{}l@{}c@{}l@{}}
        \hphantom{[[}\begin{array}{ccc}
            \pre{[[}(1,d_1)\post{,} & (2, d_2)\post{,} & (3, d_3)\post{],} \\[-5.5pt]
            \smallplus & \smallplus & \smallplus \\[-4pt]
            \pre{[}(4, d_4)\post{,} & (5, d_5)\post{,} & (6, d_6)\post{],} \\[-5.5pt]
            \smallplus & \smallplus & \smallplus \\[-4pt]
            \pre{[}(7, d_7)\post{,} & (8, d_8)\post{,} & (9, d_9)\post{]]}
        \end{array}\hphantom{]]}
        & = &
        \hphantom{[}\begin{array}{l}
            \pre{[}(12, \del{Share}\ \_\ (\del{Add}\ (\del{Add}\ d_1\ d_2)\ d_3))\post{,} \\
            \quad (15, \del{Share}\ \_\ (\del{Add}\ (\del{Add}\ d_4\ d_5)\ d_6))\post{,} \\
            \quad (18, \del{Share}\ \_\ (\del{Add}\ (\del{Add}\ d_7\ d_8)\ d_9))\post{]}
        \end{array}\hphantom{]}
    \end{array}
    \label{eq:dsumOuter-matrix}
\end{equation}
Note that there must be $\del{Share}$ nodes around the $\tDelta$ terms in the result because they may be used multiple times; the `$\_$'s stand for unique generated IDs.
For the time being, let us assume that there is some function `$\name{DsumOuter}$' in the target language that does precisely this: take an $n$-dimensional array of dual numbers and return an $(n-1)$-dimensional array of dual numbers by summing elementwise along the outer dimension.
With this, we get:
\[
    D[\cc{sumOuter}\ t] = \mathbf{do}\ a \leftarrow D[t]; \name{DsumOuter}\ a
\]
Note that `$\name{DsumOuter}$' is a monadic operation because it needs to generate unique IDs for the $\del{Share}$ nodes in \cref{eq:dsumOuter-matrix}.

Observe that we needed to examine the performed computation, differentiate it, and represent the differentiated result again as a program.
For $\cc{build1}$ and $\cc{index}$, the story is quite different.
This is because these operations are both parametrically polymorphic in the element type of the arrays they produce ($\cc{build1}$) or consume ($\cc{index}$).
They just ``move elements around'', and are sufficiently uncaring when the array element type changes from scalar to non-scalar, or (indeed!)\ to dual numbers.
In this, $\cc{build1}$ and $\cc{index}$ are no different than `$\name{fst}$', lambda-abstraction, etc.\ from \cref{fig:dn-ad-terms}, which we could just differentiate to themselves (modulo monadic lifting).
And indeed, it turns out that doing the same to $\cc{build1}$ and $\cc{index}$ works equally well, as long as we handle the fact that any functions passed \emph{to} $\cc{build1}$ are monadically lifted too.
We obtain the following derivatives:
\begin{equation}
\begin{array}{@{}r@{\;}c@{\;}l@{}}
    D[\cc{build1}\ k\ (\clam i t)] &=& \mathbf{sequence}\ (\cc{build1}\ k\ (\clam i D[t])) \\
    D[\cc{index}\ t\ \ixvec{t_1, \ldots, t_n}] &=& \mathbf{do}\ \multilineT{
        a \leftarrow D[t]; i_1 \leftarrow D[t_1]; \ldots, i_n \leftarrow D[t_n] \\
        \mathbf{return}\ (\cc{index}\ a\ \ixvec{i_1, \ldots, i_n})
    }
\end{array} \label{eq:naive-array-derivs}
\end{equation}
where `$\mathbf{sequence}$' has type $\Array\ \var{sh}\ (\tIdGen\ \rho) \to \tIdGen\ (\Array\ \var{sh}\ \rho)$ and evaluates all monadic computations in the array, producing an array of results.

Apart from this wrinkle of having to propagate the effects, we indeed maintain the structure-preserving quality of the transformation.
Because these array operations do not act on scalars directly, they are just ``structure'', and are thus preserved by the algorithm.\footnote{%
    Proving correctness of these derivatives is somewhat subtle.
    The $\cc{index}$ operation is algebraically linear, and since a (forward) derivative is the best linear approximation of a function, the forward derivative of a linear function is just itself.
    (This idea was explored further by \citeN{2022-multilinear-derivatives}.)
    For `$\cc{build1}$', one can build confidence by expanding into individual scalar operations; a full proof requires an induction argument (using logical relations) following~\cite{2020-ad-gluing,2022-ad-logical-relations}.
}

\newcommand\texDot{t_{\text{dot}}}

While these definitions are correct, and the complexity requirements are met, the resulting performance is very unsatisfactory.
Consider a simple dot product operation, expressed by the term $\texDot$ with two free variables, $a$ and $b$:
\[ \begin{array}{@{}l@{}}
    a : \Array\ \shvec{n}\ \R,
    b : \Array\ \shvec{n}\ \R \\
    \vdash
    \texDot = \cc{sumOuter}\ (\cc{build1}\ n\ (\clam{\ixvec{i}} \cc{index}\ a\ \ixvec{i} \times_\R \cc{index}\ b\ \ixvec{i}))
    : \Array\ \shvec{}\ \R
\end{array} \]
Of course, this implementation is suboptimal: a dedicated loop can easily be more than $4\times$ faster than this program, in part by eliminating the materialised intermediate array of products.
Its naive derivative, however, is far worse still, and exemplifies the problem with the array operation derivatives in \cref{eq:naive-array-derivs} just above.
After transformation to dual numbers (and some basic simplifications for readability), the program looks as follows:
\[
    D[\texDot] = \mathbf{do}\ \multilineT{
        c \leftarrow \mathbf{sequence}\ (\cc{build1}\ n\ (\clam i \mathbf{do} \\
        \hspace{2.65cm} \multilineT{
            \mathbf{let}\ (x_1, d_1) = \cc{index}\ a\ \ixvec{i}; (x_2, d_2) = \cc{index}\ b\ \ixvec{i} \\
            \var{id} \leftarrow \mgenid \\
            \mathbf{return}\ (\multilineT{
                x_1 \times_\R x_2 \\
                \leadcomma \del{Share}\ \var{id}\ (\del{Add}\ (\del{Scale}\ x_2\ d_1)\ (\del{Scale}\ x_1\ d_2)))))
            }
        } \\
        \name{DsumOuter}\ c
    }
\]
Consider what $D[\texDot]$ does: it will build \emph{a $\tDelta$ term the size of the input array}, containing (when counting carefully) $5n+1$ $\tDelta$ data constructors when given inputs of length $n$.\footnote{$4n$ from the lambda to $\cc{build1}$ and $n+1$ for the $\del{Add}$s and the outer $\del{Share}$ in $\name{DsumOuter}$.}
Aside from using far too much memory (thus also destroying memory locality), this whole tree will need to be interpreted, node by node, in the reverse pass, which allocates even more memory to hold various administrative data about all the $2n$ inputs plus the $n+1$ $\del{Share}$ nodes.
Furthermore, there is little hope of vectorising the (actually very structured) multiplications and additions in the reverse pass.

Meanwhile, a proper implementation of the reverse derivative of a dot product simply consists of two (very efficiently implementable) multiplications of a scalar with a vector.
Thus, even if this approach of adding arrays to our language is very neat, simple and extensible, it will not fly in practice.

\paragraph{Doing better}
Let us call the problem of allocating (and interpreting) far too many $\tDelta$ nodes the \emph{$\tDelta$ explosion} problem.
To fix this problem, the first thing we notice is that while we do create a tremendous number of $\tDelta$ nodes, many of them look very similar!
Indeed, \emph{all} of the runs of the lambda in `$\cc{build1}$' in $D[\texDot]$ return a $\tDelta$ subgraph with the exact same structure: $\del{Share}\ \_\ (\del{Add}\ (\del{Scale}\ \_\ d_1)\ (\del{Scale}\ \_\ d_2))$, where $d_1$ and $d_2$ are the $\tDelta$ terms of the scalars $\cc{index}\ a\ \ixvec{i}$ and $\cc{index}\ b\ \ixvec{i}$.
The subgraphs differ only in the $\tID$ in the $\del{Share}$ node and the scalars in the $\del{Scale}$ nodes.
It would be good if we can represent the $\tDelta$ term computed by $D[\texDot]$ more compactly: from the viewpoint of the arrays, are there really only a few (bulk) operations being done, and each of those has a simple derivative.

In essence, the $\tDelta$ term produced by a given differentiated program is really a \emph{trace} of the primitive operations executed by the program:\footnote{The ``entries'' in this trace contain only the partial derivatives of the operations executed, not the operations themselves.} with `$\cc{build1}$' and `$\cc{sumOuter}$' as-is, this trace is too fine-grained for efficient differentiation.
By itself, the fact that dual-numbers reverse AD generates a trace is unsurprising and already noted in~\cite[\S8]{2023-smeding-dualrev}.
However, what what we really want is a trace of size $O(\#\text{array operations})$ instead of $O(\#\text{scalar operations})$, so that we still know the array structure of the source program when we start computing its gradient.

\subsection{Dual Arrays: A Step In the Right Direction}\label{sec:ad-array-dual-arrays}

To accomplish this reduction of the trace (i.e.\ $\tDelta$ term) size, we have to teach the algorithm to consider arrays differentiable objects in and of themselves.
To that effect, we replace the naive rule $D[\Array\ \var{sh}\ \rho] = \Array\ \var{sh}\ D[\rho]$ with the following:
\[
    D[\Array\ \var{sh}\ \rho] = (\Array\ \var{sh}\ \rho, \tDelta\ \var{sh})
\]
With this rule, $\tDelta$ must now also be able to represent forward derivatives of \emph{array} computations.
Hence, we give it a type parameter $\var{sh}$: where `$\tDelta$' described the forward derivative of a computation of type $\R$ before, this `$\tDelta\ \var{sh}$' describes the forward derivative of a computation of type $\Array\ \var{sh}\ \R$.
Its base is the same as before ($\del{Zero}$, $\del{Input}$, $\del{Add}$, $\del{Scale}$, $\del{Share}$), but we add more constructors for each of the primitive array operations:
\[ \begin{array}{@{}l@{}}
    \mathbf{data}\ \tDelta\ \var{sh}\ \mathbf{where} \\
    \quad \begin{array}[t]{@{}l@{\ }c@{\ }l@{}}
        \del{Zero} &::& \tDelta\ \var{sh} \\
        \del{Input} &::& \tDVarName \to \tDelta\ \var{sh} \\
        \del{Add} &::& \tDelta\ \var{sh} \to \tDelta\ \var{sh} \to \tDelta\ \var{sh} \\
        \del{Scale} &::& \Array\ \var{sh}\ \R \to \tDelta\ \var{sh} \to \tDelta\ \var{sh} \\
        \del{Share} &::& \tID \to \tDelta\ k \to \tDelta\ k \\
        \mathrlap{\ccomment{Most array operations get a dedicated constructor:}} \\
        \del{Index} &::& \tDelta\ \shvec{k_1, \ldots, k_n} \to \Ix\ m \to \tDelta\ \shvec{k_{m+1}, \ldots, k_n} \\
        \del{SumOuter} &::& \tDelta\ (k ::: \var{sh}) \to \tDelta\ \var{sh} \\
        \del{Gather}\ &\multicolumn{2}{@{}l@{}}{:: \tDelta\ \shvec{k_1, \ldots, k_{m_2}, k_{m_2+1}, \ldots, k_n} \to (\Ix\ m_1 \to \Ix\ m_2)} \\
          &\multicolumn{2}{@{}l@{}}{\to \tDelta\ \shvec{k'_1, \ldots, k'_{m_1}, k_{m_2+1}, \ldots, k_n}} \\
        \del{Replicate} &::& \tDelta\ \var{sh} \to \tDelta\ (k ::: \var{sh}) \\
        \mathrlap{\ccomment{etc., others elided}}
    \end{array}
\end{array} \]
where $\Ix\ m$ is an $m$-dimensional index, i.e.\ simply $m$ integers:
\[ \begin{array}{@{}l@{}}
    \mathbf{data}\ \Ix\ k\ \mathbf{where} \\
    \quad \begin{array}[t]{@{}l@{\ }l@{\ }l@{}}
        \IZ &::& \Ix\ 0 \\
        (\IS) &::& \Int \to \Ix\ k \to \Ix\ (k + 1)
    \end{array}
\end{array} \]

Indeed, because we use $\tDelta$ to represent the forward derivative of programs in our language, $\tDelta$ must certainly be able to (somehow) represent the forward derivatives of all \emph{primitive operations} in our language.
So far, for primitive arithmetic operations with at most one scalar output, we could make do with $\del{Zero}$, $\del{Scale}$ and $\del{Add}$ for this purpose, because (forward) derivatives are linear and all linear functions $\R^n \to \R$ are simply linear combinations.
For example, the forward derivative of $\lambda x\,y.\ x \times_\R y$ at inputs $x, y$ is $\lambda \var{dx}\,\var{dy}.\ y \times_\R \var{dx} + x \times_\R \var{dy}$, which is written $\lambda \var{dx}\,\var{dy}.\ \del{Add}\ (\del{Scale}\ y\ \var{dx})\ (\del{Scale}\ x\ \var{dy})$ in the $\tDelta$ language (compare \cref{fig:dn-ad-terms}).

For more general linear functions, however, this normal form (a linear combination) generalises to a matrix: for each individual scalar in the output of the linear function, we could give an $\del{Add}/\del{Scale}/\del{Zero}$ expression in terms of the operation's inputs.
The resulting array of scale factors (for each output with respect to each input) is precisely the Jacobian matrix of the operation that this linear function is the derivative of.
The size of this matrix is $(\#\text{scalars in output}) \cdot (\#\text{scalars in input})$, and while typically sparse, the sparsity pattern is heavily dependent on the specific operation (e.g.\ $\cc{gather}$, $\cc{replicate}$, $\cc{sumOuter}$).
So a good sparse representation would need to be a sum type over all the primitive operations, storing in each case just the information necessary to reconstruct the full Jacobian.\footnote{It would be grossly inefficient to just materialise all those Jacobians densely.}

It turns out that a very good sparse representation of the Jacobian of a particular array operation is simply a program that computes the forward derivative (how its output changes in response to a particular change to its inputs; the \emph{total derivative}).
The $\tDelta$ constructors that we add for the primitive array operations, such as $\del{Index}$, $\del{SumOuter}$, etc.\ in $\mathbf{data}\ \tDelta\ \var{sh}$ above, are precisely that --- and their semantics is what one expects, just like the semantics of $\del{Add}$, $\del{Scale}$, etc.\ whas precisely what one expects under $\eval$.

\paragraph{Forward derivative of a linear function is itself}
Given the function $f = \lambda x\,y.\ 2x + 5y$, some input $x, y$, and some small change $\Delta x, \Delta y$ to that input, how much does $f\ (x + \Delta x)\ (y + \Delta y)$ differ from $f\ x\ y$?
Well, $2\Delta x + 5\Delta y$, surely, because $f$ is a linear function (a vector space homomorphism).
The forward derivative of $f$ at some input $x,y$ is simply $f$ itself.
This holds for all linear functions, and surprisingly many useful functions are linear: as a special case, as long as a function just rearranges and/or adds values from its input, it is certainly linear, and this is in fact true for all array operations in our core language (\cref{fig:core-grammar}) except for the arithmetic operations $\mathit{op}$ and our higher-order operation $\cc{build1}$ (because its lambda argument may be non-linear).
Thus, for example, the Delta constructor corresponding to $\cc{index}$ (i.e. its forward derivative) has the same semantics as the $\cc{index}$ operation itself, and is hence simply called `$\del{Index}$'.

\subsection{The $\tDelta$ explosion problem again}\label{sec:ad-array-build-explosion}

So the linear array operations are relatively straightforward, and it turns out that because our primitive arithmetic operations on arrays are elementwise, their derivative $\tDelta$ terms are essentially the same as we wrote for the scalar algorithm in \cref{fig:dn-ad-terms}.
So what about the one remaining operation, $\cc{build1}$?
\[
    D[\cc{build1}\ k\ (\clam i t)] = {??}
\]
Say for simplicity that $t : \Array\ \shvec{}\ \R$.
Then the source term is of type $\Array\ \shvec{k}\ \R$, hence the `$??$' should be of type $D[\Array\ \shvec{k}\ \R] = (\Array\ \shvec{k}\ \R, \tDelta\ \shvec{k})$.
Regardless of how exactly we compute it, the second component of `$??$' (call it $d$) should be a $\tDelta$ term that describes the forward derivative of the whole $\cc{build1}$ operation, and furthermore the number of nodes in this $\tDelta$ term $d$ should be much less than $k$ (if we are to fix the $\tDelta$ explosion problem).

However, regardless of how we extend the $\tDelta$ data type, surely $d$ depends on $t$, and furthermore it depends on the execution paths that each \emph{individual} execution of $t$ took for each index $i$.
Indeed, the source term's derivative depends on those execution paths, and $d$ should express precisely that derivative.
The most we can do (apart from non-compositionally handling special cases) is something like this:
\[ \begin{array}{@{}l@{}}
    \mathbf{data}\ \tDelta\ \var{sh}\ \mathbf{where} \\
    \quad \begin{array}[t]{@{}l@{}}
        \ccomment{\ldots{} \del{Zero}, \del{Input}, etc.\ as before} \\
        \del{Build1} :: \Array\ \shvec{k}\ (\tDelta\ \var{sh}) \to \tDelta\ (k ::: \var{sh}) \\
    \end{array}
\end{array} \]
The array of $\tDelta$s records the forward derivatives of the lambda for each index $i$ of the $\name{build}$ operation.
While this extension of $\tDelta$ surely allows us to write down a derivative of $\cc{build1}$, it does not solve anything: we have as many $\tDelta$ nodes as before, only organised differently.\footnote{Where previously we had $D[\Array\ \var{sh}\ \R] = \Array\ \var{sh}\ (\R, \tDelta)$, we now essentially have $D[\Array\ \var{sh}\ \R] = (\Array\ \var{sh}\ \R, \Array\ \var{sh}\ \tDelta)$; that $\Array\ \var{sh}\ \tDelta$ is just wrapped in a $\del{Build1}$ constructor.}

The source of this problem is that there is still element-wise scalar computation in the program (in particular, inside $\cc{build1}$), and this scalar computation must be differentiated faithfully to a $\tDelta$ term.\footnote{Specifically, computation on \emph{scalars}: computation on individual \emph{integers} is perfectly fine, because differentiation does not touch that.}
In contrast, \emph{first-order} array operations, e.g.\ elementwise arithmetic operators such as `$(\vecadd) :: \Array\ \var{sh}\ \R \to \Array\ \var{sh}\ \R \to \Array\ \var{sh}\ \R$' but also other bulk array operations that we already have such as $\cc{sumOuter}$, can be differentiated without any trouble as primitive operations in the language.
The trouble comes from user-written scalar-level code that is executed many times.

Before we fix the problem with $\cc{build1}$, let us investigate a second problem that has silently appeared: the \emph{one-hot problem}.

\subsection{The One-Hot Problem}\label{sec:ad-array-onehot}

As it turns out, the move to dual arrays has not only left the $\tDelta$ explosion problem unsolved (in $\cc{build1}$), it also creates a new problem: the derivative for $\cc{index}$, while easily written down, has very bad performance.
This problem is not yet visible in $\tDelta$; a priori, the $\del{Index}$ constructor of $\tDelta$ that we gave in \cref{sec:ad-array-dual-arrays} seems quite reasonable:
\[ \begin{array}{@{}l@{}}
    \mathbf{data}\ \tDelta\ \var{sh}\ \mathbf{where} \\
    \quad \begin{array}[t]{@{}l@{}}
        \ccomment{\ldots{} etc.} \\
        \del{Index} :: \tDelta\ \shvec{k_1, \ldots, k_n} \to \Ix\ m \to \tDelta\ \shvec{k_{m+1}, \ldots, k_n} \\
        \ccomment{\ldots{} etc.}
    \end{array}
\end{array} \]
%
However, while the forward derivative of indexing looks (and is) innocuous, its reverse derivative is problematic: the transposed interpretation of $\del{Index}$ (as $\eval$ will need to implement) has to take the cotangent of a single array element and produce a cotangent for the array that the element was projected from.
This produced array cotangent will be a \emph{one-hot array}: in case the indexed array is single-dimensional, this looks like $\arrlit{0, \ldots, 0, d, 0, \ldots, 0}$, where $d$ is the incoming cotangent for $\del{Index}$ and its position in the one-hot vector is the original projection index.
Especially if $\cc{index}$ is used many times, as in e.g.\ $\texDot$ from \cref{sec:ad-array-fail-scalar}, the fact that all these one-hot vectors will be added together means that the reverse derivative of $\texDot$ takes $O(n^2)$ time instead of $O(n)$ time!
This is unacceptable.

\paragraph{No sparse arrays}
The reader may wonder if we can solve this one-hot problem by representing cotangent arrays (i.e.\ the backpropagated derivatives in $\eval$) sparsely.
Indeed, by giving array cotangents a sparse runtime representation, creating a one-hot array becomes a constant-time operation, and with a sufficiently clever reduction implementation, an array of such sparse cotangents can even be summed relatively efficiently.
However, in practice, the majority of array cotangents are, or become, dense, and it is well-known that performing array operations on a sparse array that is actually completely full (i.e.\ dense data with a sparse representation) has significant overhead as compared to working with dense arrays directly.
Furthermore, sparse arrays do nothing to alleviate the $\tDelta$ explosion problem, and our solution to the $\tDelta$ explosion problem also mostly addresses the one-hot problem anyway.
Hence, we ignore sparse arrays as a potential solution in this paper.

\paragraph{Mutable accumulators}
Another way the one-hot problem can likely be addressed is to perform a Cayley transform, somewhat similar to the one described in~\cite[\S5]{2023-smeding-dualrev}, and replace one-hot vectors by local modifications of a mutable gradient accumulator.
However, this still does not solve the $\tDelta$ explosion problem, and avoiding mutable updates keeps the algorithm purely functional.

\subsection{Dual Arrays With Bulk Operations: Our Solution}\label{sec:ad-array-fail-solution}

As already remarked at the end of \cref{sec:ad-array-fail-scalar}, the computation paths that the lambda invocations in a $\cc{build1}$ actually take are in practice often extremely similar.
In the approach taken in this paper, we make the most of this observation: we design a code transformation that \emph{eliminates} the higher-order $\cc{build1}$ operation and turns it into first-order bulk array operations that can be differentiated neatly as-is.
This code transformation ``pushes'' $\cc{build1}$ and $\cc{index}$ down into expressions, and thus looks a lot like a certain kind of vectorisation, or ``unfusion''.
Representative rules are the following:
\[ \begin{ietarray}
\ietrule
    {\cc{build1}\ k\ (\clam i \mathit{op}\ t\ u)}
    {\mathit{op}\ (\cc{build1}\ k\ (\clam i t))\ (\cc{build1}\ k\ (\clam i u))}
    {} \\
\ietrule
    {\cc{index}\ (\cletin{x = v}{u})\ ix}
    {\cletin{x = v}{\cc{index}\ u\ ix}}
    {}
\end{ietarray} \]
To avoid ascribing even more meanings to the word ``vectorisation'', we call our transformation the \emph{\botnamelc}, or \BOT.

As an example, consider the following source program (fragment):
\[
    \cc{build1}\ k\ (\clam i \cc{index}\ a\ \ixvec{i} + 1)
\]
The \BOT will turn this into the following:\footnote{Of course, if $a$ has length $k$, then $\cc{gather}\ \shvec{k}\ a\ (\clam i i) = a$; in general, however, a $\cc{gather}$ is required.}
\[
    \cc{gather}\ \shvec{k}\ a\ (\clam i i) + \cc{replicate}\ k\ 1
\]

The crucial point is that these more specific array combinators do not suffer from the $\tDelta$ explosion problem like $\cc{build1}$ does.
In fact, their derivatives are quite small; examples, including $\cc{gather}$ and $\cc{replicate}$, are given in \cref{fig:ad-rules-dual-arrays} in \cref{sec:ad-dual-arrays}.\footnote{Elementwise-broadcasted primitive operations have an elementwise forward derivative; operations like $\cc{gather}$ have a small derivative because the lambda passed to $\cc{gather}$ need not be differentiated. We will see this again in \cref{sec:ad-dual-arrays-terms}.}

The result of the \BOT is thus that the user can write explicitly indexed code using `$\name{build}$' (as well as derived operations such as `$\name{map}$'), yet the AD algorithm can ignore the existence of `$\name{build}$' and work solely on effectively differentiable bulk array operations.
This solves the $\tDelta$ explosion problem from \cref{sec:ad-array-build-explosion}.

\paragraph{The one-hot problem}
The \BOT also has an effect on array indexing because an indexing operation inside $\cc{build1}$ is turned into a bulk $\cc{gather}$ operation, like in the small example above.
This is a big improvement: where the reverse derivative of $\cc{index}$ was a one-hot array, the reverse derivative of $\cc{gather}$ is a ``multi-hot'' array that contains non-zero values at all the positions that are read by at least one of the indexing operations collected together in that $\cc{gather}$, and zeros elsewhere.
In the simple example above, the entirety of $a$ is used (assuming $a$ has length $k$), so this ``multi-hot'' derivative of $\cc{gather}$ is actually fully dense.

More generally, it still holds that one usually does not ``ignore'' a large fraction of an array --- and if one does, there is typically some other part of the program that conversely uses just the part that was ignored here.
Hence, we expect that in practice, these ``multi-hot'' arrays arising from the reverse derivative of `$\name{gather}$' will be quite dense.
For the programs for which this is true, the \BOT not only solves the $\tDelta$ explosion problem, it also solves the indexing one-hot problem.

\paragraph{No general fold}
The downside of the \BOT is that while $\cc{build1}$ is fully supported, other higher-order array operations like `$\name{foldl}$' would significantly complicate the algorithm.
The reason is that while we could eliminate the higher-orderness inherent in $\cc{build1}$, we cannot eliminate a higher-order fold in the same way, so elementwise code remains in the program to be differentiated, and the $\tDelta$ explosion problem returns.
Thus, such other higher-order array operations are \emph{unsupported} in this paper.
However, a general reduction operation (as opposed to the typical first-order ones, such as `$\name{sum}$' and `$\name{maximum}$', which are supported just fine) is much less common in typical numerical code than a general elementwise computation, so the algorithm remains useful even with this limitation.

We describe the full transformation in \cref{sec:bot}.


\subsection{Paper structure}

The full description of the algorithm consists of three parts:
\begin{itemize}
\item
    \cref{sec:bot}: The \botnamelc that eliminates $\cc{build1}$ from the input program.

\item
    \cref{sec:ad-dual-arrays}: The adaptation of the dual-numbers reverse AD algorithm from \cref{sec:ad-naive} to work on dual arrays.
    This works out and results in an efficient gradient computation because there is no \emph{active} (roughly: differentiable) non-broadcasted elementwise computation any more.

\item
    \cref{sec:restaging}: Making the reverse pass (in particular, $\eval$) symbolic.
    This allows us to differentiate a program \emph{once} and run it on many inputs, solving a problem identified in \cref{sec:ad-naive}.
\end{itemize}

\section{\botnametc}\label{sec:bot}

As introduced in \cref{sec:ad-array-fail-solution}, the aim of the \botnamelc (\BOT) is to eliminate `$\cc{build1}$', and as much as possible `$\cc{index}$', from the core language.
This allows users to write explicitly indexed code, but lets the AD algorithm of \cref{sec:ad-dual-arrays} work on mostly first-order code.
As a result, (1) the $\tDelta$ trace generated by the AD-transformed code will be small (it does not refer to individual scalar operations any more, but only the bulk operations that contain them) and (2) projections from large structures (i.e.\ $\cc{index}$ and $\cc{gather}$) are batched as much as possible, meaning that we generate very few one-hot/multi-hot cotangent arrays.
This addresses the two problems ($\tDelta$ exposion and one-hot cotangent arrays) that we saw in \cref{sec:ad-array-fail-scalar,sec:ad-array-build-explosion,sec:ad-array-onehot}.

\subsection{The transformation}

\begin{figure}
\[ \begin{ietarray}
\ietrule
    {\cc{build1}\ k\ (\clam i i)}
    {\ravel{0,\ldots,k - 1}}
    {} \\
\ietrule
    {\cc{build1}\ k\ (\clam i t)}
    {\cc{replicate}\ k\ t}
    {$i \not\in \mathit{FV}(t)$} \\
\ietrule
    {\cc{build1}\ k\ (\clam i \cletin{x = v}{u})}
    {\multilineT{
        \clet x = \cc{build1}\ k\ (\clam i v) \\
        \cin \cc{build1}\ k\ (\clam i u[\cc{index}\ x\ \ixvec{i}/x])
    }}
    {} \\
\ietrule
    {\cc{build1}\ k\ (\clam i \cc{cond}\ b\ u\ v)}
    {\cc{build1}\ k\ (\clam i \cc{index}\ \ravel{u,v}\ \ixvec{\cc{cond}\ b\ 0\ 1})}
    {} \\
\ietrule
    {\cc{build1}\ k\ (\clam i \mathit{op}\ t\ u)}
    {\mathit{op}\ (\cc{build1}\ k\ (\clam i t))\ (\cc{build1}\ k\ (\clam i u))}
    {} \\
\ietrule
    {\cc{build1}\ k\ (\clam i \mathit{op}\ t)}
    {\mathit{op}\ (\cc{build1}\ k\ (\clam i t))}
    {} \\
\ietrule
    {\cc{build1}\ k\ (\clam i \cc{sumOuter}\ t)}
    {\cc{sumOuter}\ (\cc{tr}\ (\cc{build1}\ k\ (\clam i t)))}
    {} \\
\ietrule
    {\cc{build1}\ k\ (\clam i \cc{gather}\ sh\ t\ (\clam {is} ix)))}
    {\cc{gather}\ \multilineT{
        (k\ \cc{:::}\ sh)\ (\cc{build1}\ k\ (\clam i t)) \\
        (\clam {(i\ \cc{:::}\ is)} i\ \cc{:::}\ ix)
    }}
    {} \\
\ietrule
    {\cc{build1}\ k\ (\clam i \cc{scatter}\ sh\ t\ (\clam {is} ix))}
    {\cc{scatter}\ \multilineT{
        (k\ \cc{:::}\ sh)\ (\cc{build1}\ k\ (\clam i t)) \\
        (\clam {(i\ \cc{:::}\ is)} i\ \cc{:::}\ ix)
    }}
    {} \\
\ietrule
    {\cc{build1}\ k\ (\clam i \ravel{t_1,\ldots,t_n})}
    {\cc{tr}\ \ravel{\cc{build1}\ k\ (\clam i t_1), \ldots, \cc{build1}\ k\ (\clam i t_n)}}
    {} \\
\ietrule
    {\cc{build1}\ k\ (\clam i \cc{replicate}\ n\ t)}
    {\cc{tr}\ (\cc{replicate}\ n\ (\cc{build1}\ k\ (\clam i t)))}
    {} \\
\ietrule
    {\cc{build1}\ k\ (\clam i \ctransp{b_0, \ldots, b_n}\ t)}
    {\ctransp{0, b_0 + 1, b_1 + 1, \ldots, b_n + 1}\ (\cc{build1}\ k\ (\clam i t))}
    {} \\
\ietrule
    {\cc{build1}\ k\ (\clam i \cc{reshape}\ sh\ t)}
    {\cc{reshape}\ (k\ \cc{:::}\ sh)\ (\cc{build1}\ k\ (\clam i t))}
    {}
\end{ietarray} \]
\caption{\label{fig:bot-rules-build}
    The rules for the \BOT for $\cc{build1}$, excluding $\cc{build1}$ of $\cc{index}$.
    `$\mathit{op}$' is an elementwise unary or binary arithmetic operator.
}
\end{figure}

\begin{figure}
\[ \begin{ietarray}
\ietrule
    {\cc{index}\ t\ \ixvec{}}
    {t}
    {} \\
\ietrule
    {\cc{index}\ (\cc{index}\ t\ \ixvec{u_1, \ldots, u_m})\ \ixvec{t_1, \ldots, t_n}}
    {\cc{index}\ t\ \ixvec{u_1, \ldots, u_m, t_1, \ldots, t_n}}
    {} \\
\ietrule
    {\cc{index}\ (\cletin{x = v}{t})\ ix}
    {\cletin{x = v}{\cc{index}\ t\ ix}}
    {} \\
\ietrule
    {\cc{index}\ (\cc{cond}\ b\ u\ v)\ \ixvec{t_1,\ldots,t_n}}
    {\multilineT{
        \clet{i_1 = t_1}\ \cinnosp \ldots \clet{i_n = t_n} \\
        \cin \cc{cond}\ b\ \multilineT{
            (\cc{index}\ u\ \ixvec{i_1,\ldots,i_n}) \\
            (\cc{index}\ v\ \ixvec{i_1,\ldots,i_n})
        }
    }}
    {} \\
\ietrule
    {\cc{index}\ (\mathit{op}\ t\ u)\ \ixvec{t_1,\ldots,t_n}}
    {\multilineT{
        \clet{i_1 = t_1}\ \cinnosp \ldots \clet{i_n = t_n} \\
        \cin \mathit{op}\ (\cc{index}\ t\ \ixvec{i_1,\ldots,i_n})\ (\cc{index}\ u\ \ixvec{i_1,\ldots,i_n})
    }}
    {} \\
\ietrule
    {\cc{index}\ (\mathit{op}\ t)\ ix}
    {\mathit{op}\ (\cc{index}\ t\ ix)}
    {} \\
\ietrule
    {\cc{index}\ (\cc{sumOuter}\ t)\ \var{ix}}
    {\cc{sumOuter}\ (\cc{index}\ (\ctransp{1, ..., n, 0}\ t)\ \var{ix})}
    {} \\
\ietrule
    {\cc{index}\ \ravel{t_1, \ldots, t_k}\ \ixvec{u_1, \ldots, u_n}}
    {\multilineT{
        \clet{i_2 = u_2}\ \cinnosp \ldots \clet{i_n = u_n} \\
        \cin \cc{index}\ \multilineT{
            \multilineT{\ravel{
                \cc{index}\ t_1\ \ixvec{i_2,...,i_n}, \ldots, \\
                \qquad \cc{index}\ t_k\ \ixvec{i_2,...,i_n}
            }} \\
            \ixvec{u_1}
        }
    }}
    {$n > 1$} \\
\ietrule
    {\cc{index}\ (\cc{replicate}\ k\ t)\ \ixvec{u_1, \ldots, u_n}}
    {\cc{index}\ t\ \ixvec{u_2, \ldots, u_n}}
    {$n > 0$} \\
\ietrule
    {\cc{index}\ (\ctransp{b_0, \dots, b_k}\ t)\ ix}
    {\cc{index}\ \multilineT{
        (\multilineT{
            \cc{gather}\ sh\ t \\
            \quad (\clam{\ixvec{i_{b_0}, \ldots, i_{b_k}}} \ixvec{i_0, \ldots, i_k}))
        } \\
        ix
    }}
    {} \\
\ietrule
    {\cc{index}\ (\cc{reshape}\ sh\ t)\ ix}
    {\cc{index}\ \multilineT{
        (\multilineT{
            \cc{gather}\ sh\ t\ (\clam {is} \\
            \quad \cc{fromLinearIdx}\ (\cc{shape}\ t) \\
            \quad \quad (\cc{toLinearIdx}\ sh\ is)))
        } \\
        ix
    }}
    {} \\
\ietrule
    {\multilineT{
        \cc{index}\ (\cc{gather}\ (k\ \cc{:::}\ sh)\ t\ (\clam{(i\ \cc{:::}\ is)} ix)) \\
        \hphantom{\cc{index}\ } (u ::: ix')
    }}
    {\cc{index}\ (\cc{gather}\ sh\ t\ (\clam {is} \cletin{i = u}{ix}))\ ix'}
    {$n > 0$} \\
\ietrule
    {\cc{index}\ (\cc{gather}\ \shvec{}\ t\ (\clam{\ixvec{}} ix))\ ix'}
    {\cc{index}\ (\cc{index}\ t\ ix)\ ix'}
    {}
\end{ietarray} \]
\caption{\label{fig:bot-rules-index}
    The rules of the \BOT for $\cc{index}$.
    `$\cc{shape}\ t$' is a macro that expands to the (statically-known) shape of its argument term; `$\cc{toLinearIdx}$' and `$\cc{fromLinearIdx}$' are macros that, respectively, flatten a multidimensional index into a linear one and re-nest it into a multidimensional one.
}
\end{figure}

\begin{figure}
\[ \begin{ietarray}
\ietrule 
    {\cc{build1}\ k\ (\clam i \cc{index}\ x\ ix)}
    {\cc{gather}\ (k\ \cc{:::}\ sh)\ x\ (\clam{\ixvec{i}} ix)}
    {} \vspace{0.5em} \\
\ietrule 
    {\cc{build1}\ k\ (\clam i \cc{index}\ c\ ix)}
    {\cc{gather}\ (k\ \cc{:::}\ sh)\ c\ (\clam{\ixvec{i}} ix)}
    {} \vspace{0.5em} \\
\ietrule
    {\cc{build1}\ k\ (\clam i \cc{index}\ \ravel{t_1, \ldots, t_k}\ \ixvec{t})}
    {\multilineT{
        \cc{gather}\ (k\ \cc{:::}\ sh) \\
        \hphantom{\cc{gather}\ } (\cc{build1}\ k\ (\clam i \ravel{t_1, \ldots, t_k})) \\
        \hphantom{\cc{gather}\ } (\clam{\ixvec{i}} \ixvec{i, t})
    }}
    {} \vspace{0.5em} \\
\ietrule
    {\multilineT{
        \cc{build1}\ k\ (\clam i \\
        \quad \cc{index}\ (\cc{scatter}\ sh\ t\ (\clam {is_2} ix_2))\ ix)
    }}
    {\multilineT{
        \cc{gather}\ (k\ \cc{:::}\ sh) \\
        \hphantom{\cc{gather}\ } (\cc{build1}\ k\ (\clam i \cc{scatter}\ sh\ t\ (\clam {is_2} ix_2))) \\
        \hphantom{\cc{gather}\ } (\clam{\ixvec{i}} i\ \cc{:::}\ ix)
    }}
    {$|ix| > 0$} \\
\end{ietarray} \]
\caption{\label{fig:bot-rules-desperate}
    Rules for the \BOT for $\cc{build1}$-of-$\cc{index}$.
    Recall that $x$ and $c$ refer to variables and constants, respectively.
    Note that all $\cc{index}$-headed forms that do not appear on the left-hand side here, are rewritten away in \cref{fig:bot-rules-index}.
}
\end{figure}

The \BOT is a set of rewrite rules $u \leadsto v$ on the core language; the rules can be found in \cref{fig:bot-rules-build,fig:bot-rules-desperate,fig:bot-rules-index}.
The rules are divided into three categories:
\begin{enumerate}
\item
    \cref{fig:bot-rules-build}: Rules that ``push down'' $\cc{build1}$ into the expression, eventually eliminating it when we reach a subexpression that is elementary enough.
    For example, this is the the rule for binary operators $\mathit{op}$:
    \[ \begin{ietarray}
        \ietrule
            {\cc{build1}\ k\ (\clam i \mathit{op}\ t\ u)}
            {\mathit{op}\ (\cc{build1}\ k\ (\clam i t))\ (\cc{build1}\ k\ (\clam i u))}
            {}
    \end{ietarray} \]
    We see that whenever we build an array elementwise by combining two computations ($t$ and $u$) with a binary operator, this is rewritten to building two arrays containing the results of $t$ and $u$, after which we combine those arrays elementwise.

    When rewriting reaches a leaf expression, e.g.\ some term $t$ that does not mention the index variable $i$:
    \[ \begin{ietarray}
        \ietrule
            {\cc{build1}\ k\ (\clam i t)}
            {\cc{replicate}\ k\ t}
            {$i \not\in \mathit{FV}(t)$}
    \end{ietarray} \]
    we eliminate $\cc{build1}$.

    The careful reader may note that this figure contains rules for $\cc{build1}\ k\ (\clam i t)$ for all term formers $t$, \emph{except} $\cc{index}$.
    The combination $\cc{build1}$-of-$\cc{index}$ is handled in \cref{fig:bot-rules-desperate}, and will be discussed after we discuss the rules for $\cc{index}$ itself.

\item
    \cref{fig:bot-rules-index}: Rules that ``push down'' $\cc{index}$ into the expression.
    In many cases, we can eventually cancel the $\cc{index}$ against a suitable, typically elementwise operation.
    In cases where we cannot, \cref{fig:bot-rules-index} has a missing rule; these are: $\cc{index}$ of a variable reference, of a constant array, of combined arrays $\ravel{t_1, \ldots, t_n}$, and of $\cc{scatter}$.
    These four forms are precisely the normal forms for rewriting $\cc{index}$ listed in \cref{th:bot-normal-form} below.
    To ensure that we can still always eliminate $\cc{build1}$, the final figure (\cref{fig:bot-rules-desperate}) contains rules that commute $\cc{build1}$ below $\cc{index}$ for these four normal forms.

\item
    \cref{fig:bot-rules-desperate}: Rules for commuting $\cc{build1}$ under $\cc{index}$ (turning the $\cc{index}$ into a $\cc{gather}$ simultaneously).
    These belong with the list of $\cc{build1}$-rules from \cref{fig:bot-rules-build}, but are set in a separate figure to make it easier to discuss them separately.

    While these rules preserve semantics and time complexity, the 3rd and 4th rule in this figure do not preserve memory usage under sequential execution.
    For example, when sequentially executing the left-hand side of the fourth rule:
    \[ \cc{build1}\ k\ (\clam i \cc{index}\ (\cc{scatter}\ sh\ t\ (\clam {is_2} ix_2))\ ix) \]
    each $\cc{scatter}$ is executed independently, and its output (which is immediately mostly discarded by the $\cc{index}$) can be deallocated before starting on the next $i$.
    The right-hand side of the rule, however, first computes \emph{all} $\cc{scatter}$s before doing a bulk projection from this big array.

    While unfortunate, the situation is not as bad as it may seem, because when executing the left-hand side in \emph{parallel}, especially on massively parallel hardware like a GPU, many or even all of the $\cc{scatter}$s would be executed in parallel.
    The resulting memory usage is less than the rewritten right-hand side only to the extent that the degree of parallelism is less than $k$.
\end{enumerate}

\paragraph{Strongly normalising}
We can make the behaviour of the rewrite system more formal by looking more precisely at its normal forms.
Indeed, the rewrite system is strongly normalising, meaning that rewriting terminates (for all $u$ there is a $v$, the \emph{normal form} of $u$, such that $u \leadsto^* v$ and $\not\exists t.\ v \leadsto t$) and rewriting order is irrelevant (i.e.\ normal forms are unique).
Therefore, one can talk usefully about this set of normal forms (terms in which no more rewrites are possible), and this set tells us something about the capabilities and limitations of the rewrite system.
\begin{theorem}\label{th:bot-normal-form}
When considering only well-typed terms, the set of normal forms of the rewrite system in \cref{fig:bot-rules-build,fig:bot-rules-desperate,fig:bot-rules-index} consists precisely of those terms $t$ that satisfy the following two properties:
\begin{enumerate}
\item\label{item:bot-prop-nobuild} `$\cc{build1}$' does not occur in $t$.
\item\label{item:bot-prop-noindex}
    Every occurrence of `$\cc{index}\ u\ \var{ix}$' in $t$ is of the form `$\cc{index}\ x\ (v ::: \var{ix})$' (for `$x$' a variable reference), `$\cc{index}\ c\ (v ::: \var{ix})$' (for `$c$' a constant), `$\cc{index}\ \ravel{\ldots}\ \ixvec{v}$', or `$\cc{index}\ (\cc{scatter}\ \_\ \_\ \_)\ (v ::: \var{ix})$'.
\end{enumerate}
\end{theorem}
\begin{proof}
For (\ref{item:bot-prop-nobuild}): assume there is a normal form $t$ that contains $\cc{build1}$; then $t$ contains a subterm $t' \coloneqq \cc{build1}\ k\ (\clam{i} b)$ for some term $b$ where $b$ does \emph{not} contain $\cc{build1}$.
Thus $b$ is headed by one of the other syntactic forms in \cref{fig:core-grammar}, and for each of those (note that if a variable $x$ is unequal to $i$, we certainly have $i \not\in FV(x)$) there is a left-hand side in \cref{fig:bot-rules-build,fig:bot-rules-desperate} that then matches $t'$.
Therefore $t'$, and thus $t$, can be rewritten, contradicting normality of $t$.
Hence, there is no such $t$ after all.

For (\ref{item:bot-prop-noindex}): similarly, assume there is a normal form $t$ that contains a subterm $t' \coloneqq \cc{index}\ u\ \var{ix}$ that does not match any of the stated forms.
By (\ref{item:bot-prop-nobuild}), $u$ does not contain $\cc{build1}$.
Then it can be verified using \cref{fig:core-grammar,fig:core-typing} that one of the left-hand sides of \cref{fig:bot-rules-index} matches $t'$, leading to the same contradiction as before, meaning that there is no such $t$.
\end{proof}

From property (\ref{item:bot-prop-nobuild}) of \cref{th:bot-normal-form} we know that we have successfully eliminated $\cc{build1}$ from the source program by applying the transform.
Property (\ref{item:bot-prop-noindex}) is unfortunately more nuanced, because we cannot always fully eliminate $\cc{index}$.
The upside is that it certainly cannot occur inside $\cc{build1}$ any more --- because there are no more $\cc{build1}$s to occur inside of in the first place.
Furthermore, the only other places in the grammar (\cref{fig:core-grammar}) where a term is executed multiple times are inside the lambda argument to $\cc{gather}$ and $\cc{scatter}$, and because the output type of those lambdas is discrete (namely, an index), their bodies need not be differentiated (see \cref{sec:ad-dual-arrays-terms}), so no one-hots are generated.

Together, this means that the number of one-hots created in the derivative program is at most the number of lexical `$\cc{index}$' occurrences in the source program, which is not too large.


\subsection{Core language design justification}\label{sec:core-language-design-justification}

The \BOT-induced restrictions on the core language listed in \cref{sec:core-language} can be better justified now that we have the \BOT rules in front of us.


\paragraph{Static shapes}
The type system of the core language (\cref{fig:core-typing}) ensures that all array shapes are statically known.
This requirement is a weakening of the actual requirement: ``the \BOT must not get stuck'', or more precisely: all the intermediate values computed in a $\cc{build1}$-lambda must have shapes that are independent of the index at which the lambda is called.

Let us look at an example to see why this requirement exists.
Suppose that the core language contained a primitive, called `$\cc{filter}$', of which the output shape is unknown statically:
\[
    \infrule{\Gamma, x : \rho \vdash s : \Array\ \shvec{}\ \Bool \precsep \Gamma \vdash t : \Array\ \var{sh}\ \rho}
            {\Gamma \vdash \cc{filter}\ (\clam x s)\ t : \Array\ {??}\ \rho}
\]
The semantics is to filter an array on a predicate: $\cc{filter}\ (\clam x x > 4)\ \arrlit{3, 8, -16, 7, 2} = \arrlit{8, 7}$.
Of course, there is no sensible shape to substitute for `$??$' here --- which is the point --- but suppose that we had a weaker type system that allowed this.

The question now is: what does $\cc{build1}\ 10\ (\clam i \cc{filter}\ (\clam x s)\ t)$ vectorise to?

Regardless of what term this would map to, the array that it ought to produce is \emph{not rectangular}: it is not a regular multi-dimensional array, also called a \emph{jagged array}.
Such arrays pose problems with efficient indexing, bounds checking of indexing, semantics of array transposition, etc.
Hence, we disallow such arrays: all our arrays are regular.
This implies that the computation in a `$\cc{build1}$' lambda, \emph{including all its intermediate values}, must have uniform shapes over all values of the index variable $i$.
This is, strictly speaking, a weaker requirement than static shapes, but it is not very much weaker in practice, and static shapes are much easier to enforce for us and to understand for a user.

\paragraph{Strict conditionals}
Consider the \BOT rule for $\cc{build1}$-of-$\cc{cond}$:
\[ \begin{ietarray}
\ietrule
    {\cc{build1}\ k\ (\clam i \cc{cond}\ b\ u\ v)}
    {\cc{build1}\ k\ (\clam i \cc{index}\ \ravel{u,v}\ \ixvec{\cc{cond}\ b\ 0\ 1})}
    {}
\end{ietarray} \]
The reason the conditionals in our language are strict is that this justifies the right-hand side of this translation: it computes the two arguments first, then picks the correct one using $\cc{index}$.
Having proper conditionals would not compose nearly as well with the \BOT as these strict conditionals.

To nevertheless support some algorithms that would otherwise require proper conditionals, it is important that our built-in operations never crash: this permits the user, at least when reasoning semantically, to think of \cc{cond} as a proper conditional.
For example, consider the following program that concatenates an array $a$ (assumed in scope with length 10) to itself:
\[
    \cc{build1}\ 20\ (\clam{i} \cc{cond}\ (i < 10)\ (\cc{index}\ a\ \ixvec{i})\ (\cc{index}\ a\ \ixvec{i - 10}))
\]
After the \BOT, this code will evaluate both \cc{index} expressions for the \emph{full} domain $\{0,\ldots,19\}$ instead of only their intended domain.
Hence, this code only works because our $\cc{index}$ operation, as well as other normally partial operations (such as $\cc{gather}$ and the division operator), check their arguments and still return a value even if the arguments are invalid.

\paragraph{Single expression}
The core language does not admit separate top-level functions: the program must be a single expression with let-bindings.
The \BOT requires this because the generated code for any particular subterm depends on the context in which it runs, all the way to the top level of the program, so all of this context must be \emph{visible} to the code transformation.
Modularity via e.g.\ top-level functions would make this impossible.

This limitation can be ameliorated somewhat by an inlining pass before the \BOT and the AD algorithm proper runs, that eliminates any user-written top-level functions by simply inlining them at every call site.
This may blow up the program significantly in some cases, but note that the trace that AD will generate is on the order of the size of the \emph{fully inlined} program anyway.

\section{Dual Arrays: Differentiating Bulk Array Programs}\label{sec:ad-dual-arrays}

In the existing scalar-level dual-numbers reverse AD algorithm (described in \cref{sec:ad-naive}), each scalar is considered an independent object during differentiation.
We saw in \cref{sec:ad-array-fail} that while this approach can be easily and naturally extended to arrays, it results in very slow gradient code.
As a solution to this problem, we lifted the granularity of the algorithm to entire arrays of scalars (thus creating \emph{dual arrays} instead of ``dual numbers''); the idea here is that a single array operation translates to very many individual scalar operations, and the fewer operations in the program to be differentiated, the lower the overhead introduced by differentiation.

To allow the user to write code that nevertheless works on individual scalars (using $\cc{build1}$ and derived operations such as `map'), the \BOT from \previoussection{sec:bot}{the previous section} rewrites $\cc{build1}$ into bulk array operations with a bulk derivative.

In this section, we start from the output of the \BOT, and explain the dual arrays AD algorithm that we apply to it.
Afterwards, in \cref{sec:restaging}, we will lift the evaluator (the reverse pass) to symbolic tensors to make it possible to differentiate a term once and then compute many different gradients with it.

\begin{figure}
    \begin{center}
    \( \begin{array}{l}
        D[\Array\ \var{sh}\ \R] = (\Array\ \var{sh}\ \R, \tDelta\ \var{sh}) \\
        D[\Array\ \var{sh}\ \Int] = \Array\ \var{sh}\ \Int \\
        D[\Array\ \var{sh}\ \Bool] = \Array\ \var{sh}\ \Bool
    \end{array} \) \\[0.5em]

    \( \begin{array}{l}
        \mathbf{data}\ \tDelta\ \var{sh}\ \mathbf{where} \\
        \quad \multilineT{
            \ccomment{scalar linear maps and input} \\
            \del{Zero} :: \tDelta\ \var{sh} \\  
            \del{Input} :: \tDVarName \to \tDelta\ \var{sh} \\
            \del{Add} :: \tDelta\ \var{sh} \to \tDelta\ \var{sh} \to \tDelta\ \var{sh} \\
            \del{Scale} :: \Array\ \var{sh}\ \R \to \tDelta\ \var{sh} \to \tDelta\ \var{sh} \\
            \ccomment{encoding sharing} \\
            \del{Share} :: \tID \to \tDelta\ \var{sh} \to \tDelta\ \var{sh} \\
            \ccomment{linear array operations} \\
            \del{Index} :: \tDelta\ \shvec{k_1, \ldots, k_n} \to \Ix\ m \to \tDelta\ \shvec{k_{m+1}, \ldots, k_n} \\
            \del{SumOuter} :: \tDelta\ (k ::: \var{sh}) \to \tDelta\ \var{sh} \\
            \del{Gather}\ \multilineT{
                :: \tDelta\ \shvec{k_1, \ldots, k_{m_2}, k_{m_2+1}, \ldots, k_n} \to (\Ix\ m_1 \to \Ix\ m_2) \\
                \to \tDelta\ \shvec{k'_1, \ldots, k'_{m_1}, k_{m_2+1}, \ldots, k_n}
            } \\
            \del{Scatter}\ \multilineT{
                :: \tDelta\ \shvec{k_1, \ldots, k_{m_1}, k_{m_1+1}, \ldots, k_n} \to (\Ix\ m_1 \to \Ix\ m_2) \\
                \to \tDelta\ \shvec{k'_1, \ldots, k'_{m_2}, k_{m_1+1}, \ldots, k_n}
            } \\
            \del{LitArray} :: \Array\ \shvec{k}\ (\tDelta\ \var{sh}) \to \tDelta\ (k ::: \var{sh}) \\
            \del{Replicate} :: \tDelta\ \var{sh} \to \tDelta\ (k ::: \var{sh}) \\
            \ctranspx{\del{Transpose}}{j_1,\ldots,j_m} :: \tDelta\ \shvec{k_1, \ldots, k_n} \to \tDelta\ \shvec{k_{j_1+1}, \ldots, k_{j_m+1}, k_{m+1}, \ldots, k_n} \\
            \del{Reshape} :: \tDelta\ \var{sh} \to \tDelta\ \var{sh'}
        }
    \end{array} \)
    \end{center}
    \caption{\label{fig:array-dn-ad}
        Types for array-level dual-numbers reverse AD.
        $\tDelta\ \var{sh}$ represents the derivative of a term of type $\Array\ \var{sh}\ \R$.
        Slightly modified in \cref{eq:array-delta-sh} on page \pageref{eq:array-delta-sh}.
    }
\end{figure}

\subsection{Types of the transformation}

Recall from \cref{sec:core-language} that the type system of the core language is very simple:
\[ \begin{array}{@{}r@{\;}c@{\;}l@{}}
    \rho &\coloneqqq& \R \mid \Int \mid \Bool \\
    \sigma, \tau &\coloneqqq& \Array\ \var{sh}\ \rho
\end{array} \]
Because of this, the type transformation of the AD algorithm is also very simple (\cref{fig:array-dn-ad}, top).

As witnessed by $D[\Array\ \var{sh}\ \R] = (\Array\ \var{sh}\ \R, \tDelta\ \var{sh})$, the output of this code transformation uses tuples where the input did not.
This is because the algorithm is, in a way, still \emph{dual} numbers reverse AD.
We elide the precise grammar and type system extensions to the core language that allow tuples at the top level (i.e.\ not as elements of arrays!); there are no surprises here.

While the type transformation is simple, we do need to add a number of constructors to the $\tDelta$ data type; we already observed this in \cref{sec:ad-array-fail}.
We can still use $\del{Add}$ and $\del{Scale}$ for the binary and unary arithmetic operators in the language --- where the scaling constant in $\del{Scale}$ becomes array-valued, and the scaling is performed element-wise --- which is why they still appear in $\tDelta$ in \cref{fig:array-dn-ad}.
This way, we avoid a proliferation of $\tDelta$ constructors, one for each broadcasted arithmetic operator.
For the other array operations, however, we generally have a bespoke $\tDelta$ constructor whose semantics is precisely its forward derivative.

For indices into multidimensional arrays, we use the $\Ix$ data type from \cref{sec:ad-array-onehot}:
\[
    \begin{array}{@{}l@{}}
        \mathbf{data}\ \Ix\ k\ \mathbf{where} \\
        \quad \begin{array}[t]{@{}l@{\ }l@{}}
            \IZ &:: \Ix\ 0 \\
            (\IS) &:: \Int \to \Ix\ k \to \Ix\ (k + 1)
        \end{array}
    \end{array}
\]

As an example of how these $\tDelta$ terms correspond to forward derivatives, consider $\del{Index}$.
This $\tDelta$ constructor represents the forward derivative of the core primitive `$\cc{index}$', which has the following typing rule:
\[
    \infrule
        {\Gamma \vdash t : \Array\ \shvec{k_1, \ldots, k_n}\ \rho \precsep
         \jisindex{\Gamma}{\var{ix}}{m} \precsep
         m \leq n}
        {\Gamma \vdash \cc{index}\ t\ \var{ix} : \Array\ \shvec{k_{m+1}, \ldots, k_n}\ \rho}
\]
Because the result of `$\cc{index}\ t\ \var{ix}$' is simply the $\var{ix}$'th element (more accurately, $(n-m)$-dimensional subarray) of $t$, the forward derivative of $\cc{index}\ t\ \var{ix}$ is also simply the $\var{ix}$'th element (subarray) of the forward derivative of $t$.
$\tDelta$ terms are programs that \emph{compute} the forward derivative, so given the $\tDelta$ term computing the forward derivative of $t$, $\del{Index}$ should index that result at this same position $\var{ix}$.
This is precisely the semantics that we give to the $\del{Index}$ constructor.

Note that it is unsurprising that the $\tDelta$ term $\del{Index}$ does precisely the same thing as $\cc{index}$ in the core language; as observed in \cref{sec:ad-array-fail}, this follows from the fact that array indexing is an algebraically linear operation.\footnote{The (forward) derivative is the best linear approximation of a function, so if a function is already linear, then its best linear approximation is itself.}
In fact, all our array operations, except for the broadcasted arithmetic operations, are algebraically linear, thus all $\tDelta$ constructors apart from the basic ones necessary to differentiate primitive arithmetic operations ($\del{Zero}$, $\del{Input}$, $\del{Add}$, $\del{Scale}$ and $\del{Share}$) mirror the semantics of their corresponding core language operation.

\subsection{A non-differentiating transformation}\label{sec:ad-dual-arrays-nondiff}

Before we can move on to the term transformation, we have to introduce yet another code transformation ($\ndD[-]$) that makes a term compatible with surrounding differentiated code, but does not actually differentiate it.
Not differentiating the term in question means that the transformed code need not live inside the monad.
We will use this non-differentiating transformation on the index functions passed to $\cc{gather}$ and $\cc{scatter}$ in the actual differentiating transformation in \cref{sec:ad-dual-arrays-terms}.

\begin{figure}
\[
    x_1 : \tau_1, \ldots, x_n : \tau_n, \Gamma \vdash t : \tau
    \quad \leadsto \quad
    x_1 : D[\tau_1], \ldots, x_n : D[\tau_n], \Gamma \vdash \ndD_\Gamma[t] : \tau
\]
\[ \begin{array}{@{}r@{\ }l@{}}
    \ndD_\Gamma[x] &= \begin{array}[t]{@{}ll@{\ \ }l@{}}
        x & (\text{if } x \in \Gamma) & \ccomment{Internal local variable} \\
        \name{fst}\ x & (\text{else, if } x :: \Array\ \var{sh}\ \R) & \ccomment{Free dual number variable} \\
        x & (\text{otherwise}) & \ccomment{Free discrete variable}
        \vspace{0.1em}
    \end{array} \\
    \ndD_\Gamma[c] &= c \\
    \ndD_\Gamma[\cletin{x = u}{v}] &= \cletin{x = \ndD_\Gamma[u]}{\ndD_{\Gamma,x}[v]} \\
    \ndD_\Gamma[\cc{cond}\ t\ u\ v] &= \cc{cond}\ \ndD_\Gamma[t]\ \ndD_\Gamma[u]\ \ndD_\Gamma[v] \\
    \ndD_\Gamma[\mathit{op}\ u\ v] &= \mathit{op}\ \ndD_\Gamma[u]\ \ndD_\Gamma[v] \\
    \ndD_\Gamma[\mathit{op}\ t] &= \mathit{op}\ \ndD_\Gamma[t] \\
    \ndD_\Gamma[\cc{index}\ t\ \ixvec{t_1, \ldots, t_n}] &= \cc{index}\ \ndD_\Gamma[t]\ \ixvec{\ndD_\Gamma[t_1], \ldots, \ndD_\Gamma[t_n]} \\
    \ndD_\Gamma[\cc{sumOuter}\ t] &= \cc{sumOuter}\ \ndD_\Gamma[t] \\
    \ndD_\Gamma[\cc{gather}\ \var{sh}\ t\ (\clam{\var{is}} \ixvec{t_1, \ldots, t_n})] &= \cc{gather}\ \var{sh}\ \ndD_\Gamma[t]\ (\clam{\var{is}} \ixvec{\ndD_{\Gamma,\var{is}}[t_1], \ldots, \ndD_{\Gamma,\var{is}}[t_n]}) \\
    \ndD_\Gamma[\cc{scatter}\ \var{sh}\ t\ (\clam{\var{is}} \ixvec{t_1, \ldots, t_n})] &= \cc{scatter}\ \var{sh}\ \ndD_\Gamma[t]\ (\clam{\var{is}} \ixvec{\ndD_{\Gamma,\var{is}}[t_1], \ldots, \ndD_{\Gamma,\var{is}}[t_n]}) \\
    \ndD_\Gamma[\ravel{t_1, \ldots, t_n}] &= \ravel{\ndD_\Gamma[t_1], \ldots, \ndD_\Gamma[t_n]} \\
    \ndD_\Gamma[\cc{replicate}\ k\ t] &= \cc{replicate}\ k\ \ndD_\Gamma[t] \\
    \ndD_\Gamma[\ctransp{k_1,\ldots,k_n}\ t] &= \ctransp{k_1,\ldots,k_n}\ \ndD_\Gamma[t] \\
    \ndD_\Gamma[\cc{reshape}\ \var{sh}\ t] &= \cc{reshape}\ \var{sh}\ \ndD_\Gamma[t] \\
    \ndD_\Gamma[\cc{build1}\ k\ (\clam i t)] &= \cc{build1}\ k\ (\clam i \ndD_{\Gamma,i}[t]) \\
\end{array} \]
\caption{\label{fig:non-differentiating-trans}
    The non-differentiating code transformation of \cref{sec:ad-dual-arrays-nondiff}.
}
\end{figure}

The rules can be found in \cref{fig:non-differentiating-trans}.
Note that in contrast to the \emph{differentiating} transformation $D[-]$, the type of the transformed expression is not $\tIdGen\ D[\tau]$ but instead simply $\tau$: only the types of the free variables change.

Because only the \emph{free} variables change type, $\ndD[-]$ needs to take special care to distinguish free variables (bound outside the term initially passed to $\ndD[-]$) from local variables (those bound inside).
In \cref{fig:non-differentiating-trans}, $x_1,\ldots,x_n$ are the free variables and $\Gamma$ contains the locally-bound variables.
The transformation is indexed\footnote{For conciseness, we elide types in the bindings added to $\Gamma$ in the rules for $\cletnosp$, $\cc{gather}$, $\cc{scatter}$ and $\cc{build1}$.} by $\Gamma$ so that the case for variable references $x$, where the actual logic happens, can choose the correct result term depending on whether $x$ is locally bound (and hence not differentiated, so always of the original type) or free (and hence transformed to a dual number if the original $x$ was of type $\Array\ \var{sh}\ \R$).

\subsection{The term transformation}\label{sec:ad-dual-arrays-terms}

We now have all the pieces to extend the (differentiating) code transformation from \cref{sec:ad-naive} to the core array language from \cref{sec:core-language} --- or more precisely, the fragment of the core language that is produced by the \BOT (\cref{sec:bot}): all but $\cc{build1}$.

The term transformation is given in \cref{fig:ad-rules-dual-arrays,fig:ad-rules-dual-arrays-op}.\footnote{
    Note the lack of expressivity of our core language syntax here, as noted in \cref{sec:core-language}: because index lists are not first-class in the language, the expression under the lambda in a $\cc{gather}$ is a list of terms, not a term that produces a list.
    We present the core language this way for simplicity only.
}
The extension generally follows the pattern set out in \cref{sec:ad-naive}, but some aspects benefit from closer examination.

\begin{figure}
\begin{sizeddisplay}{\small}
\[ \begin{array}{@{}r@{\ }l@{\qquad}r@{\ }l@{}}
    D[c]
        &= \mathbf{return}\ (c, \del{Zero})
        &D[x]
        &= \mathbf{return}\ x \\
    D[\cc{cond}\ t\ u\ v]
        &= \mathbf{do}\ \multilineT{
            x \leftarrow D[t] \\
            y \leftarrow D[u] \\
            z \leftarrow D[v] \\
            \mathbf{return}\ (\cc{cond}\ x\ y\ z)
        }
        &D[\cletin{x = u}{v}]
        &= \mathbf{do}\ \multilineT{
            x \leftarrow D[u] \\
            D[v]
        }
\end{array} \]
\begin{align*}
D[\cc{index}\ t\ \ixvec{t_1, \ldots, t_n}]
    &= \multilineT{
        \ccomment{If $t :: \Array\ \var{sh}\ \R$:} \\
        \mathbf{do}\ \multilineT{
            (x, d) \leftarrow D[t] \\
            \mathbf{let}\ i_1 = \ndD_\varepsilon[t_1]; \ldots; i_n = \ndD_\varepsilon[t_n] \\
            \var{id} \leftarrow \mgenid \\
            \mathbf{return}\ (\multilineT{
                \cc{index}\ x\ \ixvec{i_1, \ldots, i_n} \\
                \leadcomma \del{Share}\ \var{id}\ (\del{Index}\ d\ \ixvec{i_1, \ldots, i_n}))
            }
        } \\
        \ccomment{Otherwise (i.e.\ $t :: \Array\ \var{sh}\ \Int$ or $t :: \Array\ \var{sh}\ \Bool$):} \\
        \mathbf{return}\ (\cc{index}\ \ndD_\varepsilon[t]\ \ixvec{\ndD_\varepsilon[t_1], \ldots, \ndD_\varepsilon[t_n]})
    } \\
D[\cc{sumOuter}\ t]
    &= \mathbf{do}\ \multilineT{
        (x, d) \leftarrow D[t] \\
        \var{id} \leftarrow \mgenid \\
        \mathbf{return}\ (\cc{sumOuter}\ x, \del{Share}\ \var{id}\ (\del{SumOuter}\ d)) \\
    } \\
D[\cc{gather}\ \var{sh}\ t\ (\clam{\var{is}} \ixvec{t_1, \ldots, t_n})]
    &= \multilineT{
        \ccomment{If $t :: \Array\ \var{sh}\ \R$:} \\
        \mathbf{do}\ \multilineT{
            (x, d) \leftarrow D[t] \\
            \var{id} \leftarrow \mgenid \\
            \mathbf{return}\ (\multilineT{
                \cc{gather}\ \var{sh}\ x\ (\clam{\var{is}} \ixvec{\ndD_{\var{is}}[t_1], \ldots, \ndD_{\var{is}}[t_n]}) \\
                \leadcomma \del{Share}\ \var{id}\ (\del{Gather}\ d\ (\clam{\var{is}} \ixvec{\ndD_{\var{is}}[t_1], \ldots, \ndD_{\var{is}}[t_n]})))
            }
        } \\
        \ccomment{Otherwise:} \\
        \mathbf{return}\ (\cc{gather}\ \var{sh}\ \ndD_\varepsilon[t]\ (\clam{\var{is}} \ixvec{\ndD_{\var{is}}[t_1], \ldots, \ndD_{\var{is}}[t_n]}))
    } \\
D[\cc{scatter}\ \var{sh}\ t\ (\clam{\var{is}} \ixvec{t_1, \ldots, t_n})]
    &= \ccomment{Elided to save space; analogous to $\cc{gather}$.} \\
D[\ravel{t_1,\ldots,t_n}]
    &= \multilineT{
        \ccomment{If $t_i :: \Array\ \var{sh}\ \R$:} \\
        \mathbf{do}\ \multilineT{
            (x_1, d_1) \leftarrow D[t_1]; \ldots; (x_n, d_n) \leftarrow D[t_n] \\
            \var{id} \leftarrow \mgenid \\
            \mathbf{return}\ (\ravel{x_1, \ldots, x_n}, \del{Share}\ \var{id}\ (\del{LitArray}\ \arrlit{d_1, \ldots, d_n}))
        } \\
        \ccomment{Otherwise:} \\
        \mathbf{return}\ \ravel{\ndD_\varepsilon[t_1], \ldots, \ndD_\varepsilon[t_n]}
    } \\
D[\cc{replicate}\ k\ t]
    &= \multilineT{
        \ccomment{If $t :: \Array\ \var{sh}\ \R$:} \\
        \mathbf{do}\ \multilineT{
            (x, d) \leftarrow D[t] \\
            \var{id} \leftarrow \mgenid \\
            \mathbf{return}\ (\cc{replicate}\ x, \del{Share}\ \var{id}\ (\del{Replicate}\ d))
        } \\
        \ccomment{Otherwise:} \\
        \mathbf{return}\ (\cc{replicate}\ \ndD_\varepsilon[t])
    } \\
D[\ctransp{k_1,\ldots,k_n}\ t]
    &= \multilineT{
        \ccomment{If $t :: \Array\ \var{sh}\ \R$:} \\
        \mathbf{do}\ \multilineT{
            (x, d) \leftarrow D[t] \\
            \var{id} \leftarrow \mgenid \\
            \mathbf{return}\ (\ctransp{k_1,\ldots,k_n}\ x, \del{Share}\ \var{id}\ (\ctranspx{\del{Transpose}}{k_1,\ldots,k_n}\ d))
        } \\
        \ccomment{Otherwise:} \\
        \mathbf{return}\ (\ctransp{k_1,\ldots,k_n}\ \ndD_\varepsilon[t])
    } \\
D[\cc{reshape}\ \var{sh}\ t]
    &= \multilineT{
        \ccomment{If $t :: \Array\ \var{sh}\ \R$:} \\
        \mathbf{do}\ \multilineT{
            (x, d) \leftarrow D[t] \\
            \var{id} \leftarrow \mgenid \\
            \mathbf{return}\ (\cc{reshape}\ \var{sh}\ x, \del{Share}\ \var{id}\ (\del{Reshape}\ \var{sh}\ d))
        } \\
        \ccomment{Otherwise:} \\
        \mathbf{return}\ (\cc{reshape}\ \var{sh}\ \ndD_\varepsilon[t])
    }
\end{align*}
\end{sizeddisplay}
\caption{\label{fig:ad-rules-dual-arrays}
    The AD transformation on the core array language, except for arithmetic operations $\mathit{op}$.
}
\end{figure}

\begin{figure}
\begin{align*}
D[t_1 \times_{\Array\ \var{sh}\ \R} t_2]
    &= \mathbf{do}\ \multilineT{
        (x_1, d_1) \leftarrow D[t_1]; (x_2, d_2) \leftarrow D[t_2] \\
        \var{id} \leftarrow \mgenid \\
        \mathbf{return}\ (x_1 \times_{\Array\ \var{sh}\ \R} x_2, \del{Share}\ \var{id}\ (\del{Add}\ (\del{Scale}\ x_2\ d_1)\ (\del{Scale}\ x_1\ d_2)))
    } \\
D[t_1 \times_{\Array\ \var{sh}\ \Int} t_2]
    &= \mathbf{return}\ (\ndD_\varepsilon[t_1] \times_{\Array\ \var{sh}\ \Int} \ndD_\varepsilon[t_2]) \\
& \hspace{-1cm} \ccomment{etc.\ other broadcasted arithmetic operations on arrays}
\end{align*}
\caption{\label{fig:ad-rules-dual-arrays-op}
    AD transformation on the core language for arithmetic operations.
    Analogous to the arithmetic operations in \cref{fig:dn-ad-terms}.
    Completes \cref{fig:ad-rules-dual-arrays}.
}
\end{figure}

\begin{itemize}
\item
    In the core language, we have operations that apply both\footnote{They are ``polymorphic'', albeit not by explicit polymorphism in the language, but instead by a custom typing rule.} to the \emph{dualised type} ($\Array\ \var{sh}\ \R$) and to non-dualised types (arrays of non-$\R$ elements), and yet they look at the internal structure of their argument.
    In the simple language in \cref{sec:ad-naive}, we did not have such constructs: a construct either monomorphically worked on scalars (e.g.\ $(\times_\R)$, which got a derivative that works on dual numbers specifically) or kept the values of the scalars as-is (e.g.\ `$\name{fst}$', pairing or lambda-abstraction, each of which got derivatives oblivious of the existence of dual numbers).
    The presence of operations in the core language that mix the two (e.g.\ $\cc{gather}$, $\cc{replicate}$) means that their derivative under $D[-]$ necessarily differs depending on whether they work on an array of scalars or not.

    Notable is that the derivative of $\cc{sumOuter}$ in \cref{fig:ad-rules-dual-arrays} does not need multiple different versions, because it is monomorphic: it works only on arrays of scalars.

\item
    Despite the fact that e.g.\ the terms $t_1, \ldots, t_n$ in `$\cc{index}\ t\ \ixvec{t_1, \ldots, t_n}$' are of type $\Array\ \var{sh}\ \Int$ and that $D[\Array\ \var{sh}\ \Int] = \Array\ \var{sh}\ \Int$, we cannot simply use those $t_1, \ldots, t_n$ in the differentiated program as-is: the types of their free variables are wrong.
    While using $D[t_i]$ instead would do the trick, we use $\ndD[t_i]$ to avoid potentially building $\tDelta$ terms that will only be discarded later.
    In the non-scalar version of $D[\cc{index}]$, even the array being indexed is not dualised, so we can convert that term using $\ndD[-]$ too.
    (`$\varepsilon$' denotes the empty environment.)

    The same holds for the terms $t_1, \ldots, t_n$ in `$\cc{gather}\ \var{sh}\ t\ (\clam{\var{is}} \ixvec{t_1,\ldots,t_n})$', but there the usage of $\ndD[t_i]$ is essential: $D[t_i]$ would run in the $\tIdGen$ monad, and $\cc{gather}$ and $\del{Gather}$ expect a non-monadic function.

\item
    The reader might wonder: if avoiding elementwise scalar computation in $\cc{build1}$ is what we did the whole \BOT for, why is the elementwise computation in $\cc{gather}$ fine (despite the fact that it may indeed contain scalar computation too, if the results are subsequently converted back to integers!)?
    The answer is the same as for why we could use $\ndD[-]$ for those terms: any scalar computation that happens inside the function passed to $\cc{gather}$ cannot \emph{continuously} influence the final program result (because it can only influence said result through the discrete, integral results of that function), so it does not need to be differentiated.
    Hence this computation does not end up as sub-traces in $\tDelta$, and the $\tDelta$ explosion problem of \cref{sec:ad-array-fail-scalar} does not arise.

\item
    Finally, the arithmetic operations in \cref{fig:ad-rules-dual-arrays-op} generalise straight from \cref{fig:dn-ad-terms} in \cref{sec:ad-naive}, with operations on $\Array\ \var{sh}\ \R$ getting differentiated and operations on arrays with discrete elements being fully preserved.
\end{itemize}

\subsection{The Evaluator}\label{sec:ad-dual-arrays-evaluator}

In \cref{sec:ad-naive} and in \cite[Fig.\ 11]{2022-krawiec-dualrev}, the reverse pass (evaluation of $\tDelta$ terms) works on a state represented by $\Map$ values keyed by $\tDVarName$ and $\tID$.
This was admissible because every input, as well as every node ID, corresponded to a single scalar, and hence the maps were homogeneous.
Now that $\tDVarName$ and $\tID$ correspond to entire arrays of scalars (which may differ in shape and hence in type), the evaluation state needs to contain \emph{heterogeneous} maps.
To retain type safety, we use \emph{dependent maps}: if normal maps (associative arrays) can be seen as a collection of pairs, then a dependent map is a collection of \emph{dependent} pairs.
The types of the methods on a dependent map, as far as we use them, are given in \cref{fig:dmap-types}.

\begin{figure}
\[ \begin{array}{@{}l@{}}
    \mathbf{data}\ \tDMap\ f\ g \\
    \dmEmpty :: \tDMap\ f\ g \qquad\ccomment{shorthand: \textup{`$\{\}$'}} \\
    \dmInsert :: \texttt{GCompare}\ f \Rightarrow f\ a \to g\ a \to \tDMap\ f\ g \to \tDMap\ f\ g \\
    \dmLookup :: \texttt{GCompare}\ f \Rightarrow f\ a \to \tDMap\ f\ g \to \text{Maybe}\ (g\ a) \\
    (\dmIndex) :: \texttt{GCompare}\ f \Rightarrow f\ a \to \tDMap\ f\ g \to g\ a \qquad\ccomment{partial version of $\dmLookup$} \\
    \dmDelete :: \texttt{GCompare}\ f \Rightarrow f\ a \to \tDMap\ f\ g \to \tDMap\ f\ g \\
    \dmInsertWith :: \texttt{GCompare}\ f \multilineT{
        {} \Rightarrow (g\ a \to g\ a \to g\ a) \\
        {} \to f\ a \to g\ a \to \tDMap\ f\ g \to \tDMap\ f\ g
    } \\
    \dmMaxViewWithKey :: \tDMap\ f\ g \to \texttt{Maybe}\ (\exists a.\,(f\ a, g\ a), \tDMap\ f\ g)
\end{array} \]
\caption{\label{fig:dmap-types}
    Types of methods on $\tDMap$, as provided by \url{https://hackage.haskell.org/package/dependent-map}.
    (The existential in the type of `$\dmMaxViewWithKey$' is encoded as a newtype in Haskell.)
}
\end{figure}

The point of a dependent map is to be able to map a type-indexed key to a type-indexed value.
Our values (arrays and $\tDelta$ terms) are shape-indexed and will, at least in the reverse pass, always contain scalars, hence shape-indexing is sufficient.
Previously, in the reverse pass for the scalar-level algorithm in \cref{sec:ad-naive} (\cref{fig:ad-efficient-eval}), the evaluation state looked as follows:
\[ \begin{array}{@{}l@{}}
    \mathbf{data}\ \EState = \EState \\
    \quad \{ \begin{array}[t]{@{}l@{\quad}l@{}}
        \ \name{grad} :: \Map\ \tDVarName\ \R & \ccomment{input cotangents: will collect final gradient} \\
        \leadcomma\ \name{dfrag} :: \Map\ \tID\ \tDelta & \ccomment{delta fragments} \\
        \leadcomma\ \name{accum} :: \Map\ \tID\ \R\ \} & \ccomment{accumulated node cotangents} \\
    \end{array}
\end{array} \]
We see that our map keys are $\tDVarName$ (for input values) and $\tID$ (for intermediate nodes in the graph).
These types gain a type parameter indicating the shape of the array they refer to:
\begin{equation}
\begin{array}{@{}l@{}}
    \mathbf{data}\ \tDelta\ \var{sh}\ \mathbf{where} \\
    \quad \del{Input} :: \tDVarName\ \textcolor{red}{\var{sh}} \to \tDelta\ \var{sh} \\
    \quad \del{Share} :: \tID\ \textcolor{red}{\var{sh}} \to \tDelta\ \var{sh} \to \tDelta\ \var{sh} \\
    \quad \ccomment{... other constructors ...}
\end{array}
\label{eq:array-delta-sh}
\end{equation}
And the evaluation state changes accordingly:\footnote{
    For operations on these $\tDMap$s to typecheck, $\tID$ and $\tDVarName$ must implement the $\texttt{GCompare}$ type class.
    They can do so if they, in addition to the integer ID/name itself, also contain a singleton representing the type index (i.e. the rank) on the value level.
    Key equality checking will then also compare the singletons for equality.
    Because a list of \texttt{GHC.TypeLits.SNat}s (themselves simply integers) can suffice for these singletons, this need not be very expensive, but if the overhead is still unacceptable, the only thing that removing the singleton (and using an unsafe coercion to conjure up the type equality evidence) compromises is \emph{confidence} in the algorithm's type-correctness --- if two $\tID$s have equal identifiers but types, that is a bug in our code.
}
\[ \begin{array}{@{}l@{}}
    \mathbf{data}\ \EState = \EState \\
    \quad \{ \begin{array}[t]{@{}l@{}}
        \ \name{grad} :: \tDMap\ \tDVarName\ (\lambda \var{sh}.\ \Array\ \var{sh}\ \R) \\
        \leadcomma\ \name{dfrag} :: \tDMap\ \tID\ \tDelta \\
        \leadcomma\ \name{accum} :: \tDMap\ \tID\ (\lambda \var{sh}.\ \Array\ \var{sh}\ \R)\ \} \\
    \end{array}
\end{array} \]
Note our use of a type-level lambda here; in Haskell, this needs to be encoded using a newtype.

With this updated typing, we can update the reverse pass for arrays.
The $D[-]$ code transformation gave the \emph{forward derivatives} of our language constructs, expressed in the language of $\tDelta$ terms.
The reverse pass now has the task of transposing these forward derivatives to \emph{reverse derivatives}.
First, we change the type of $\reversePass$ to work with arrays instead of scalars (parts changed from \cref{fig:ad-efficient-eval} highlighted in \textcolor{red}{red}):
\[ \begin{array}{@{}l@{}}
    \reversePass :: \textcolor{red}{\Array\ \var{sh}}\ \R \to \tDelta\ \textcolor{red}{\var{sh}} \to \textcolor{red}{\tDMap}\ \tDVarName\ \textcolor{red}{(\lambda \var{sh}.\ \Array\ \var{sh}}\ \R\textcolor{red}{)} \\
    \reversePass\ c\ d = \name{grad}\ (\backprop\ (\eval\ c\ d\ (\EState\ \{\}\ \{\}\ \{\})))
\end{array} \]
taking an array-valued incoming cotangent to be immediately passed to $\eval$.

The definition of $\backprop$ from \cref{fig:ad-efficient-eval} remains unchanged except for a simple textual substitution of ``$\Map$'' to ``$\tDMap$''.
The major change is in $\eval$, which now has to (1) deal with array values, and (2) handle more $\tDelta$ constructors than before.
The result is shown in \cref{fig:ad-array-eval}; let us discuss some important aspects.

\begin{figure}
\[ \begin{array}{@{}l@{}}
    \eval :: \Array\ \var{sh}\ \R \to \tDelta\ \var{sh} \to \EState \to \EState \\
    \begin{array}{@{}l@{\ }c@{\ }l@{}}
        \eval\ c\ \del{Zero} &=& \id \\
        \eval\ c\ (\del{Input}\ v) &=& \clam{s} s\ \{\ \name{grad} = \name{DMap.insertWith}\ (\vecadd)\ v\ c\ (\name{grad}\ s)\ \} \\
        \eval\ c\ (\del{Add}\ d_1\ d_2) &=& \eval\ c\ d_2 \circ \eval\ c\ d_1 \\
        \eval\ c\ (\del{Scale}\ \mathit{arr}\ d) &=& \eval\ (c \vecmul \mathit{arr})\ d \\
        \eval\ c\ (\del{Share}\ i\ d) &=& \clam{s} s\ \{ \multilineT{
            \ \name{dfrag} = \name{DMap.insert}\ i\ d\ (\name{dfrag}\ s) \\
            \leadcomma\ \name{accum} = \name{DMap.insertWith}\ (\vecadd)\ i\ c\ (\name{accum}\ s)\ \} \\
        } \\
        \eval\ c\ (\del{Index}\ d\ i) &=& \eval\ (\cc{oneHot}\ (\name{shapeDelta}\ d)\ i\ c) \\
        \eval\ c\ (\del{SumOuter}\ d) &=& \eval\ (\cc{replicate}\ (\name{head}\ (\shapeDelta\ d))\ c)\ d \\
        \eval\ c\ (\del{Gather}\ d\ f) &=& \eval\ (\cc{scatter}\ (\shapeDelta\ d)\ c\ f)\ d \\
        \eval\ c\ (\del{Scatter}\ d\ f) &=& \eval\ (\cc{gather}\ (\shapeDelta\ d)\ c\ f)\ d \\
        \eval\ c\ (\del{LitArray}\ \var{ds}) &=& \multilineT{
            \clet \shvec{n} = \cc{shape}\ \var{ds} \\
            \cin \multilineT{
                \eval\ (\cc{index}\ c\ \ixvec{n-1})\ (\cc{index}\ \var{ds}\ \ixvec{n-1}) \\
                \quad \circ \cdots \circ \eval\ (\cc{index}\ c\ \ixvec{0})\ (\cc{index}\ \var{ds}\ \ixvec{0})
            }
        } \\
        \eval\ c\ (\del{Replicate}\ d) &=& \eval\ (\cc{sumOuter}\ c)\ d \\
        \eval\ c\ (\ctranspx{\del{Transpose}}{j_1,\ldots,j_m}\ d) &=& \eval\ (\ctransp{\text{inversePermutation}(j_1,\ldots,j_m)}\ c)\ d \\
        \eval\ c\ (\del{Reshape}\ d) &=& \eval\ (\cc{reshape}\ (\name{shapeDelta}\ d)\ c)\ d
    \end{array}
\end{array} \]
\caption{\label{fig:ad-array-eval}
    The evaluator in the reverse pass for the array AD algorithm.
}
\end{figure}

When comparing \cref{fig:ad-array-eval} with \cref{fig:ad-efficient-eval}, we see that the incoming cotangent $c$ is now array-valued, just like the cotangent accumulators in the evaluation state and the scaling constant in $\del{Scale}$.
Hence, the addition and scaling operators from $\eval_3$ in \cref{fig:ad-efficient-eval} have to be broadcasted (i.e.\ evaluated elementwise) on arrays now; this we indicate with $\vecadd$ and $\vecmul$.

\newcommand\matrowvec[1]{\bigl(\,#1\,\bigr)}
In the cases for the array operations, we see that $\eval$ computes the linear transpose of the forward derivatives described by the $\tDelta$ term, and applies them to the $c$ argument.
The transpose of $\del{SumOuter}$ (summation, matrix $\matrowvec{1 \ \  \cdots \ \  1}$) is replication (matrix $\matrowvec{1 \ \  \cdots \ \  1}^\top$); the transpose of a gather is a scatter.
This was in fact already true of the existing $\tDelta$ constructors, but it was somewhat less obvious: the transpose of addition (matrix $\matrowvec{1 \ \  1}$) is duplication (matrix $\matrowvec{1 \ \  1}^\top$), and the transpose of scaling (matrix $\matrowvec{r}$) is scaling (matrix $\matrowvec{r}^\top = \matrowvec{r}$).

In order to compute the transposes of many of the array operations, we need the shape of the original argument arrays to the operation (conveniently equal to the shape of their forward derivative, as computed by the $\tDelta$ term arguments).
For example, in $\eval\ c\ (\del{SumOuter}\ d)$, we need the shape of the array computed by $d$ in order to $\cc{replicate}$ up the incoming cotangent to a cotangent appropriate for reverse-evaluating $d$.
This shape can be computed by inspection of the $\tDelta$ term; we implement this in a function `$\shapeDelta$', the implementation of which we elide.
This `$\shapeDelta$' function can be made constant-time by caching well-chosen explicit shape vectors in $\tDelta$ constructors.

Finally, we use a few array operations that are not strictly speaking in the core language.
The $\cc{shape}$ function was already used in \cref{fig:bot-rules-index}, and returns the shape of its argument array.
The $\cc{oneHot}$ operation produces an array with one specified entry having a particular value, and zeros elsewhere; it could be implemented as follows:
\[
    \cc{oneHot}\ \var{sh}\ \var{idx}\ x =
        \cc{gather}\ \var{sh}\ \ravel{0, x}\ (\clam{\var{idx'}} \cc{cond}\ (\var{idx} = \var{idx'})\ \ixvec{1}\ \ixvec{0})
\]
but likely benefits from a specialised implementation.

\subsection{Wrapper}\label{sec:ad-dual-arrays-wrapper}

As before in \cref{sec:ad-naive}, this reverse AD algorithm for array programs needs a wrapper to be useful.
Compared to the wrapper described in \cref{sec:ad-naive-wrapper}, the one for the array AD algorithm is mostly identical; the only wrinkle is that we cannot simplify by assuming the program to differentiate has only one argument (because the core language does not support tuples).
The type of the wrapper thus becomes the following:
\[
    \name{wrapper}\ \multilineT{
        :: \tAST\ (\Array\ \shvec{}\ \R) \quad\ccomment{free variables: $x_1 : \Array\ \var{sh}_1\ \rho_1, \ldots, x_n : \Array\ \var{sh}_n\ \rho_n$} \\
        \to (\Array\ \var{sh}_1\ \rho_1, \ldots, \Array\ \var{sh}_n\ \rho_n) \\
        \to \Array\ \shvec{}\ \R \\
        \to (\Array\ \var{sh}_1\ \rho_1, \ldots, \Array\ \var{sh}_n\ \rho_n)
    }
\]
The implementation is completely analogous to the one in \cref{sec:ad-naive-wrapper}.

\section{Compile-Time Differentiation}\label{sec:restaging}

\TODO{This section needs an example that exhibits that real fusion opportunities arise when optimising the differentiated program.}

This wrapper obtained in \cref{sec:ad-dual-arrays-wrapper} takes a term, but returns a \emph{numeric} gradient.
In particular, it does this by passing the input to the transformed version of the source program, and \emph{interpreting} (with $\reversePass$) the resulting $\tDelta$ term.
Typically, one requires a gradient at many different input points, and with the current setup, this results in interpretation overhead from $\reversePass$ for every such input point.
In this section, we improve upon this by extracting the $\tDelta$ term from the \emph{unevaluated} transformed program, and evaluating (transposing) this term symbolically.

We first show an example (\cref{sec:restaging-example}) that illustrates the pipeline so far, and introduces the observation that allows us to symbolically evaluate the $\tDelta$ term in a forward-differentiated program without having to provide an input first.
Then, in \cref{sec:restaging-symbolic-eval-fail,sec:restaging-symbolic-eval,sec:restaging-unglobal}, we show how to symbolically evaluate this extracted $\tDelta$ term, yielding in \cref{sec:restaging-wrapper} a full program that computes the gradient in one go without relying on any $\tDelta$ interpreter.
Finally, we show in \cref{sec:restaging-induction} that the $\tDelta$ extraction observed in the example works in general, by an induction argument over the syntax of the core language.

\subsection{Example}\label{sec:restaging-example}

\newcommand\texSc{t_{\text{sc}}}

At the top of \cref{fig:restaging-example}, we give a simple term $\texSc$ that multiplies an array $a :: \Array\ \shvec{n}\ \R$ elementwise with its reverse and sums the result, thus computing one element of $a$'s self-convolution.
Note that a zero-dimensional array contains exactly one value, hence $\texSc$ is suitable for reverse differentiation.

\begin{figure}
{\textbf{Source term:}\hfill}\vspace{-2pt}
\[ \begin{array}{@{}l@{}}
    a : \Array\ \shvec{n}\ \R \\
    \vdash \texSc = \cc{sumOuter}\ (\cc{build1}\ n\ (\clam{i} \cc{index}\ a\ i \times_R \cc{index}\ a\ (n - 1 - i))) \\
    \quad : \Array\ \shvec{}\ \R
\end{array} \]
{\textbf{After \BOT:}\hfill}\vspace{-2pt}
\[
    \texSc' = \cc{sumOuter}\ \bigl(\cc{gather}\ \shvec{n}\ a\ (\clam{\ixvec{i}} \ixvec{i}) \times_\R \cc{gather}\ \shvec{n}\ a\ (\clam{\ixvec{i}} \ixvec{n - 1 - i})\bigr)
\]
{\textbf{After forward differentiation:}\hfill}\vspace{-2pt}
\[
    D[\texSc'] = \mathbf{do}\ \multilineT{
        (x, d) \leftarrow \multilineT{
            \mathbf{do}\ \multilineT{
                (x_1, d_1) \leftarrow \mathbf{do}\ \multilineT{
                    (y_1, dy_1) \leftarrow \mathbf{return}\ a \\
                    \var{id}_1 \leftarrow \mgenid \\
                    \mathbf{return}\ (\multilineT{
                        \cc{gather}\ \shvec{n}\ y_1\ (\clam{\ixvec{i}} \ixvec{i}) \\
                        \leadcomma \del{Share}\ \var{id}_1\ (\del{Gather}\ dy_1\ (\clam{\ixvec{i}} \ixvec{i})))
                    }
                } \\
                (x_2, d_2) \leftarrow \mathbf{do}\ \multilineT{
                    (y_2, dy_2) \leftarrow \mathbf{return}\ a \\
                    \var{id}_2 \leftarrow \mgenid \\
                    \mathbf{return}\ (\multilineT{
                        \cc{gather}\ \shvec{n}\ y_2\ (\clam{\ixvec{i}} \ixvec{n - 1 - i}) \\
                        \leadcomma \del{Share}\ \var{id}_2\ (\del{Gather}\ dy_2\ (\clam{\ixvec{i}} \ixvec{n - 1 - i})))
                    }
                } \\
                \var{id}_3 \leftarrow \mgenid \\
                \mathbf{return}\ (x_1 \times_{\Array\ \var{sh}\ \R} x_2, \del{Share}\ \var{id}_3\ (\del{Add}\ (\del{Scale}\ x_2\ d_1)\ (\del{Scale}\ x_1\ d_2)))
            }
        } \\
        \var{id}_4 \leftarrow \mgenid \\
        \mathbf{return}\ (\cc{sumOuter}\ x, \del{Share}\ \var{id}_4\ (\del{SumOuter}\ d))
    }
\]
{\textbf{After simplification:}\hfill}\vspace{-2pt}
\[
    D[\texSc'] = \mathbf{do}\ \multilineT{
        \mathbf{let}\ (a_p, a_d) = a \\
        \mathbf{let}\ x_1 = \cc{gather}\ \shvec{n}\ a_p\ (\clam{\ixvec{i}} \ixvec{i}) \\
        \mathbf{let}\ x_2 = \cc{gather}\ \shvec{n}\ a_p\ (\clam{\ixvec{i}} \ixvec{n - 1 - i}) \\
        \var{id}_1 \leftarrow \mgenid;
        \var{id}_2 \leftarrow \mgenid;
        \var{id}_3 \leftarrow \mgenid;
        \var{id}_4 \leftarrow \mgenid \\
        \mathbf{return}\ (\multilineT{
            \cc{sumOuter}\ (x_1 \times_{\Array\ \var{sh}\ \R} x_2) \\
            \leadcomma \del{Share}\ \var{id}_4\ (\del{SumOuter}\ (\del{Share}\ \var{id}_3 \\
            \quad (\del{Add}\ \multilineT{
                (\del{Scale}\ x_2\ (\del{Share}\ \var{id}_1\ (\del{Gather}\ a_d\ (\clam{\ixvec{i}} \ixvec{i})))) \\
                (\del{Scale}\ x_1\ (\del{Share}\ \var{id}_2\ (\del{Gather}\ a_d\ (\clam{\ixvec{i}} \ixvec{n - 1 - i})))))
            }
        }
    }
\]
\caption{\label{fig:restaging-example}
    Example of the full pipeline so far.
}
\end{figure}

First, we vectorise the term by passing it through the \BOT.
The $\cc{sumOuter}$ in $\texSc$ remains as-is, because it is not enclosed in $\cc{build1}$ or $\cc{index}$.
The $\cc{build1}$ part of $\texSc$ gets vectorised: indexing turns into $\cc{gather}$ and the scalar primitive operation $(\times_\R)$ now operates on arrays, notated $(\times_{\Array\ \var{sh}\ \R})$.
Of course, `$\cc{gather}\ \shvec{n}\ a\ (\clam{\ixvec{i}} \ixvec{i})$' can be simplified to `$a$', and an implementation should perform this optimisation, but we keep the term as-is throughout this section to illustrate the general case.

For (forward) differentiation, we take $\texSc$ through the code transformation in \cref{fig:ad-rules-dual-arrays,fig:ad-rules-dual-arrays-op}.
The result is monadic code in the $\tIdGen$ monad, and is a single term $D[\texSc']$ of the following type:
\[
    a : (\Array\ \shvec{n}\ \R, \tDelta\ \shvec{n})
    \vdash D[\texSc'] : \tIdGen\ (\Array\ \shvec{}\ \R, \tDelta\ \shvec{})
\]
Internally, the term uses array operations from the core language ($\cc{gather}$, $\cc{sumOuter}$, $\times_{\Array\ \var{sh}\ \R}$) as well as pairs, let-binding, and monadic operations.

\paragraph{Simplification}
This term could of course benefit from some simplification, most particularly using the monad laws.
Formulated in $\mathbf{do}$-notation, the properties about monads that we use are the following:
\[ \begin{array}{@{}l@{\ =\ }ll@{}}
    (\mathbf{do}\ x \leftarrow \mathbf{return}\ E_1; E_2)
        & (\mathbf{do}\ \mathbf{let}\ x = E_1; E_2)
        & \textit{\small(left identity)} \\
    (\mathbf{do}\ y \leftarrow (\mathbf{do}\ x \leftarrow E_1; E_2); E_3)
        & (\mathbf{do}\ x \leftarrow E_1; y \leftarrow (\mathbf{do}\ E_2); E_3)
        & \textit{\small(associativity)} \\
    (\mathbf{do}\ y \leftarrow (\mathbf{do}\ \mathbf{let}\ x = E_1; E_2); E_3)
        & (\mathbf{do}\ \mathbf{let}\ x = E_1; y \leftarrow (\mathbf{do}\ E_2); E_3)
        & \textit{\small(by def.\ of do-notation)}
\end{array} \]
with $E_i$ standing for arbitrary code.
Additionally, we reorganise some let-bindings, and contract let-bindings with the same right-hand side (specifically: ``$a$''):
\[
    (\mathbf{let}\ x = E\ \mathbf{in}\ \mathbf{let}\ y = E\ \mathbf{in}\ \ldots x \ldots y \ldots)
    = (\mathbf{let}\ x = E\ \mathbf{in}\ \ldots x \ldots x \ldots)
\]
The result is the simplified term at the bottom of \cref{fig:restaging-example}.

Note that because the original term ($\texSc$) did not have any let-bindings, its trace (the $\tDelta$ term) does not have any shared subterms, and hence the $\del{Share}$ nodes in $D[\texSc']$ actually turn out to be unnecessary.
In general, however, they could be needed if the source program used the result of certain subterms multiple times.

\paragraph{Evaluation}
To compute a gradient with this differentiated term at a certain input $x : \Array\ \shvec{n}\ \R$, run $D[\texSc']$ in the environment $\{a \coloneqq (x, \del{Input}\ \var{var})\}$ for some $\var{var} : \tDVarName\ \shvec{n}$, and pass the resulting $\tDelta$ term (together with the initial cotangent $\arrlit{1.0} : \Array\ \shvec{}\ \R$) to $\reversePass$ from \cref{sec:ad-dual-arrays}.

The $\reversePass$ function calls $\backprop$, which evaluates the $\tDelta$ term in the second component of the result of $D[\texSc']$ in reverse, starting with its topmost $\del{Share}\ \var{id}_4$ node.
The initial cotangent $\arrlit{1.0}$ arrives in $\eval$ (\cref{fig:ad-array-eval}) via the `$\name{accum}$' map of the evaluation state (at key $\var{id}_4$).
Evaluation then reverse-evaluates $\del{SumOuter}$, replicating up the initial cotangent to an $n$-element array which, again via the `$\name{accum}$' map, gets contributed to the $\del{Add}$ node, which propagates it on to both $\del{Scale}$ nodes.
There the array gets elementwise-multiplied with the primal arrays $x_2$ and $x_1$ before $\eval$ of $\del{Gather}$ uses a `$\name{scatter}$' operation to invert the index mappings and construct the two contributions to the gradient with respect to $a$ in `$\name{grad}$'.
These are added together using $\vecadd$ in the $\dmInsertWith$ call in $\eval$ of $\del{Input}$.

Finally, the final gradient of type $\Array\ \shvec{n}\ \R$ is at the $\var{var}$ key of the $\tDMap$ returned from $\reversePass$.
Because the array operations in $\texSc'$ operate on scalar arrays in bulk, the $\tDelta$ trace contains only a few nodes, and evaluation overhead is limited.

\paragraph{Separated $\tDelta$ term}
A curious thing has happened when simplifying $D[\texSc']$ to the form at the bottom of \cref{fig:restaging-example}: the $\tDelta$ term is fully extracted from the primal computation.
The simplified $D[\texSc']$ has the following structure:
\newcommand\eqRestagingForm{
    \begin{array}{@{}l@{}} 
    \mathbf{do}\ \multilineT{
        \mathbf{let}\ \textit{most of the primal computation} \\
        (\var{id}_1, \ldots, \var{id}_n) \leftarrow \textit{generate $\tID$s} \\
        \mathbf{return}\ (\multilineT{
            \textit{result of the primal computation} \\
            \leadcomma \textit{symbolic $\tDelta$ term})
        }
    }
    \end{array}
}
\begin{equation}
    \eqRestagingForm
    \label{eq:restaging-form}
\end{equation}
In fact, it turns out that $D[t]$ has this structure (after simplification) for \emph{any} term $t$ in our core language!
This is good news, because with a separated-out $\tDelta$ term, we can symbolically evaluate that $\tDelta$, put the resulting symbolic gradient back in replacing the dual component of the result, and obtain a single term that computes a gradient without the need of any interpreter.
In \cref{sec:restaging-induction} we (informally) prove that we can always simplify $D[t]$ to the form in \cref{eq:restaging-form}, but for now let us simply assume that the example in \cref{fig:restaging-example} is representative.

Let us look more closely at the structure of the dual component of the result: ``\textit{symbolic $\tDelta$ term}'' in \cref{eq:restaging-form}.
It is of course not actually a value of type $\tDelta\ \var{sh}$ for some $\var{sh}$, but instead a term (call it $t_d$) in an extension of our core language that \emph{produces} a $\tDelta$ term.
However, as we can see in \cref{fig:restaging-example}, $t_d$ in fact contains almost only $\tDelta$ constructors; the exceptions are as follows:
\begin{enumerate}
\item
    The $\tID$ field in a $\del{Share}$ constructor is a variable reference to one of the $\mgenid$ results.
\item
    The first field of a $\del{Scale}$ constructor (the scaling constant) is a variable reference into the first $\mathbf{let}$-part of the form.
    This is valid for the rules in \cref{fig:ad-rules-dual-arrays-op}, but if there were a $D[-]$ rule for some primitive operation that puts a non-trivial term in the first field of $\del{Scale}$, this can be easily reduced to a variable reference by let-binding that term first.
\item
    Index values (in $\del{Index}$) and index functions (in $\del{Gather}$ and $\del{Scatter}$) are terms instead of concrete values.
\item
    The $\tDelta$ terms for the inputs, in \cref{fig:restaging-example} just $a_d$, may appear instead of a $\tDelta$ constructor term.
\end{enumerate}
Exceptions 1--3 concern positions in the $\tDelta$ data type where non-$\tDelta$ values are embedded, so it is to be expected that they ``escape'' from the strict form of a term with only $\tDelta$ constructors.
Exception 4 is for inputs, and if we think back to the wrapper (\cref{sec:ad-naive-wrapper}), we already know that these will become $\del{Input}$ constructors.
Hence we can fill in the exception-4 variable references with the appropriate $\del{Input}$ terms, and leave just exceptions 1--3.

\subsection[Symbolic evaluation of Delta: without proper sharing]{Symbolic evaluation of $\tDelta$: without proper sharing}\label{sec:restaging-symbolic-eval-fail}

A term that consists just of $\tDelta$ constructors apart from exceptions 1--3 listed in \previoussubsection{sec:restaging-example}{the previous subsection} is similar enough to an actual $\tDelta$ term that we can symbolically evaluate it: all the information is already present for $\eval$ to ``know what to do''.
To capture this particular not-quite-$\tDelta$, define a data type $\tSymDelta$ (short for ``symbolic delta'') just like $\tDelta$ from \cref{fig:array-dn-ad}, but with the exceptional constructors modified as follows:
\[ \begin{array}{@{}l@{}}
    \mathbf{data}\ \tSymDelta\ \var{sh}\ \mathbf{where} \\
    \quad \multilineT{
        \del{Scale} :: \textcolor{red}{\tASTVarName\ (}\Array\ \var{sh}\ \R\textcolor{red}{)} \to \tSymDelta\ \var{sh} \to \tSymDelta\ \var{sh} \\
        \del{Share} :: \textcolor{red}{\tASTVarName}\ \tID \to \tSymDelta\ \var{sh} \to \tSymDelta\ \var{sh} \\
        \del{Index} :: \tDelta\ \shvec{k_1, \ldots, k_n} \to \textcolor{red}{\ASTIx}\ m \to \tDelta\ \shvec{k_{m+1}, \ldots, k_n} \\
        \del{Gather}\ \multilineT{
            :: \tDelta\ \shvec{k_1, \ldots, k_{m_2}, k_{m_2+1}, \ldots, k_n} \to \textcolor{red}{\ASTIxFun\ m_1\ m_2} \\
            \to \tDelta\ \shvec{k'_1, \ldots, k'_{m_1}, k_{m_2+1}, \ldots, k_n}
        } \\
        \del{Scatter}\ \multilineT{
            :: \tDelta\ \shvec{k_1, \ldots, k_{m_1}, k_{m_1+1}, \ldots, k_n} \to \textcolor{red}{\ASTIxFun\ m_1\ m_2} \\
            \to \tDelta\ \shvec{k'_1, \ldots, k'_{m_2}, k_{m_1+1}, \ldots, k_n}
        } \\
        \ccomment{Other $\tDelta$ constructors with ``$\tDelta$'' replaced by ``$\tSymDelta$''}
    }
\end{array} \]
Here we assume that the output of the $D[-]$ code transformation from \cref{fig:ad-rules-dual-arrays,fig:ad-rules-dual-arrays-op} is represented in a data type called `$\tAST$', and that variable references in `$\tAST$' are represented with `$\tASTVarName$'s.
The `$\ASTIx$' data type is a symbolic `$\Ix$':
\[ \begin{array}{@{}l@{}}
    \mathbf{data}\ \ASTIx\ k\ \mathbf{where} \\
    \quad \begin{array}{@{}l@{\ }l@{}}
        \IZ &:: \ASTIx\ 0 \\
        (\IS) &:: \tAST\ (\Array\ \shvec{}\ \Int) \to \ASTIx\ k \to \ASTIx\ (k + 1)
    \end{array}
\end{array} \]
and `$\ASTIxFun\ k_1\ k_2$' similarly captures a symbolic function $\Ix\ k_1 \to \Ix\ k_2$.
By the observation at the end of \cref{sec:restaging-example}, we can express the dual part of the return value of $D[t]$ (the ``\textit{symbolic $\tDelta$ term}'' part of \cref{eq:restaging-form}) as a value of type $\tSymDelta$.

The $\eval$ function in the non-symbolic reverse pass of \cref{sec:ad-dual-arrays} (\cref{fig:ad-array-eval}) has the following type:
\[ \eval :: \Array\ \var{sh}\ \R \to \tDelta\ \var{sh} \to \EState \to \EState \]
taking an incoming cotangent and a $\tDelta$ term to evaluate, and working on an evaluation state of type $\EState$.
For evaluation of a $\tSymDelta$, its type would become symbolic:
\[ \eval :: \tAST\ (\Array\ \var{sh}\ \R) \to \tSymDelta\ \var{sh} \to \EState \to \EState \]
in addition, of course, to changing $\EState$ as well as the type of $\reversePass$.
However, making $\eval$ evaluate symbolic $\tDelta$ terms is not quite as easy as just changing its use of the core language array operations to constructors of `$\tAST$'.
Specifically, $\eval$ sometimes duplicates the incoming cotangent $c$.
In \cref{fig:ad-array-eval}, we had:
\[ \eval\ c\ (\del{Add}\ d_1\ d_2) = \eval\ c\ d_2 \circ \eval\ c\ d_1 \]
Because this $c$ is now not simply an array but instead a term that \emph{computes} an array, passing the same $c$ to $\eval$ twice will result in inserting that term twice in the gradient program.

This is work duplication, and this duplication would occur for all $\del{Add}$ and $\del{LitArray}$ nodes; because $\del{Add}$ is emitted for every primitive operation with more than 1 argument, this work duplication is indeed quite egregious.
For example, suppose we had the following $\tSymDelta$ term, with IDs of $\del{Share}$ nodes indicated using subscripts:\footnote{This $\tDelta$ term would arise from $D[\mathbf{let}\ x = \var{inp} + \var{inp}\ \mathbf{in}\ \mathbf{let}\ y = x + x\ \mathbf{in}\ \mathbf{let}\ z = y + y\ \mathbf{in}\ z]$ with input $\var{inp}$.}

\vspace{3pt}
\begin{center}
\begin{tikzpicture}[xscale=1.5]
    \node (1) at (0, 0) {$\vphantom{\del{Input}}\del{Share}_3$};
    \node (2) at (1, 0) {$\vphantom{\del{Input}}\del{Add}$};
    \node (3) at (2, 0) {$\vphantom{\del{Input}}\del{Share}_2$};
    \node (4) at (3, 0) {$\vphantom{\del{Input}}\del{Add}$};
    \node (5) at (4, 0) {$\vphantom{\del{Input}}\del{Share}_1$};
    \node (6) at (5, 0) {$\vphantom{\del{Input}}\del{Add}$};
    \node (7) at (6, 0) {$\del{Input}$};
    \draw[->] (1) -- (2);
    \draw[->] (2) edge[bend left=30] (3);
    \draw[->] (2) edge[bend right=30] (3);
    \draw[->] (3) -- (4);
    \draw[->] (4) edge[bend left=30] (5);
    \draw[->] (4) edge[bend right=30] (5);
    \draw[->] (5) -- (6);
    \draw[->] (6) edge[bend left=30] (7);
    \draw[->] (6) edge[bend right=30] (7);
\end{tikzpicture}
\end{center}
Symbolically evaluating this term with an initial cotangent term $c$ would pass ``$c \vecadd c$'' to node 2,\footnote{The broadcasted additions here (``$\vecadd$'') come from the use of $(\vecadd)$ in $\eval$ of $\del{Share}$ in \cref{fig:ad-array-eval}.} ``$(c \vecadd c) \vecadd (c \vecadd c)$'' to node 1, and ``$((c \vecadd c) \vecadd (c \vecadd c)) \vecadd ((c \vecadd c) \vecadd (c \vecadd c))$'' to the $\del{Input}$ node!

This is precisely the same problem as we had in \cref{sec:ad-naive-sharing}: we are building a term symbolically (there a $\tDelta$ term, here a program that computes a gradient), but we cannot properly express the sharing that we need.
And there is no good place in $\eval$ to create a let-binding: $\eval$ puts $c$ in the evaluation state, and passing the same $c$ to multiple calls to $\eval$ just means that it ends up in multiple different places in the evaluation state.

\subsection[Symbolic evaluation of Delta with sharing]{Symbolic evaluation of $\tDelta$ with sharing}\label{sec:restaging-symbolic-eval}

Faced with the same problem, we apply the same solution: global sharing.
Like we have $\del{Share}$ in $\tDelta$, we add to (our extension of) the core language a new operation called `$\cc{share}$':
\[ \cc{share} : \tASTID\ (\Array\ \var{sh}\ \rho) \to \Array\ \var{sh}\ \rho \to \Array\ \var{sh}\ \rho \]
where $\tASTID$ is a type like $\tID$, except that (1) it is indexed by the full type of the array instead of just the shape, and (2) it refers to a primal program fragment of `$\tAST$' type, not a $\tDelta$ term fragment.

In its semantics on normal, concrete arrays, $\cc{share}$ does nothing: it simply returns its second argument.
However, in an AST describing a program, it indicates sharing: two terms wrapped by $\cc{share}$ nodes with the same $\tASTID$ are equal and must be computed only once.

To solve the problem from \previoussubsection{sec:restaging-symbolic-eval-fail}{the previous subsection}, we must ensure that any cotangent terms that $\eval$ duplicates are ``protected'' by $\cc{share}$.
This can be done in two different ways:
\begin{enumerate}
\item\label{item:restaging-share-pess}
    Pessimistically, i.e.\ the same way we placed $\del{Share}$ constructors in $D[-]$: whenever we construct a non-trivial cotangent, protect it against possible later duplication.
    This means that in $\eval$ in \cref{fig:ad-array-eval}, all recursive calls that do not simply pass on ``$c$'' as the first argument would be changed to first generate an $\tASTID$ with $\mgenastid$, and then instead of calling $\eval\ (\ldots c \ldots)\ d$, call $\eval\ (\cc{share}\ \var{id}\ (\ldots c \ldots))\ d)$, where $\var{id}$ is the generated $\tASTID$.
    To support the case where the initial cotangent passed to $\reversePass$ is non-trivial, $\reversePass$ would also wrap the initial cotangent in $\cc{share}$ before passing it to $\eval$.
\item\label{item:restaging-share-opt}
    By optimistically assuming that cotangents may not be duplicated at all, and to add $\cc{share}$ wrappers when duplicating cotangents instead.
    This means generating an $\tASTID$ in the right-hand side of $\eval$ for $\del{Add}$ and $\del{LitArray}$, and to pass $\cc{share}$-wrapped cotangents (with the same $\tASTID$!) to each recursive call to $\eval$ there.
\end{enumerate}
With both approaches, the reverse pass will need to run in the $\tIdGen$ monad: as we will see in \cref{sec:restaging-unglobal}, just like with $\del{Share}$ in \cref{sec:ad-naive-sharing}, we need the generated IDs to be monotonically increasing, so that subterms always have lower IDs.
This invariant is upheld by both (\ref{item:restaging-share-pess}) and (\ref{item:restaging-share-opt}).

While both approaches are valid and allow us to preserve all sharing, each has advantages and disadvantages.
Placing $\cc{share}$ pessimistically as in (\ref{item:restaging-share-pess}), we may end up placing many unnecessary $\cc{share}$ nodes if the program never actually duplicates values.
When placing them optimistically as in (\ref{item:restaging-share-opt}), a large tree of $\del{Add}$ nodes produces many internal $\cc{share}$ nodes on the cotangents, even if no actual computation happens on the cotangents in between and hence there is nothing of worth to deduplicate.

When placing $\del{Share}$ nodes to encode sharing of $\tDelta$ terms in \cref{sec:ad-naive-sharing}, we were obliged to choose approach (\ref{item:restaging-share-pess}): since the user program can pass around and compute with dual numbers (and hence, indirectly, pass around $\tDelta$ terms) without us knowing precisely where duplication happens, we could not share at duplication sites only.
With dual arrays in \cref{sec:ad-dual-arrays}, because we have a restricted input language, we \emph{can} see all duplication (namely, when a let-bound variable is used multiple times), but for consistency we kept using approach (\ref{item:restaging-share-pess}).
Here, we instead choose the optimistic approach (\ref{item:restaging-share-opt}), not because we cannot use (\ref{item:restaging-share-pess}), but because it results in a much smaller change to $\eval$: only in the equations for $\del{Add}$ and $\del{LitArray}$.

\begin{figure}
\begin{align*}
    &\begin{array}{@{}l@{}}
        \mathbf{data}\ \EState = \EState \\
        \quad \{ \multilineT{
            \ \name{grad} :: \tDMap\ \tDVarName\ (\lambda \var{sh}.\ \tAST\ (\Array\ \var{sh}\ \R)) \\
            \leadcomma\ \name{dfrag} :: \tDMap\ \tID\ \tSymDelta \\
            \leadcomma\ \name{accum} :: \tDMap\ \tID\ (\lambda \var{sh}.\ \tAST\ (\Array\ \var{sh}\ \R))\ \}
        }
    \end{array} \\
    &\begin{array}{@{}l@{}}
        \reversePass :: \tAST\ (\Array\ \var{sh}\ \R) \to \tSymDelta\ \var{sh} \to \tDMap\ \tDVarName\ (\lambda \var{sh}.\ \tAST\ (\Array\ \var{sh}\ \R)) \\
        \textcolor{gray}{\reversePass\ c\ d = \name{grad}\ (\backprop\ (\eval\ c\ d\ (\EState\ \{\}\ \{\}\ \{\})))}
    \end{array} \\
    &\begin{array}{@{}l@{}}
        \textcolor{gray}{\backprop :: \EState \to \EState} \\
        \textcolor{gray}{\backprop\ s = \multilineT{
            \mathbf{case}\ \dmMaxViewWithKey\ (\name{accum}\ s)\ \mathbf{of} \\
            \quad \multilineT{
                \name{Just}\ ((i, c), \mathit{acc}') \to \\
                \quad \multilineT{
                    \mathbf{let}\ \multilineT{
                        d = \name{dfrag}\ s \dmIndex i \\
                        s' = \eval\ c\ d\ (s\ \{\ \name{accum} = acc', \name{dfrag} = \dmDelete\ i\ (\name{dfrag}\ s)\ \})
                    } \\
                    \mathbf{in}\ \backprop\ s' \\
                } \\
                \name{Nothing} \to s \\
            }
        }}
    \end{array} \\
    &\begin{array}{@{}l@{}}
        \eval :: \tAST\ (\Array\ \var{sh}\ \R) \to \tSymDelta\ \var{sh} \to \EState \to \tIdGen\ \EState \\
        \begin{array}{@{}l@{\ }l@{}}
        \eval\ c\ \del{Zero} &= \mathbf{return} \\
        \eval\ c\ (\del{Input}\ v) &= \textcolor{gray}{\lambda s.}\ \mathbf{return}\ (\textcolor{gray}{s\ \{\ \name{grad} = \dmInsertWith\ (\vecadd)\ v\ c\ (\name{grad}\ s)\ \}}) \\
        \eval\ c\ (\del{Add}\ d_1\ d_2) &= \mathbf{do}\ \multilineT{
            \var{id} \leftarrow \mgenastid \\
            \mathbf{let}\ \var{cShared} = \cc{share}\ \var{id}\ c \\
            \eval\ \var{cShared}\ d_1 \hsklcomp \eval\ \var{cShared}\ d_2
        } \\
        \textcolor{gray}{\eval\ c\ (\del{Scale}\ \var{arr}\ d)} &\textcolor{gray}{= \eval\ (c \vecmul \var{arr})\ d} \\
        \eval\ c\ (\del{Share}\ i\ d) &= \textcolor{gray}{\lambda s.}\ \mathbf{return}\ (\textcolor{gray}{s\ \{\multilineT{
            \ \name{dfrag} = \dmInsert\ i\ d\ (\name{dfrag}\ s) \\
            \leadcomma\ \name{accum} = \dmInsertWith\ (\vecadd)\ i\ c\ (\name{accum}\ s)\ \}\textcolor{black}{)}}
        } \\
        \textcolor{gray}{\eval\ c\ (\del{Index}\ d\ i)} &\textcolor{gray}{= \eval\ (\cc{\textcolor{gray}{oneHot}}\ (\shapeDelta\ d)\ c)\ d} \\
        \textcolor{gray}{\eval\ c\ (\del{SumOuter}\ d)} &\textcolor{gray}{= \eval\ (\cc{\textcolor{gray}{replicate}}\ (\name{head}\ (\shapeDelta\ d))\ c)\ d} \\
        \textcolor{gray}{\eval\ c\ (\del{Gather}\ d\ f)} &\textcolor{gray}{= \eval\ (\cc{\textcolor{gray}{scatter}}\ (\shapeDelta\ d)\ c\ f)\ d} \\
        \textcolor{gray}{\eval\ c\ (\del{Scatter}\ d\ f)} &\textcolor{gray}{= \eval\ (\cc{\textcolor{gray}{gather}}\ (\shapeDelta\ d)\ c\ f)\ d} \\
        \eval\ c\ (\del{LitArray}\ \var{ds}) &= \mathbf{do}\ \multilineT{
            \var{id} \leftarrow \mgenastid \\
            \mathbf{let}\ \var{cShared} = \cc{share}\ \var{id}\ c \\
            \eval\ (\cc{index}\ \var{cShared}\ \ixvec{n - 1})\ (\cc{index}\ \var{ds}\ \ixvec{n - 1}) \\
            \quad \hsklcomp \cdots \hsklcomp \eval\ (\cc{index}\ \var{cShared}\ \ixvec{0})\ (\cc{index}\ \var{ds}\ \ixvec{0})
        } \\
        \textcolor{gray}{\eval\ c\ (\del{Replicate}\ d)} &\textcolor{gray}{= \eval\ (\cc{\textcolor{gray}{sumOuter}}\ c)\ d} \\
        \textcolor{gray}{\eval\ c\ (\ctranspx{\del{Transpose}}{j_1,\ldots,j_m}\ d)} &\textcolor{gray}{= \eval\ (\ctranspx{\cc{\textcolor{gray}{tr}}}{\text{inversePermutation}(j_1,\ldots,j_m)}\ c)\ d} \\
        \textcolor{gray}{\eval\ c\ (\del{Reshape}\ d)} &\textcolor{gray}{= \eval\ (\cc{\textcolor{gray}{reshape}}\ (\name{shapeDelta}\ d)\ c)\ d}
        \end{array}
    \end{array}
\end{align*}
\caption{\label{fig:restaging-eval}
    The symbolically-executing reverse pass.
    Notational punning: the core language operations ($\cc{gather}$, $\cc{index}$, etc.) denote $\tAST$ constructors here.
}
\end{figure}

The resulting reverse pass is shown in \cref{fig:restaging-eval}.
The $\EState$ data type and $\reversePass$ are adapted for $\tAST$-typed cotangents from \cref{sec:ad-dual-arrays-evaluator}; $\backprop$ is still unchanged from \cref{fig:ad-efficient-eval}, except for substituting ``$\tDMap$'' for ``$\Map$''.
The $\eval$ function is modified from \cref{fig:ad-array-eval} to share the incoming cotangent whenever it would be duplicated in multiple sub-evaluations.

It should be noted that in \cref{fig:restaging-eval} (the symbolic reverse pass), we use some convenient notational punning to emphasise the similarity to \cref{fig:ad-array-eval} (the non-symbolic array reverse pass): in \cref{fig:ad-array-eval}, the core language operations ($\cc{gather}$, $\cc{index}$, etc.) referred to \emph{array operations} that were performed on concrete arrays.
In \cref{fig:restaging-eval}, they instead refer to $\tAST$ constructors of the gradient term that is being built up.

Notable is that we do \emph{not} need to wrap $\cc{share}$ around cotangents that are added to the cotangent accumulation map (the `$\name{accum}$' field of $\EState$) or the gradient accumulation map (the `$\name{grad}$' field): they are never duplicated, because they are in fact used exactly once (either to retrieve their final value or to be added to yet another contribution).

\paragraph{Example}
Let $s$ be the the symbolic $\tDelta$ term extracted from the simplified version of $D[\texSc']$ in \cref{fig:restaging-example}, with $a_d$ replaced with $\del{Input}\ \texttt{"inp"}$.
If we run $\reversePass$ from \cref{fig:restaging-eval} on it, we get the following:
\[ \begin{array}{@{}l@{}}
    \reversePass\ \arrlit{1.0}\ s = \\
    \quad \{\ \texttt{"inp"} \mapsto \multilineT{
        \cc{scatter}\ \shvec{n}\ (x_2 \vecmul \cc{share}\ 1\ (\cc{replicate}\ n\ \arrlit{1.0}))\ (\clam{\ixvec{i}} \ixvec{i}) \\
        \quad {} \vecadd \cc{scatter}\ \shvec{n}\ (x_1 \vecmul \cc{share}\ 1\ (\cc{replicate}\ n\ \arrlit{1.0}))\ (\clam{\ixvec{i}} \ixvec{n - 1 - i})\ \}
    }
\end{array} \]
The first argument to $\reversePass$ is the incoming cotangent, which in this case is a zero-dimensional array of scalars (i.e.\ a single scalar).

Note that the `1' in the $\cc{share}$ nodes was generated by the $\mgenastid$ call in $\eval$ of the $\del{Add}$ node.
Further, because an interpreter of this term is supposed to memoise the results of evaluating $\cc{share}$-wrapped subterms, a proper interpreter would \emph{not} execute the $\cc{replicate}$ (the reverse derivative of the $\cc{sumOuter}$ in $\texSc$) twice.
To avoid excessive memory use in representing this AST in a compiler, one should ensure that $\cc{share}$-wrapped terms are actually shared in-memory in the compiler too.

\subsection{Converting global sharing to let-bindings}\label{sec:restaging-unglobal}

In principle, we can just generate the final gradient program by looking back at the structure of $D[t]$ after simplification (\cref{eq:restaging-form} on page \pageref{eq:restaging-form}), and replacing ``\textit{symbolic $\tDelta$ term}'' with `$\name{reconstruct}$' (\cref{fig:dn-ad-wrapper}) applied to the output of $\reversePass$ (\cref{fig:restaging-eval}) on the original ``\textit{symbolic $\tDelta$ term}''.
However, if we do this, the result is a term that has both let-bindings \emph{and} global sharing using $\cc{share}$: the let-bindings occur in the primal computation, and the global sharing is in the dual computation (produced by the symbolic reverse pass).
Because global sharing is rather seldomly used in compilers, for compiling this gradient program to native code it is likely necessary to ensure that the whole term uses let-bindings for sharing, and nothing else.

Thus, we need to eliminate $\cc{share}$ from the term produced by the (wrapper around the) symbolic reverse pass.
Fortunately, this term only contains a very limited set of constructs:
\begin{itemize}
\item
    It contains variable references: these come from the first field of $\del{Scale}$, which we restricted to be a variable reference in \cref{sec:restaging-example}.
\item
    It contains the initial cotangent passed to $\reversePass$, which we will assume to be either a contant literal or a variable reference.
    (The wrapper in \cref{sec:restaging-wrapper} will put a variable reference here.)
\item
    Otherwise, it contains only what $\eval$ from \cref{fig:restaging-eval} produces explicitly.
    For our core language, one can verify (by closely reading \cref{fig:restaging-eval}) that this is limited to:
    $\cc{share}$, $(\vecadd)$, $(\vecmul)$, $\cc{oneHot}$, $\cc{replicate}$, $\cc{scatter}$, $\cc{gather}$, $\cc{index}$, $\cc{sumOuter}$, $\cc{tr}$, and $\cc{reshape}$.
    Some subtleties:
    \begin{itemize}
    \item
        The $\cc{index}$ construct contains a list of terms for the position to index at; these are arbitrary terms from the primal program.
    \item
        Similarly, the $\cc{gather}$ and $\cc{scatter}$ constructs contain index functions, which also come straight from the primal program.
    \end{itemize}
    These components may contain arbitrary terms, but because $\eval$ just preserves these terms as-is, we are sure that they do not contain $\cc{share}$, meaning that we can also keep them as-is here.
\end{itemize}
In particular, \emph{it does not contain let-bindings outside of untouched, copied subterms}.
(There may of course be internal let-bindings inside of the index and index-function arguments to $\cc{index}$ and $\cc{gather}$/$\cc{scatter}$, but we do not need to care about those.)
Because the semantics of $\cc{share}$ is somewhat murky in the presence of let-bindings, their absence makes our task of conversion to standard let-bindings much more straightforward.

\begin{figure}
\[ \begin{array}{@{}l@{}}
    \unshare\ \multilineT{
        :: \tDMap\ \tASTID\ (\lambda \tau.\ (\tASTVarName\ \tau, \tAST\ \tau)) \to \tAST\ \tau \\
        \to (\tDMap\ \tASTID\ (\lambda \tau.\ (\tASTVarName\ \tau, \tAST\ \tau)), \tAST\ \tau)
    } \\
    \begin{array}{@{}l@{\ }l@{}}
    \unshare\ m\ x &= (m, x) \qquad\ccomment{variable references} \\
    \unshare\ m\ c &= (m, c) \qquad\ccomment{constant literals} \\
    \unshare\ m\ (\cc{share}\ i\ t) &= \multilineT{
        \mathbf{case}\ \dmLookup\ i\ m\ \mathbf{of} \\
        \quad \name{Just}\ (\var{var}, \_) \to (m, \var{var}) \quad\ccomment{return a variable reference term} \\
        \quad \name{Nothing} \to \\
        \quad\quad \multilineT{
            \ccomment{These are meta-`let's, not $\tAST$ constructors} \\
            \mathbf{let}\ \multilineT{
                \var{var} = \langle\text{\underline{\smash{generate a fresh variable name}}}\rangle \\
                t' = \unshare\ m\ t
            } \\
            \mathbf{in}\ (\dmInsert\ i\ (\var{var}, t')\ m, \var{var})
        }
    } \\
    \unshare\ m\ (s \vecadd t) &= \multilineT{
        \mathbf{let}\ \multilineT{
            (m_1, s') = \unshare\ m\ s \qquad\ccomment{idem} \\
            (m_2, t') = \unshare\ m_1\ t \\
        } \\
        \mathbf{in}\ (m_2, s' \vecadd t')
    } \\
    \unshare\ m\ (s \vecmul t) &= \multilineT{
        \mathbf{let}\ \multilineT{
            (m_1, s') = \unshare\ m\ s \qquad\ccomment{etc.} \\
            (m_2, t') = \unshare\ m_1\ t \\
        } \\
        \mathbf{in}\ (m_2, s' \vecmul t')
    } \\
    \unshare\ m\ (\cc{oneHot}\ \var{sh}\ t) &= \mathbf{let}\ (m', t') = \unshare\ m\ t\ \mathbf{in}\ (m', \cc{oneHot}\ \var{sh}\ t') \\
    \unshare\ m\ (\cc{replicate}\ k\ t) &= \mathbf{let}\ (m', t') = \unshare\ m\ t\ \mathbf{in}\ (m', \cc{replicate}\ k\ t') \\
    \unshare\ m\ (\cc{scatter}\ \var{sh}\ t\ f) &= \mathbf{let}\ (m', t') = \unshare\ m\ t\ \mathbf{in}\ (m', \cc{scatter}\ \var{sh}\ t'\ f) \quad\ccomment{$f$ untouched} \\
    \unshare\ m\ (\cc{gather}\ \var{sh}\ t\ f) &= \mathbf{let}\ (m', t') = \unshare\ m\ t\ \mathbf{in}\ (m', \cc{gather}\ \var{sh}\ t'\ f) \\
    \unshare\ m\ (\cc{index}\ t\ \var{ix}) &= \mathbf{let}\ (m', t') = \unshare\ m\ t\ \mathbf{in}\ (m', \cc{index}\ t'\ \var{ix}) \quad\ccomment{$\var{ix}$ untouched} \\
    \unshare\ m\ (\cc{sumOuter}\ t) &= \mathbf{let}\ (m', t') = \unshare\ m\ t\ \mathbf{in}\ (m', \cc{sumOuter}\ t') \\
    \unshare\ m\ (\ctransp{k_1,\ldots,k_n}\ t) &= \mathbf{let}\ (m', t') = \unshare\ m\ t\ \mathbf{in}\ (m', \ctransp{k_1,\ldots,k_n}\ t') \\
    \unshare\ m\ (\cc{reshape}\ \var{sh}\ t) &= \mathbf{let}\ (m', t') = \unshare\ m\ t\ \mathbf{in}\ (m', \cc{reshape}\ \var{sh}\ t')
    \end{array} \vspace{0.4em} \\
    \begin{array}{@{}l@{}}
    \stackLets :: \tDMap\ \tASTID\ (\lambda \tau.\ (\tASTVarName\ \tau, \tAST\ \tau)) \to \tAST\ \tau \to \tAST\ \tau \\
    \stackLets\ m\ t = \multilineT{
        \mathbf{case}\ \dmMaxViewWithKey\ m\ \mathbf{of} \\
        \quad \name{Just}\ ((\_, (\var{var}, t')), m') \to \\
        \quad\quad \stackLets\ m'\ (\cletin{\var{var} = t'}{t}) \quad\ccomment{An $\tAST$ let-binding!} \\
        \quad \name{Nothing} \to t
    }
    \end{array} \vspace{0.4em} \\
    \begin{array}{@{}l@{}}
    \shareToLet :: \tAST\ \tau \to \tAST\ \tau \\
    \shareToLet\ t = \mathbf{let}\ (m, t') = \unshare\ \{\}\ t\ \mathbf{in}\ \stackLets\ m\ t'
    \end{array}
\end{array} \]
\caption{\label{fig:restaging-unshare}
    Converting global sharing to let-bindings.
}
\end{figure}

The conversion function, for this peculiar language (the output of symbolic $\eval$) that we need to support, is given in \cref{fig:restaging-unshare}.
In this figure, we indicate meta-`let' (i.e.\ a let-binding in the language that $\unshare$ etc.\ themselves are written in) by ``$\mathbf{let}$'', and a term-`let' in the $\tAST$ data type by ``$\cletnosp$''.
All `let's in $\unshare$ are meta; the one `let' in $\stackLets$ is a term.

The conversion function from global sharing to let-bindings is $\shareToLet$, which first collects all shared term fragments in a $\tDMap$ using $\unshare$, and then stacks $\cletnosp$ terms on top of the root term fragment using $\stackLets$.
The $\unshare$ function takes a term $t$ with $\cc{share}$ nodes inside, and returns a dictionary of all the fragments inside $t$, together with a single root term fragment $t'$ that consists of the constructors near the root of $t$ above the first $\cc{share}$ nodes.
In both halves of the return value of $\unshare$, $\cc{share}$-wrapped subterms have been replaced with fresh variable names.

Thus, inside $\unshare$, whenever we encounter a $\cc{share}$ node, we have to decide on a variable name that will represent this fragment in the unshared term.
However, all we have is an $\tASTID$.
If we have external knowledge that $\tASTID$ and $\tASTVarName$ refer to equal (or convertible) types, and that the $\tASTVarName$ derived from an $\tASTID$ in this fashion will never be used by the user or generated by the differentiation machinery before this section, the side-effect ``$\langle\text{\underline{\smash{generate a fresh variable name}}}\rangle$'' in $\unshare$ could simply convert $i$ to a variable name instead.
If not, then $\unshare$ should additionally run in a simple state monad for generating fresh names.

Having a $\tDMap$ of all term fragments of $t$, we simply build a long stack of let-bindings on top of the root term fragment.
For this to make sense, we need to ensure that the fragments without dependencies are bound at the top of the stack, the fragments that depend just on those come right after, etc.
Fortunately, because the IDs are generated monotonically in the symbolic $\eval$ function (\cref{fig:restaging-eval}), it suffices to simply bind the IDs from lowest to highest.
Because $\stackLets$ builds the stack from the bottom up instead of from the top down, it starts with the term with the \emph{highest} ID at the bottom of the stack, and then proceeds upwards with lower and lower IDs until the whole $\tDMap$ is exhausted.

\paragraph{Example}
Running $\unshare$ on the gradient term (with $\cc{share}$) that we derived for the example at the end of \cref{sec:restaging-symbolic-eval}, we get:
\[ \begin{array}{@{}l@{}}
    \unshare\ \{\}\ (\reversePass\ \arrlit{1.0}\ s) = \\
    \quad \bigl(\multilineT{
        \{\ 1 \mapsto (\texttt{"shared1"}, \cc{replicate}\ n\ \arrlit{1.0})\ \} \\
        \leadcomma \cc{scatter}\ \shvec{n}\ (x_2 \vecmul \var{shared1})\ (\clam{\ixvec{i}} \ixvec{i}) \vecadd \cc{scatter}\ \shvec{n}\ (x_1 \vecmul \var{shared1})\ (\clam{\ixvec{i}} \ixvec{n - 1 - i})\bigr)
    }
\end{array} \]
Here we generated the name ``shared1'' for the ID 1.
Completing with $\stackLets$ (which has a rather easy task in this case with just one shared fragment), we get a term without $\cc{share}$:
\[ \begin{array}{@{}l@{}}
    \shareToLet\ (\reversePass\ \arrlit{1.0}\ s) = \\
    \quad \clet{\var{shared1} = \cc{replicate}\ n\ \arrlit{1.0}} \\
    \quad \cin \cc{scatter}\ \shvec{n}\ (x_2 \vecmul \var{shared1})\ (\clam{\ixvec{i}} \ixvec{i}) \vecadd \cc{scatter}\ \shvec{n}\ (x_1 \vecmul \var{shared1})\ (\clam{\ixvec{i}} \ixvec{n - 1 - i})
\end{array} \]

\subsection{Wrapper}\label{sec:restaging-wrapper}

Now we have all the components to assemble the final algorithm in a wrapper that the user can make sense of.
A pseudocode rendering of this wrapper is shown in \cref{fig:restaging-wrapper}.

\begin{figure}
\[ \begin{array}{@{}l@{}}
    \name{wrapper}\ \multilineT{
        :: \tAST\ (\Array\ \shvec{}\ \R) \quad\ccomment{free variables: $x_1 : \Array\ \var{sh}_1\ \rho_1, \ldots, x_n : \Array\ \var{sh}_n\ \rho_n$} \\
        \to \tAST\ (\Array\ \shvec{}\ \R, (\Array\ \var{sh}_1\ \rho_1, \ldots, \Array\ \var{sh}_n\ \rho_n)) \\
        \hphantom{\to}\quad \ccomment{free variables: $x_1 : \Array\ \var{sh}_1\ \rho_1, \ldots, x_n : \Array\ \var{sh}_n\ \rho_n$ and $c : \Array\ \shvec{}\ \R$}
    } \\
    \name{wrapper}\ t = \\
    \quad \mathbf{let}\ \multilineT{
        t_{\text{bulk}} = \BOT[t] \\
        t_{\text{diff}} = \cletin{(x_1, \ldots, x_n) = ((x_1, \del{Input}\ 1), \ldots, (x_n, \del{Input}\ n))}{D[t_{\text{bulk}}]} \\
        \left(\begin{array}{@{}l@{}}\mathbf{do}\ \multilineT{
            \mathbf{let}\ \var{primalBinds} \\
            \var{id}_1 \leftarrow \mgenid; \ldots; \var{id}_m \leftarrow \mgenid \\
            \mathbf{return}\ (t_1, t_2)
        }\end{array}\right) = \name{extractBySimplification}[t_{\text{diff}}] \\
        t_2' = \shareToLet\ (\reversePass\ c\ (t_2[\var{id}_1 \vcentcolon= 1, \ldots, \var{id}_m \vcentcolon= m]))
    } \vspace{0.5em} \\
    \quad \mathbf{in}\ \raisebox{-0.62em}{\(\left(\begin{array}{@{}l@{}}
        \clet \var{primalBinds} \\
        \cin (t_1, t_2')
    \end{array}\right)\)}
\end{array} \]
\caption{\label{fig:restaging-wrapper}
    The wrapper for the full algorithm.
    See \cref{sec:restaging-wrapper} for details on the notation.
}
\end{figure}

All stages of the pipeline come to the fore here:
\begin{itemize}
\item
    We start with a program (`$t$') that returns a single scalar (i.e.\ $t :: \tAST\ (\Array\ \shvec{}\ \R)$) with a number of free variables that consistute the input parameters that we differentiate with respect to.

\item
    We first vectorise $t$ to use bulk operations using the \BOT (\cref{sec:bot}); this results in $t_{\text{bulk}}$.

\item
    We (forward-)differentiate $t_{\text{bulk}}$ using $D[-]$ (\cref{sec:ad-dual-arrays}), and rebind its free variables: the $\cletnosp$ in the assignment to $t_{\text{diff}}$ is an embedded let-binding.
    Where $t_{\text{bulk}}$ (still) had free variables $x_i : \Array\ \var{sh}_i\ \rho_i$, the differentiated term $D[t_{\text{bulk}}]$ has free variables $x_i : (\Array\ \var{sh}_i\ \rho_i, \tDelta\ \var{sh}_i)$.
    The $\cletnosp$, which should be read as a \emph{non-recursive} let-binding, provides the second components of those and ensures that $t_{\text{diff}}$'s free variables have type $\Array\ \var{sh}_i\ \rho_i$ again.

\item
    Then we simplify as in the example (\cref{fig:restaging-example} in \cref{sec:restaging-example}), and as more rigorously justified in \cref{sec:restaging-induction}.
    This produces a term of a specific form (\cref{eq:restaging-form}); we pattern-match out the components: $\var{primalBinds}$, $m$, $\var{id}_i$, $t_1$ and $t_2$.
    (The notation between the large parentheses is a \emph{term}.)

\item
    Finally, in the assignment to $t_2'$, we first substitute $1,\ldots,m$ for $\var{id}_1, \ldots, \var{id}_m$ in $t_2$ (essentially ``running the $\tIdGen$ monad'' in poor man's fashion), and then put the result through the machinery of \cref{sec:restaging-symbolic-eval,sec:restaging-unglobal}.
    The ``$c$'' here is a \emph{term}: a variable reference to the variable `$c$', the initial cotangent.
    This means that $t_2'$ has as free variables:
    \begin{itemize}
    \item The primal inputs $x_i$;
    \item Any names bound in $\var{primalBinds}$;
    \item The initial cotangent: $c : \Array\ \shvec{}\ \R$.
    \end{itemize}

\item
    The result of the wrapper is a term (again written between large parentheses) that first runs the part of the primal that was shared between $t_1$ and $t_2$, and then returns a pair of the primal result ($t_1$) and the gradient that is, by now, a standard term ($t_2'$).
    The result of the wrapper computes not only a gradient, but also the primal result, as can be seen in its type; the free variables of this term are the inputs $x_i$ as well as the initial cotangent $c$.
\end{itemize}

\subsection[Delta extraction works in general]{$\tDelta$ extraction works in general}\label{sec:restaging-induction}

In \cref{sec:restaging-example} we saw that the derivative of the example term, $D[\texSc']$, simplified to a particularly useful form:
\begin{equation*}
    \eqRestagingForm
    \tag{\ref{eq:restaging-form} again}
\end{equation*}
and we claimed that this in fact holds for \emph{all} source terms.
Furthermore, we specified that the ``\textit{symbolic $\tDelta$ term}'' had to conform to some requirements:
\begin{enumerate}
\item The $\tID$ field of a $\del{Share}$ constructor is a variable reference;
\item Similarly for the scaling field of a $\del{Scale}$ constructor;
\item The index (function) fields of $\del{Index}$, $\del{Gather}$ and $\del{Scatter}$ can be arbitrary terms;
\item Otherwise, it consists of only $\tDelta$ constructors, except for a variable reference to a $\tDelta$ component of an input.
\end{enumerate}

We can prove that this form holds for the output of $D[-]$ on our core language by induction.
To do so, we look at every equation in \cref{fig:ad-rules-dual-arrays,fig:ad-rules-dual-arrays-op}.
We will consider only source terms of type $\Array\ \var{sh}\ \R$ for some shape $\var{sh}$; of course, such a program may also contain subterms of type $\Array\ \var{sh}\ \rho$ with $\rho$ unequal to $\R$, but due to the structure of the equations in the transformation, we end up being able to ignore those subterms.

Observe that for constructs of type $\Array\ \var{sh}\ \R$ in \cref{fig:ad-rules-dual-arrays,fig:ad-rules-dual-arrays-op}, the right-hand side of the equation fits the following form:\footnote{For the polymorphic constructs $x$, $\cc{cond}$ and $\cletnosp$, one can artificially distinguish the $\R$ and non-$\R$ cases and subsequently look only at the $\R$ case.}
\begin{equation}
    \begin{array}{@{}l@{}} 
    \mathbf{do}\ \multilineT{
        (x_1, d_1) \leftarrow D[t_1]; \ldots; (x_n, d_n) \leftarrow D[t_n] \\
        \mathbf{let}\ \textit{most of the primal computation} \\
        (\var{id}_1, \ldots, \var{id}_m) \leftarrow \textit{generate $\tID$s} \\
        \mathbf{return}\ (\multilineT{
            \textit{result of the primal computation} \\
            \leadcomma \textit{symbolic $\tDelta$ term})
        }
    }
    \end{array}
    \label{eq:restaging-induction-form}
\end{equation}
The $t_i$ are direct subterms of the term on the left-hand side of the equation that are also of scalar-array type, and the ``\textit{symbolic $\tDelta$ term}'' adheres to the same constraints as we set for \cref{eq:restaging-form} above, \emph{except} that it may also refer to the $d_i$ with variable references \emph{at most once}.
The requirement that each $d_i$ is referred to at most once prevents us from needing to introduce global sharing here just yet.

Now assuming (by the induction hypothesis) that $D[t_i]$ is already in form \cref{eq:restaging-form}, our only task is to rewrite \cref{eq:restaging-induction-form} to \cref{eq:restaging-form}.
If we do so, then by induction, form \cref{eq:restaging-form} can be derived from the derivative of every source program term by repeated simplification, and we are done.

But indeed, this is not very difficult: subscripting the components of \cref{eq:restaging-form} with $i$ according to which $D[t_i]$ it corresponds to, the nested structure looks as follows:
\begin{equation*}
    \mathbf{do}\ \multilineT{
        (x_1, d_1) \leftarrow \mathbf{do}\ \multilineT{
            \mathbf{let}\ (\textit{most of the primal computation})_1 \\
            (\var{id}_1, \ldots, \var{id}_{m_1}) \leftarrow \textit{generate $\tID$s} \\
            \mathbf{return}\ (\multilineT{
                (\textit{result of the primal computation})_1 \\
                \leadcomma (\textit{symbolic $\tDelta$ term})_1)
            }
        } \\
        \ldots \\
        (x_n, d_n) \leftarrow \mathbf{do}\ \multilineT{
            \mathbf{let}\ (\textit{most of the primal computation})_n \\
            (\var{id}_1, \ldots, \var{id}_{m_n}) \leftarrow \textit{generate $\tID$s} \\
            \mathbf{return}\ (\multilineT{
                (\textit{result of the primal computation})_n \\
                \leadcomma (\textit{symbolic $\tDelta$ term})_n)
            }
        } \\
        \mathbf{let}\ \textit{most of the primal computation} \\
        (\var{id}_1, \ldots, \var{id}_m) \leftarrow \textit{generate $\tID$s} \\
        \mathbf{return}\ (\multilineT{
            \textit{result of the primal computation} \\
            \leadcomma \textit{symbolic $\tDelta$ term})
        }
    }
\end{equation*}
which is easily rewritten to:
\begin{equation*}
    \mathbf{do}\ \multilineT{
        \mathbf{let}\ \multilineT{
            (\textit{most of the primal computation})_1 \\
            \ldots \\
            (\textit{most of the primal computation})_n \\
            x_1 = (\textit{result of the primal computation})_1 \\
            \ldots \\
            x_n = (\textit{result of the primal computation})_n \\
            \textit{most of the primal computation}
        } \\
        (\var{id}^1_1, \ldots, \var{id}^1_{m_1}, \ldots, \var{id}^n_1, \ldots, \var{id}^n_{m_n}, \var{id}_1, \ldots, \var{id}_m) \leftarrow \textit{generate $\tID$s} \\
        \mathbf{return}\ (\multilineT{
            \textit{result of the primal computation} \\
            \leadcomma (\textit{symbolic $\tDelta$ term})[\multilineT{
                d_1 \vcentcolon= (\textit{symbolic $\tDelta$ term})_1, \ldots, \\
                \qquad d_n \vcentcolon= (\textit{symbolic $\tDelta$ term})_n])
            }
        }
    }
\end{equation*}
with some alpha-renaming for the $\var{id}$ variables, and potentially other variables to avoid clashing names.
Note that we are allowed to simply substitute $d_i$ into the symbolic $\tDelta$ term because each occurs at most once, as noted above.

This completes the induction, and confirms that the example in \cref{fig:restaging-example} at the beginning of this section was indeed representative.

\section{Implementation}\label{sec:implementation}

We have an implementation of the algorithms described in this paper in Haskell.\footnote{Available as a package here: \implementationLink}
The implementation contains various features not described in detail in the paper:
\begin{itemize}
\item
    The source language of the library is a shallowly embedded array language in Haskell; the resulting staging features are described in \cref{sec:staging}.

\item
    Dual-numbers AD admits a very elegant implementation in a functional language via type classes~\cite{2009-elliott-typeclass-ad}; dual-numbers reverse AD as presented in~\cite{2022-krawiec-dualrev,2023-smeding-dualrev} preserves this property to an extent, although it is not emphasised in the referenced articles.
    The possibility of a modular implementation based on type classes remains true for the core AD algorithm in this paper (\cref{sec:ad-dual-arrays}); the \BOT naturally cannot, but staging (\cref{sec:staging}) can save us here.

    Our implementation explicitly revives the type class approach and contains one implementation of AD (from which the dual arrays algorithm of \cref{sec:ad-dual-arrays} and the naive algorithm of \cref{sec:ad-naive} are special cases) and one implementation of staging (underlying both the staging of the source program and the symbolic evaluation in \cref{sec:restaging-symbolic-eval}).
    These are all instances of a type class modelling our core array language.
    Details about this system and the sharing-related subtleties that need to be solved are given in \cref{sec:algebra-interp,sec:algebra-ad}.

\item
    As briefly mentioned before, the implementation supports a higher-order `$\name{fold}$' operation with the restriction that the combination function must be closed --- this is not a requirement for `$\cc{build1}$'.
    Furthermore, the implementation has full support for binary tuples (i.e.\ pairs) in the source language, and limited support for \emph{regular} nested arrays: after struct-of-arrays transformation, all component arrays must be fully rectangular.
    That is to say: nested array support cannot be use to get around the prohibition of jagged arrays.
\end{itemize}

The library benchmarks favourably against \texttt{ad}\footnote{\url{https://hackage.haskell.org/package/ad}}, but performance competitive with state-of-the-art machine learning toolkits is future work.

\section{Discussion and Future Work}

Some of the simplicity and easy generalisation of dual-numbers AD is retained in our algorithm: we describe the actual AD component of \cref{sec:ad-dual-arrays} to work only on the output of the \BOT, i.e.\ a quite restricted language of bulk array operations, but the AD code transformation itself ($D[-]$) would not care if we added e.g.\ product types, sum types or function types to the language.
The role of the scalar type $\R$ is now fulfilled by \emph{arrays} of scalars ($\Array\ \var{sh}\ \R$), but the structure of the algorithm is the same.

The other parts of the algorithm (the \BOT (\cref{sec:bot}) and symbolic evaluation (\cref{sec:restaging})) are not so kind; especially proper dynamic control flow (i.e.\ a lazy conditional, loops, perhaps recursion) makes it unclear how to vectorise, and throws a wrench in the rather straightforward $\tDelta$ extraction process of \cref{sec:restaging-induction}.
\citeN{2021-paszke-dex} are able to ``unzip'' the primal computation from the dual computation even in the presence of dynamic control flow, and they do so by evaluating the conditional \emph{twice}: once in the primal, where they store and ``export'' intermediate values from the computation in the branch taken to outside the conditional, and once in the dual (analogous to our $\tDelta$ evaluation), where they make use of the stored conditional boolean as well as the stored intermediate values to run the reverse pass of the correct branch.
Perhaps an extension to $\tDelta$ could allow such tricks, but it is unclear how precisely.

From a performance perspective, the \BOT is rather uncomfortable: \emph{array fusion} (reducing the number of ``loops'' / passes over the data) is typically seen as an optimisation by reducing loop overhead and memory traffic, but the \BOT does the exact opposite thing.\footnote{
    When compiled using a good compiler that uses e.g.\ CPU vector instructions whenever possible, the expression `$\cc{build1}\ (\name{length}\ a)\ (\clam i (\cc{index}\ a\ i)^2 + \cc{index}\ b\ i)$' will be faster than `$a^2 + b$', because the former does only two memory reads and one write per element, whereas the latter does 1 read and write for the squaring and two reads and yet another write for the $(+)$.
    (They perform the same number of arithmetic operations.)
    Hence, conversion from the latter to the former is seen as an optimisation in array languages: this is called \emph{(loop) fusion}; this example in particular is a special case called \emph{map--map fusion}, or more generally \emph{vertical fusion}.
    Our \BOT produces code very much in the latter form, and array AD preserves this property.
}
As we saw in \cref{sec:restaging}, the primal program is mostly preserved in the simplified differentiated version, albeit with many additional intermediate values stored in variables; but in fact, with sufficiently clever fusion algorithms, such additional stores need not be a big impediment to re-fusion of the generated primal code~\cite{2024-diagonal-fusion}.
However, the reverse pass is generally more expensive than the primal pass in reverse AD (simply because it involves more computation~\cite{adbook-2008-griewank-walther}), and there it is less clear to what extent fusion opportunities are preserved through differentiation.
We would like to get a better understanding of the interaction between vectorisation, differentiation, and fusion (and other array-program performance optimisations).


In conclusion, the AD algorithm presented in this paper has various desirable properties, but more research is needed to extract all of its value and make it competitive with the state-of-the-art.

\section{Related Work}

\paragraph{Automatic differentiation}
A logical perspective on dual-numbers reverse-mode AD was presented in \cite{2020-brunel-dualrev,2021-mazza-dualrev-pcf}, focusing on correctness.
\citeN{2022-krawiec-dualrev} gave a correctness proof of a lambda calculus version of the algorithm that has the correct complexity; independently, inspired by the work of \citeN{2020-brunel-dualrev}, \citeN{2023-smeding-dualrev} analysed the same algorithm but focused on a complexity proof instead of a correctness proof.
They later extended the (scalar-level) algorithm to fork-join task parallelism~\cite{2024-smeding-dualrev-parallel}, an extension that is mostly orthogonal to the extensions presented in the present paper.
Semantical perspectives are given by \citeN{2020-ad-gluing} (\S 6) and \citeN{2022-ad-logical-relations}.

In \cite{2022-krawiec-dualrev,2023-smeding-dualrev}, as well as in this paper, the actual AD algorithm (excluding preprocessing in the \BOT) takes the form ``forward-differentiate (\cref{sec:ad-dual-arrays} ($D[-]$)), unzip primal from dual (\cref{sec:restaging-induction}), transpose forward derivative to reverse derivative (\cref{sec:restaging-symbolic-eval} ($\eval$))''.
Our transposition step is symbolic (\cref{sec:restaging}), whereas the ones in \cite{2022-krawiec-dualrev,2023-smeding-dualrev} are not.
This three-part structure is also used in Dex~\cite{2021-paszke-dex} and further explained, albeit on a language without arrays or dynamic control flow, by \citeN{2023-yolo}.
In Dex, the \texttt{for} syntax (their equivalent of $\cc{build}$) is differentiated by introducing mutable accumulators in the derivative program; this allows array indexing (our $\cc{index}$) to have an efficient derivative while still computing dense cotangents.
In this paper, we present a purely functional approach that avoids such pervasive mutability.

In a parallel line of work, \citeN{2018-elliott-ad} presented a categorical perspective on functional AD.
The CHAD algorithm of \citeN{2022-chad} translates this from categorical combinators to the lambda calculus and extends it to higher-order programs; a similar extension was independently presented by \citeN{2019-vytiniotis-diff-curry}.
The time complexity of CHAD was analysed and made optimal by \citeN{2024-efficient-chad}, and its theory was extended to (co)inductive data types by \citeN{2023-chad-expressive}.
CHAD (and its related approaches) does not follow the three-part structure of the algorithm in the present paper; it constructs a reverse derivative immediately.
A discussion on the relation between CHAD and dual-numbers reverse AD from the perspective of the transformation on types is given in~\cite[\S 4]{2024-smeding-dualrev-parallel}

\citeN{2022-futhark-ad} also presented a reverse AD algorithm on a (parallel) array language that accepts some recomputation in the reverse pass in return for a code transformation that produces much more structured code.

\paragraph{Vectorisation}
Aggressive vectorisation (unfusion) like our \BOT does, is uncommon in prior work.
However, a similar kind of ``vectorising map'' can be found in JAX~\cite{2018-bradbury-jax} as \texttt{jax.vmap} and as a prototype feature in PyTorch~\cite{2017-paszke-pytorch} as \texttt{torch.vmap}.
The PyTorch implementation shares our restriction that array shapes must be statically known.
These implementations are primarily for expressivity or ease of writing models.
Futhark~\cite[\S 5.1]{2017-henriksen-futhark} employs a vectorisation-like transform for making code more suitable for GPU compilation, by making more parallelism statically visible; we go further in that we do not only vectorise what is easy to vectorise, but we aggressively vectorise the entire expression.

Vectorisation in the presence of unequal array shapes traditionally requires some form of flattening, as in e.g.\ NESL~\cite{1992-nesl} or Data-Parallel Haskell~\cite{2007-dph}.
This flattening tends to give non-negligible runtime overheads; we elected to avoid these overheads by simply not supporting unequal array shapes in vectorised code.

\begin{acks}
We would like to thank Tom Ellis for good discussion and reflection on algorithms and presentation.
TS would like to thank Matthijs Vákár for guidance and advice on presenting the material.
\end{acks}

\newpage

\appendix
\crefalias{section}{appendix}  

\section*{Appendix}

\section{Generalisation of the Output Type}\label{app:output-type-generalisation}

The top-level interface to our AD algorithm as finalised in \cref{sec:restaging}, shown in \cref{fig:restaging-wrapper}, has the following type when seen as a code transformation:
\[ \begin{array}{@{}l@{}}
    x_1 : \Array\ \var{sh}_1\ \rho_1, \ldots, x_n : \Array\ \var{sh}_n\ \rho_n \\
    \qquad \vdash t : \Array\ \shvec{}\ \R \\
    \hspace{1.5cm}\leadsto \\
    x_1 : \Array\ \var{sh}_1\ \rho_1, \ldots, x_n : \Array\ \var{sh}_n\ \rho_n, c : \Array\ \shvec{}\ \R \\
    \qquad \vdash \name{wrapper}[t] : (\Array\ \shvec{}\ \R, (\Array\ \var{sh}_1\ \rho_1, \ldots, \Array\ \var{sh}_n\ \rho_n))
\end{array} \]
While this is typically sufficiently general in the input type, some applications may require more complex output types than a single real scalar.

Computing the Jacobian for such a more general function requires multiple passes with reverse AD.
(Forward AD cares little about the size of the output, and instead requires multiple passes if the \emph{input} consists of multiple scalars.)
If nested arrays are supported, one can generalise to an array of output scalars straightforwardly:
\[ \begin{array}{@{}l@{}}
    x_1 : \Array\ \var{sh}_1\ \rho_1, \ldots, x_n : \Array\ \var{sh}_n\ \rho_n, c : \Array\ \var{sh}\ \R \\
    \qquad \vdash \multilineT{
        \clet r = \cc{build}\ (\cc{shape}\ c)\ (\clam{\var{ix}} \cletin{c = \cc{index}\ c\ \var{ix}} \name{wrapper}[\cc{index}\ t\ \var{ix}]) \\
        \cin (\cc{map}\ (\clam x \name{scalar}\ (\name{fst}\ x))\ r, \cc{map}\ (\clam x \name{snd}\ x)\ r)
    } \\
    \qquad : (\Array\ \var{sh}\ \R, \Array\ \var{sh}\ (\Array\ \var{sh}_1\ \rho_1, \ldots, \Array\ \var{sh}_n\ \rho_n))
\end{array} \]
where:
\[ \begin{array}{@{}l@{}}
    \cc{map}\ f\ x = \cc{build}\ (\cc{shape}\ x)\ (\clam{\var{ix}} f\ (\cc{index}\ x\ \var{ix})) \\
    \name{scalar} :: \Array\ \shvec{}\ \tau \to \tau
\end{array} \]
In other words: run the normal algorithm for every scalar in the output, and collect the results.
This approach also generalises conceptually to more complicated output types involving e.g.\ tuples, but the resulting transformation becomes rather notation-heavy and is omitted here.

\section{Staging (Embedding in Haskell)}\label{sec:staging}

\MK{I forgot how good the old staging/algebra sections are. It would be stupid to waste them. While not impressive technically, especially to an AD person, they ace the hand-wavy style of top-down presentation, contrasted with the rigorous bottom-up approach of the new paper text. I think this is the part that has a whiff of a functional pearl, if only it was easier to formulate the problems it solves and if the earlier 10 sections were not needed as a context. E.g., the really neat part is "This is good, because it means that we no longer have the entangling of the primal and dual halves of the AD output that we started out with at the beginning of \cref{sec:restaging}." and also getting symbolic AD for free (and via dual numbers!). I wonder if somebody could get a publishable functional pearl from this by glossing over the 10 sections (it's published elsewhere (on arXiv), so just trust me here).}

One of the goals of our library is to be a \emph{library}: not only is it a hassle for a user to introduce additional code preprocessors or compiler plugins into their workflow in order to use a nice AD algorithm for array programs, it would also be a higher maintenance burden for us: a code preprocessor must diligently stay up to date with the latest changes and additions to the language syntax, and a compiler plugin must stay up to date with ever-changing compiler internals.
A library exposing an embedded language does not have these problems.
An additional advantage of implementing an embedded language, as compared to a separate language with a distinct compilation toolchain, is that it is easier to expose smaller steps of the compilation process to the user, allowing them to essentially customise and assemble their compiler.\TODOfootnote{%
    Compare to \url{https://arxiv.org/pdf/1710.08332}?
}
Despite this flexibility, type-safety of the compiler as well as with the user's other code is maintained by simply using the type checker of the host language --- in our case, Haskell.

Because we want to do a non-local code transformation (the \BOT, which we explain in detail in \cref{sec:bot}) on the program written by the user, we need a syntax tree of the embedded program --- in other words, we need a \emph{deep embedding}.
This automatically means that we get a level of staging in the interface to the library: when the user-written Haskell code runs, it generates code that gets interpreted by our library (\hordead).
In particular, instead of doing computation, every library method (that is part of the array interface) constructs a small bit of an \emph{abstract syntax tree (AST)}.
The staging that we get this way also allows the user to perform various kinds of meta-programming without us having to do anything for it; the downside is that when the user is writing their program, they have to be aware of this staging step, and that they have to make an explicit decision about what code is meta-programming and what code is embedded.
Staging in \hordead, implemented with type classes and understood in terms of universal algebra, is described in depth in \cref{sec:algebra-interp}.

\paragraph{Static control flow}
An important example of this meta-programming is \emph{static control flow}: control flow that does not end up in the program to be differentiated, but can only depend on statically-known parameters.
Such not-quite-dynamic control flow is common in probabilistic programming and machine learning.
By partially evaluating it away before interpreting the code as a program to be differentiated, we can express e.g.\ loop unrolling, or assembling model components from various bits and pieces depending on external information.
As an example of loop unrolling, consider the following Haskell code:\footnote{Here we stick to the simplified $\Array\ n\ \tau$ notation from the rest of the paper, but in a real Haskell code using the \hordead library we would write the same as $\name{Concrete}\ (\name{TKR}\ n\ \tau)$}
\begin{equation*}
\begin{array}{@{}l@{}}
\var{sillyAlt}\ ::\ \Int \to \Array\ 1\ \Float \to \Array\ 1\ \Float \\
\var{sillyAlt}\ 0\ \var v = \var v \\
\var{sillyAlt}\ \var n\ \var v\ \begin{array}[t]{@{}l@{\ }l@{}}
        {\mid}\ \name{even}\ \var n & = \name{rmap}\ (\hslam x 2 \cdot x)\ (\var{sillyAlt}\ (\var n - 1)\ \var v) \\
        {\mid}\ \name{otherwise} & = \name{rmap}\ (\hslam x x + 1)\ (\var{sillyAlt}\ (\var n - 1)\ \var v)
    \end{array}
\end{array}
\end{equation*}
Recall that the library-provided functions build up a small bit of AST instead of doing actual computation.
Thus, the function $\var{sillyAlt}$ builds up a computation consisting of $\var n$ layers, each here an elementwise `$\name{rmap}$', \hordead's function to map a function element-wise over an array.\footnote{The `r' is for \emph{ranked}, meaning that array ranks are reflected on the type level; the library also has a version of the array language for shape-typed arrays (with full shapes on the type-level) as well as a mixed variant.}
In contrast to $\Array$ (by which we mean the type of embedded arrays in \hordead), $\Int$ is not an embedded type, thus all that \hordead sees is various invocations of $\name{rmap}$ nested inside each other.

\TODO{Andrew says: $\bullet$ JAX unrolls too much: a 96-deep transformer ends up with a 96x larger gradient program than the 1-deep version, $\bullet$ but given the tape memory usage you have anyway, this is not too bad. $\bullet$ But we should show what \hordead does on a transformer ``+FFN''(?).}

Because Haskell-native operations (i.e.\ not from the library, such as plain \code{if}-expressions) only type-check on meta-values, not embedded values, and vice-versa embedded operations only type-check on embedded values, whether an operation is staged is fully apparent from the types of the values being operated on.
An `$\Int$' is not staged, but an `$\Array\ 0\ \Int$' --- i.e.\ a zero-dimensional array of integers and hence also representing a single integer --- is staged because `$\Array$' is an embedded type.

\paragraph{Dynamic control flow}
In addition to static control flow, which is evaluated away in staging, \hordead supports a limited form of dynamic control flow: conditionals.
(Loops of statically unknown length are currently unsupported due to the difficulty in handling them in the \BOT (\cref{sec:bot}) and in AD (\cref{sec:ad-dual-arrays}).)
These dynamic conditionals are exposed via an \emph{embedded} if-expression, which takes an embedded boolean expression and embedded alternatives.
\TODO{example?}

\paragraph{Sharing}
A downside of implementing a deeply-embedded language via staging is that it is easy to lose sharing introduced by the user in the form of let-bindings and similar constructs.
For example, if the user writes:
\[
    \textbf{let}\ x = \textit{expensive}\ \textbf{in}\ f\ x + g\ x
\]
then tracing and staging this program as described above will reference the AST produced by the expression `$\textit{expensive}$' at least twice.\footnote{Assuming that $f$ and $g$ actually use their argument; more than twice if they use their argument multiple times.}
This is not what the user intended by writing the let-binding.
Approaches exist, in Haskell, to automatically detect and recover sharing of values between multiple positions in a data structure~\cite{2009-gill-sharing,2013-mcdonell-sharing}, but these are non-trivial to implement.\footnote{Personal experience of one of the paper authors with the Accelerate compiler is that it also becomes fragile when processing many source files in parallel.}
(Furthermore, while a standard common-subexpression elimination (CSE) pass in a compiler might recover the sharing as well, such as pass would be very slow due to having to analyse the full exponentially-sized unfolded AST.)
In \hordead, at least for the time being, we instead choose the simpler alternative of mirroring the solution for conditionals described above.
We offer a combinator `$\name{tlet}$' using which the example can be expressed as follows:
\[
    \name{tlet}\ \textit{expensive}\ \$\,\backslash x \to f\ x + g\ x
\]
Because this combinator is implemented by \hordead, its (explicit) sharing can be retained throughout the compilation pipeline.
See \cref{sec:algebra-interp-sharing} for a more detailed discussion of sharing in the context of the type-class system of \hordead.

\TODO{Do we need to say something about how the tracing works here?}

\paragraph{Limitations of the representation}
By tracing through the user-program at runtime, including lambda abstractions passed as arguments to built-in operations like `$\name{rmap}$', \hordead collects an AST of the embedded program; it is this program that will be transformed, differentiated and executed.
\Cref{fig:core-grammar} gives the grammar for the language that this embedded program is expressed in.

With the grammar as given in \cref{fig:core-grammar}, it is impossible to represent shared computation between the components of the index returned by the lambda in a \cc{gather} or \cc{scatter} operation.
This is because the `$ix$' production in the grammar does not admit let-bindings; it is simply a list of terms.

This lack of generality makes implementation easier, but is not fundamental: index-typed values could be first-class in the type system of the core language without causing significant trouble in later stages of \hordead.\footnote{This feature is tracked for the implementation at \url{https://github.com/Mikolaj/horde-ad/issues/119}.}


\section{Algebra Interpretation}\label{sec:algebra-interp}

The grammar of our core language in \cref{fig:core-grammar} not only informs the structure of an AST to \emph{represent} terms of this language at runtime, it also specifies the language itself, the one that the library user writes programs in: independent from any representation, our ``core language'' consists of a number of syntactic constructs (the ones in \cref{fig:core-grammar}) together with a semantics for those constructs (namely, their standard interpretation as array operations).
However, having just one semantics is sometimes quite limiting: one might want to evaluate programs written in the same syntax using a \emph{different} semantics, for example to compute certain program analyses or to perform partial evaluation.
Furthermore, as we will see in \cref{sec:algebra-ad}, the AD that we did in \cref{sec:ad-naive,sec:ad-dual-arrays} can also be seen as an alternative semantics for our syntax, as can indeed ASTs themselves: the latter is how we will fix the repeated re-differentiation problem identified at the beginning of \cref{sec:restaging}.

Mathematically, the language of array operations set out in \cref{fig:core-grammar} induces a family of \emph{algebras}:\footnote{Correctly: a category of $F$-algebras, where the functor $F$ is induced by the syntax in \cref{fig:core-grammar}.} each such algebra is a semantics of the array language on some \emph{carrier} data type.
We can encode this family of algebras in Haskell using a type class:\footnote{The careful reader may note that it is unclear what the expected sharing behaviour of the methods of this type class is. We will clarify the (somewhat subtle) situation after having introduced the basic instances.}
\[ \begin{array}{@{}ll@{}}
    \begin{array}{@{}l@{}}
        \mathbf{class}\ \cTensor\ t\ \mathbf{where} \quad\ccomment{not the final version! See below.} \\
        \quad \begin{array}[t]{@{}l@{\;}c@{\;}l@{}}
            \name{tconcrete} &::& \Array\ k\ \tau \to t\ k\ \tau \\
            \name{tlet} &::& t\ k\ \sigma \to (t\ k\ \sigma \to t\ k\ \tau) \to t\ k\ \tau \\
            \name{tindex} &::& t\ k\ \tau \to \Ix\ k \to t\ k\ \tau \\
            \name{tgather} &::& \multilineT{
                \Sh\ (m + k) \to t\ (n + k)\ \tau \to \\
                \qquad (\Ix\ m \to \Ix\ n) \to t\ (m + k)\ \tau
            } \\
            \mathrlap{\ccomment{... other methods ...}}
        \end{array}
    \end{array}
    &
    \begin{array}{@{}l@{}}
        \mathbf{data}\ \Sh\ k\ \mathbf{where} \\
        \quad \begin{array}[t]{@{}l@{\;}l@{\;}l@{}}
            \ShZ &::& \Sh\ 0 \\
            (\ShS) &::& \Int \to \Sh\ k \to \Sh\ (1 + k)
        \end{array}
    \end{array}
\end{array} \]
The data type $\Sh$, encoding shapes, is analogous to the $\Ix$ data type for indices defined in \cref{sec:ad-dual-arrays}.
Regarding the type of `$\name{tlet}$': in order to be independent of the particular representation of variables in the various semantics for our language, we use higher-order abstract syntax (HOAS) style to encode let-bindings: `$\cletin{x = s}{t}$' corresponds to `$\name{tlet}\ s\ (\lambda x.\ t)$'.
Consequently, there is no ``$\name{tvar}$'' method in $\cTensor$.

The idea is that each algebra in the family is an \emph{instance} of this $\cTensor$ type class for the appropriate carrier data type.
The implementation of the methods for the instance shows in what way the carrier indeed forms an algebra for our language.
In other words: the type class is the interface that every proposed carrier must implement in order to be a semantics for our language.
For example, there would be an instance of $\cTensor$ for $\Array$, yielding standard evaluation semantics (a little functional array language); for more details, see below.

The type class may look fine for this purpose at first glance, but to make the plan actually work out, we have to change two things:
\begin{enumerate}
\item
    The data type $\Ix$ was originally defined as follows in \cref{sec:ad-dual-arrays}:
    \[ \begin{array}{@{}l@{}}
        \mathbf{data}\ \Ix\ k\ \mathbf{where} \\
        \quad \begin{array}[t]{@{}l@{\ }l@{}}
            \IZ &:: \Ix\ 0 \\
            (\IS) &:: \Int \to \Ix\ k \to \Ix\ (k + 1)
        \end{array}
    \end{array} \]
    but for use in the type class, we must generalise this.
    The reason is that different semantics (such as \emph{symbolic} array computations, i.e.\ ASTs, as we will see below) have different ideas about what the `$\Int$' inside $\Ix$ should be.
    (Indeed, in a symbolic array computation, indices are also symbolic.)
    As a solution, we let the index components be rank-zero tensors:
    \[ \begin{array}{@{}l@{}}
        \mathbf{data}\ \Ix\ \textcolor{red}{t}\ k\ \mathbf{where} \\
        \quad \begin{array}[t]{@{}l@{\ }l@{}}
            \IZ &:: \Ix\ \textcolor{red}{t}\ 0 \\
            (\IS) &:: \textcolor{red}{t\ 0}\ \Int \to \Ix\ \textcolor{red}{t}\ k \to \Ix\ \textcolor{red}{t}\ (1 + k)
        \end{array}
    \end{array} \]
    Thus if $t = \Array$, this definition is morally the same as the original, as an $\Array\ 0\ \Int$ is equivalent to a single $\Int$.

\item
    It turns out that sharing using `$\name{tlet}$' is insufficient if we want to write dual-numbers reverse AD as an instance of $\cTensor$, i.e.\ as a semantics of our array language.
    (And we do, because this is what will enable us to disentangle primal and dual ($\tDelta$) in a compositional way --- the reason why we started this subsection in the first place.)
    We need to add a second form of sharing that mirrors the $\del{Share}$ constructor of $\tDelta$:
    \[ \begin{array}{@{}l@{}}
        \textcolor{gray}{\mathbf{class}\ \cTensor\ t\ \mathbf{where}} \\
        \quad \begin{array}[t]{@{}l@{\;}c@{\;}l@{}}
            \name{tshare} &::& t\ k\ \tau \to t\ k\ \tau \\
            \mathrlap{\ccomment{other methods ...}}
        \end{array}
    \end{array} \]
    The intended meaning of `$\name{tshare}$' is to produce a tensor that can be duplicated without causing any recomputation.
    The required ID is generated inside `$\name{tshare}$'.
    We will look at methods of sharing again after we have defined some basic instances of $\cTensor$.
\end{enumerate}

After these two modifications, the class looks as follows:
\[ \begin{array}{@{}l@{}}
    \mathbf{class}\ \cTensor\ t\ \mathbf{where} \\
    \quad \begin{array}[t]{@{}l@{\;}c@{\;}l@{}}
        \name{tconcrete} &::& \Array\ k\ \tau \to t\ k\ \tau \\
        \name{tlet} &::& t\ k\ \sigma \to (t\ k\ \sigma \to t\ k\ \tau) \to t\ k\ \tau \\
        \name{tshare} &::& t\ k\ \tau \to t\ k\ \tau \\
        \name{tindex} &::& t\ k\ \tau \to \Ix\ t\ k \to t\ k\ \tau \\
        \name{tgather} &::& \Sh\ (m + k) \to t\ (n + k)\ \tau \to (\Ix\ t\ m \to \Ix\ t\ n) \to t\ (m + k)\ \tau \\
        \mathrlap{\ccomment{... other methods ...}}
    \end{array}
\end{array} \]
This class formulation does work.

\subsection{Basic instances}
A natural instance of the $\cTensor$ type class is the standard evaluation semantics of our language.
Its carrier is the $\Array$ data type (in \hordead named \name{Concrete} to underscore these are normal physical arrays, not containing any symbolic components), and the interpretations of the language constructs are the usual call-by-value ones on concrete arrays:
\[ \begin{array}{@{}l@{}}
    \mathbf{instance}\ \cTensor\ \Array\ \mathbf{where} \\
    \quad \begin{array}[t]{@{}l@{\;}c@{\;}l@{}}
        \name{tconcrete}\ a &=& a \\
        \name{tlet}\ a\ f &=& f\ a \qquad\ccomment{the metalanguage \rlap{handles sharing.}} \\
        \name{tshare}\ a &=& \mathrlap{a}\hphantom{f\ a} \qquad\ccomment{ditto.} \\
        \name{tindex}\ a\ i &=& \name{index}\ a\ i \\
        \name{tgather}\ \var{sh}\ a\ f &=& \name{gather}\ \var{sh}\ a\ f \\
        \ccomment{etc.}
    \end{array}
\end{array} \]
assuming suitable methods `$\name{index}$', `$\name{gather}$', etc.\ on arrays.

Now, assume we have an AST representation for our language, for example using the following generalised algebraic data type (GADT):
\[ \begin{array}{@{}l@{}}
    \mathbf{data}\ \tAST\ k\ \tau\ \mathbf{where} \\
    \quad \begin{array}[t]{@{}l@{\;}c@{\;}l@{}}
        \texttt{Concrete} &::& \Array\ k\ \tau \to \tAST\ k\ \tau \\
        \texttt{Var} &::& \tVarName\ k\ \tau \to \tAST\ k\ \tau \\
        \texttt{Let} &::& \tVarName\ k\ \sigma \to \tAST\ k\ \sigma \to \tAST\ k\ \tau \to \tAST\ k\ \tau \\
        \texttt{Share} &::& \tID\ k\ \tau \to \tAST\ k\ \tau \to \tAST\ k\ \tau \\
        \texttt{Index} &::& \tAST\ k\ \tau \to \Ix\ k \to \tAST\ k\ \tau \\
        \texttt{Gather} &::& \Sh\ (m + k) \to \tAST\ (n + k)\ \tau \to (\Ix\ \tAST\ m \to \Ix\ \tAST\ n) \to \tAST\ (m + k)\ \tau \\
        \ccomment{etc.}
    \end{array}
\end{array} \]
For ASTs, we choose normal abstract syntax, as opposed to HOAS, to simplify handling in the below.\footnote{An alternative approach is to use PHOAS~\cite{2009-phoas}. \TODO{would that be nicer?}}
However, we do need a \texttt{Share} constructor to be able to implement the `$\name{tshare}$' method of $\cTensor$.
In contrast to our earlier presentation of $\tDelta$, where $\tDVarName$ and $\tID$ are indexed by just the rank of the array they represent (because, for simplicity, the element type is always $\R$), $\tVarName$ and $\tID$ here are indexed by both the rank and the element type.
\TODO{separate notation of Delta ID and AST ID}

This AST data type can also be the carrier of a semantics:
\[ \begin{array}{@{}l@{}}
    \mathbf{instance}\ \cTensor\ \tAST\ \mathbf{where} \\
    \quad \begin{array}[t]{@{}l@{\;}c@{\;}l@{}}
        \name{tconcrete}\ a &=& \texttt{Concrete}\ a \\
        \name{tlet}\ a\ f &=& \mathbf{let}\ v = \tVarName\ \genid\ \mathbf{in}\ \texttt{Let}\ v\ a\ (f\ (\texttt{Var}\ v)) \\
        \name{tshare}\ a &=& \texttt{Share}\ (\tID\ \genid)\ a \\
        \name{tindex}\ a\ i &=& \texttt{Index}\ a\ i \\
        \name{tgather}\ \var{sh}\ a\ f &=& \texttt{Gather}\ \var{sh}\ a\ f \\
        \ccomment{etc.}
    \end{array}
\end{array} \]
In this semantics, the ``meaning'' of a program is not its evaluation, but instead it is simply its AST.
A program \emph{evaluates to} its AST here.
This instance can be used to implement \emph{staging} in embedded languages: a shallowly embedded language (equivalently, a language in tagless-final style~\cite{2009-tagless-final}) is an embedded language where the programmer interface is essentially a type class like our $\cTensor$.
If such an embedded language implementation wants to do e.g.\ a whole-program transformation on the embedded program, it can use an instance of the type class for an AST data type to get an inspectable representation of the program.
This approach is called \emph{staging} because there are now two stages of evaluation: the user program evaluates to an AST, and this AST is later (presumably) evaluated itself to some final result.

In \cref{sec:staging}, we discuss the main design decisions and limitations of the staging implementation in \hordead.
The general framework of this implementations is as described here --- via Haskell type-classes.
In the next subsection we explain the interplay of sharing, a particularly subtle point of staging, with the algebra interpretation (type-classes) approach.

\subsection{Sharing}\label{sec:algebra-interp-sharing}
For technical reasons as well as to guard the efficiency of the full algorithm that this paper describes, the sharing-related semantics of the $\cTensor$ methods is a bit subtle.
Let us make clear what is going on.
\TODO{Improve the text, even though the situation is complicated and making it simpler is difficult.}
\begin{itemize}
\item
    `$\name{tlet}$' models a let-binding with lexical scoping.
    That is to say: the expression `$\name{tlet}\ a\ f$' is \emph{semantically} equivalent to $f\ a$, but `$\name{tlet}$' ensures that the resulting tensor (i.e.\ the result of `$\name{tlet}\ a\ f$') does not involve, or represent, multiple redundant computations of $a$.
    In the case of the instance for $\Array$, this is moot: assuming that the metalanguage (Haskell for this paper) has reference-passing semantics, `$\name{tlet}\ a\ f = f\ a$' fulfills this goal perfectly.

    However, for the instance for $\tAST$, this is quite important to get right.
    Consider the following (contrived) function written against the $\cTensor$ interface, taking an argument array $a$ of length $n$:
    \begin{center} \( \displaystyle \multilineM{
        \var{foo1} :: \cTensor\ t \Rightarrow t\ 1\ \R \to t\ 0\ \R \\
        \var{foo1}\ a = \name{tlet}\ \multilineT{
            (\name{tgather}\ (n \ShS \ShZ)\ a\ (\lambda(i \IS \IZ).\ (n - 1 - i) \IS \IZ)) \\
            (\lambda a'.\ \name{tindex}\ a'\ (0 \IS \IZ) +_\R \name{tindex}\ a'\ (1 \IS \IZ) +_\R \name{tindex}\ a'\ (2 \IS \IZ))
        }
    } \) \end{center}
    (`$+_\R$' is one of the methods of $\cTensor$ that we elide in the code snippets in this section to save space and to prevent tedious repetition. Its type is $(+_\R) :: t\ k\ \R \to t\ k\ \R \to t\ k\ \R$ and it is one of the binary $\mathit{op}$s in the grammar in \cref{fig:core-grammar}.)
    In $\var{foo1}$, the `$\name{tgather}$' computes the reverse of $a$, after which we take the sum of the first three elements of that computed reverse.
    If we instantiate $t$ to $\tAST$, then it is quite important that `$\name{tlet}\ a\ f$' is not simply `$f\ a$', but actually creates a \texttt{Let} node!
    Otherwise the produced AST will, when evaluated, recompute the reverse of the input array three times.

\item
    `$\name{tshare}$' models \emph{global sharing}.
    We have seen this global sharing before in the $\del{Share}$ constructor of $\tDelta$, where a $\del{Share}$-wrapped term could be used anywhere in the $\tDelta$ term and still be considered shared.
    Similarly, the intended meaning of `$\name{tshare}$' is that if the tensor that it returns is used in multiple places (probably as arguments to other $\cTensor$ methods of the same instance), this does not lead to duplicate computation of the tensor wrapped by `$\name{tshare}$'.
    The same example could be written as follows using $\name{tshare}$ instead:
    \begin{center} \( \displaystyle \multilineM{
        \var{foo2} :: \cTensor\ t \Rightarrow t\ 1\ \R \to t\ 0\ \R \\
        \var{foo2}\ a = \multilineT{
            \mathbf{let}\ a' = \name{tshare}\ (\name{tgather}\ (n \ShS \ShZ)\ a\ (\lambda(i \IS \IZ).\ (n - 1 - i) \IS \IZ)) \\
            \mathbf{in}\ \name{tindex}\ a'\ (0 \IS \IZ) +_\R \name{tindex}\ a'\ (1 \IS \IZ) +_\R \name{tindex}\ a'\ (2 \IS \IZ)
        }
    } \) \end{center}
    When instantiated to the $\tAST$ instance, the subterm corresponding to the meta-variable $a'$ will indeed occur three times in the resulting AST, but because all three occurrences are wrapped in a $\texttt{Share}$ constructor containing the same ID, an evaluator will memoise the computed value (the reversed array) and not recompute it the second and third time it encounters this same $\texttt{Share}$ node.
    This is analogous to how $\del{Share}$ nodes in a $\tDelta$ term were handled in $\reversePass$, except without the requirement that these globally shared subterms are computed in any particular order.

\item
    The other methods of the type class make \emph{no guarantees}, in general, about the sharing that they preserve.
    That is to say: the `$\name{tlet}$' in $\var{foo1}$ and the `$\name{tshare}$' in $\var{foo2}$ are necessary, because $(+_\R)$ and `$\name{tindex}$' may cheerfully assume their arguments are used only once.
    But not all instances will duplicate work: of course, if one somehow knows that $\var{foo1}$ is only going to be instantiated to the $\Array$ instance of $\cTensor$, a meta-language (Haskell) $\mathbf{let}$ expression suffices --- indeed, `$\name{tlet}$' and `$\name{tshare}$' for $\Array$ do not do anything more than that.
\end{itemize}
Later, when we define more instances of $\cTensor$, we will refer back to this and explain how those instances are consistent with these rules.

The $\interpret$ function described in \cref{sec:algebra-interp-interpretation} (that interprets an AST into some other instance of the $\cTensor$ class) can not easily support let-style sharing and global sharing ($\del{Share}$-style sharing, denoted in \hordead by `$\name{tshare}$') in the same term.
The reason is the fact that the $\cTensor$ methods (see below) are written in a higher-order fashion; not only `$\name{tlet}$' is, but also things like `$\name{tgather}$'.
A proper interpreter that handles global sharing correctly has to thread a memoisation map through the program, containing the evaluated result for every $\tID$ it encountered inside a share-node.
But given the type signature of `$\name{tlet}$', the interpreter has no way to export the $\tID$s it memoised inside the \emph{body} of the let, to outside that let!

Potential avenues for fixing this:
\begin{enumerate}
    \item \label{item:no-tshare-in-tlet}
        Observe that in practice, all our programs in the core language, be they expressed as $\cTensor$ combinators or as an AST, either have lets in them or global sharing --- never both at the same time.
        Crucially, this enables the assumption in $\interpret$ that `$\name{tshare}$' never occurs inside the body of a `let', meaning that no knowledge has to be exported out of the HOAS body of `$\name{tlet}$' at all.
        With this assumption, it is possible to write $\interpret$, and in practice that will work in our algorithm.

    \item
        Modify the type of `$\name{tlet}$' so that it explicitly allows a value to be returned from the body in the meta-language:
        \[
            \name{tlet'} :: t\ k\ \sigma \to (t\ k\ \sigma \to (a, t\ k\ \tau)) \to (a, t\ k\ \tau)
        \]
        This would allow $\interpret$ to get the output memoisation map from the body of the created let-binding.

        However, this does not actually work!
        The reason is that `$\name{tlet}$' is not the only higher-order combiantor in $\cTensor$; the $(\Ix\ t\ m \to \Ix\ t\ n)$ argument to `$\name{tgather}$' would need to get the same treatment.
        But this does not work: that argument might be called only once (for the $\tAST$ instance of $\cTensor$) or many times (for the $\Array$ instance of $\cTensor$), so neither returning an extra $a$ from the hypothetical $\name{tgather'}$, nor an extra $\Array\ m\ a$, would always work.
\end{enumerate}

The \hordead library implementation implements \ref{item:no-tshare-in-tlet} and the invariant of no `$\name{tshare}$' inside `$\name{tlet}$' and no `$\name{tlet}$' inside `$\name{tshare}$' is strictly enforced by the typing of the grammar.
Moreover, the typing ensures that $\interpret$ is only ever called on terms with no global sharing in them at all.

\subsection{Interpretation}\label{sec:algebra-interp-interpretation}
One of the reasons for having the $\cTensor$ type class is that one can write array computations polymorphic in the specific tensor type, and then later instantiate them to multiple backends, or assign them different semantics by instantiating them to non-standard instances.
However, we have already run the user program through the \BOT in \cref{sec:bot}, so we have an AST now, not a polymorphic function.
\newcommand\footnoteLetShareInterp{\footnote{
    It is, in principle, possible to interpret terms containing both local (using \texttt{Let}) and global (\texttt{Share}) sharing, but the interaction becomes very subtle and not all terms are valid: a \texttt{Share} node under a \texttt{Let} binding may not reference the let-bound variable if the \texttt{Share} node also occurs outside that \texttt{Let}.
    (Such terms cannot be constructed via the $\cTensor$ methods without a knot-tying hack.)
    Furthermore, our typing of `$\name{tlet}$' actually makes it impossible for $\interpret$ to correctly thread through the memoisation map for handling \texttt{Share} nodes.
    Fortunately, because we will never need to apply $\interpret$ to terms with global sharing in the full AD algorithm, we elect to leave $\interpret$ of \texttt{Share} unimplemented in \cref{fig:algebra-interpreter}.
}}%
\newcommand\footnoteCommuteInitial{\footnote{
    Assuming we want our interpretation to commute with primitives in our language, which seems quite reasonable.
    This requirement comes from universal algebra: the $\tAST$ algebra is a term algebra and thereby an \emph{initial} algebra in our family.
}}%
Fortunately, this is no obstacle, because the AST instance of $\cTensor$ is somewhat special: its values (ASTs)\footnoteLetShareInterp{} can be uniquely\footnoteCommuteInitial{} interpreted into any other semantics of the language.
The type of this interpretation function is as follows:
\[
    \interpret :: \cTensor\ t \Rightarrow \tDMaptwo\ \tVarName\ t \to \tAST\ k\ \tau \to t\ k\ \tau \\
\]
The first parameter (the $\tDMaptwo$) is the environment giving the interpretation for any free variables ($\texttt{Var}$) that occur in input term.
We define $\tDMaptwo$ in terms of $\tDMap$ by uncurrying $f$ and $g$:
\[
    \tDMaptwo\ f\ g = \tDMap\ (\lambda (k, \tau).\ f\ k\ \tau)\ (\lambda (k, \tau).\ g\ k\ \tau) \quad
\]
and then defining methods on $\tDMaptwo$ analogous to the ones on $\tDMap$.
For example, it is instructive to look at the types of the two versions of $\name{lookup}$:
\[ \begin{array}{@{}l@{\;}c@{\;}l@{}}
    \dmLookup &::& \texttt{GCompare}\ f \Rightarrow f\ a \to \tDMap\ f\ g \to \text{Maybe}\ (g\ a) \quad\ccomment{from \cref{fig:dmap-types}} \\
    \dmtwoLookup &::& \texttt{GCompare}_2\ f \Rightarrow f\ a\ b \to \tDMaptwo\ f\ g \to \text{Maybe}\ (g\ a\ b) \\
\end{array} \]
That is to say: $\tDMaptwo$ is to data types with 2 type parameters what $\tDMap$ is to data types with 1 type parameter.

As an example usage of $\interpret$, specialising the type variable $t$ to `$\Array$' and passing an empty initial enrivonment ($\dmtwoEmpty$), one obtains:
\[
    \interpret' :: \tAST\ k\ \tau \to \Array\ k\ \tau
\]
which evaluates a closed term to its value as an array.

Despite implementing a fairly fundamental function, $\interpret$ is somewhat cumbersome to write.
The definition we use is given in \cref{fig:algebra-interpreter}; let us walk through its major components.

\begin{figure}
\begin{align*}
&\begin{array}{@{}l@{}}
    \interpret :: \cTensor\ t \Rightarrow \tDMaptwo\ \tVarName\ t \to \tAST\ k\ \tau \to t\ k\ \tau \\
    \begin{array}{@{}l@{\;}c@{\;}l@{}}
    \interpret\ \var{env}\ (\texttt{Concrete}\ t)
        &=& \name{tconcrete}\ t \\
    \interpret\ \var{env}\ (\texttt{Var}\ v)
        &=& \mathbf{case}\ \dmtwoLookup\ v\ \var{env}\ \mathbf{of}\ \multilineT{
            \text{Just}\ x \to x \\
            \text{Nothing} \to \mathbf{error}\ \texttt{"Free variable"}
        } \\
    \interpret\ \var{env}\ (\texttt{Let}\ v\ s\ t)
        &=& \name{tlet}\ \multilineT{
            (\interpret\ \var{env}\ s) \\
            (\lambda x.\ \interpret\ (\dmtwoInsert\ v\ x\ \var{env})\ t)
        } \\
    \interpret\ \var{env}\ (\texttt{Share}\ \var{id}\ t)
        &=& \mathbf{error}\ \texttt{"Unimplemented"} \\
    \interpret\ \var{env}\ (\texttt{Index}\ t\ i)
        &=& \name{tindex}\ (\interpret\ \var{env}\ a)\ (\interpretIx\ \var{env}\ i) \\
    \interpret\ \var{env}\ (\texttt{Gather}\ \var{sh}\ a\ f)
        &=& \name{tgather}\ \var{sh}\ (\interpret\ \var{env}\ a)\ (\interpretIxFun\ \var{env}\ f) \\
    \ccomment{etc.}
    \end{array}
\end{array} \\
&\begin{array}{@{}l@{}}
    \interpretIx :: \cTensor\ t \Rightarrow \tDMaptwo\ \tVarName\ t \to \Ix\ \tAST\ k \to \Ix\ t\ k \\
    \begin{array}{@{}l@{\;}c@{\;}l@{}}
    \interpretIx\ \var{env}\ \IZ &=& \IZ \\
    \interpretIx\ \var{env}\ (i \IS \var{ix}) &=& \interpret\ \var{env}\ i \IS \interpretIx\ \var{env}\ \var{ix}
    \end{array}
\end{array} \\
&\begin{array}{@{}l@{}}
    \interpretIxFun :: \cTensor\ t \Rightarrow \tDMaptwo\ \tVarName\ t \to (\Ix\ \tAST\ m \to \Ix\ \tAST\ n) \to (\Ix\ t\ m \to \Ix\ t\ n) \\
    \interpretIxFun\ \var{env}\ f\ \var{ix}
        = \mathbf{let}\ (\var{env}', \var{ix}') = \mathit{extend}\ \var{env}\ \var{ix}
        \ \mathbf{in}\ \interpretIx\ \var{env}'\ (f\ \var{ix}') \\
    \qquad \mathbf{where}\ \multilineT{
        \mathit{extend} :: \tDMaptwo\ \tVarName\ t \to \Ix\ t\ k \to (\tDMaptwo\ \tVarName\ t, \Ix\ \tAST\ k) \\
        \begin{array}{@{}l@{\;}c@{\;}l@{}}
        \mathit{extend}\ \var{env}\ \IZ &=& \var{env} \\
        \mathit{extend}\ \var{env}\ (i \IS \var{ix}) &=& \multilineT{
            \mathbf{let}\ \multilineT{
                (\var{env}', \var{ix}') = \mathit{extend}\ \var{env}\ \var{ix} \\
                v = \tVarName\ \genid
            } \\
            \mathbf{in}\ (\dmtwoInsert\ v\ i\ \var{env}', \texttt{Var}\ v \IS \var{ix}')
        }
        \end{array}
    }
\end{array}
\end{align*}
\caption{\label{fig:algebra-interpreter}
    Interpreter from $\tAST$ to an arbitrary instance of $\cTensor$.
    In other words: this implements the unique homomorphism (function that commutes with all the primitives in our language) from the initial algebra to another algebra on the same language.
    While the function is unique, its \emph{implementation} is not, but the one given here has the advantage of being parametrically polymorphic over all $\cTensor$ instances.
}
\end{figure}

The clauses of $\interpret$ itself map each of the $\tAST$ constructors (i.e.\ primitives in our language) to the corresponding $\cTensor$ method on $t$; sharing operations are mapped to sharing operations, and the environment is extended as needed for let-bindings.
Aside from subterms, ASTs also contain more complicated structures such as index values (in e.g.\ \texttt{Index}) and index mapping functions (in e.g.\ \texttt{Gather}).
An index (interpreted by $\interpretIx$) is just a list of terms.
To interpret an index mapping function ($\interpretIxFun$), we implement this diagram:
\begin{center}
\begin{tikzpicture}[xscale=1.4, yscale=0.7]
\node (ixt1) at (0, 0) {$\Ix\ t\ m$};
\node (extend) at (0, 1) {$\var{extend}$};
\node (ixa1) at (0, 2) {$\Ix\ \tAST\ m$};
\node (f) at (1, 2) {$f$};
\node (ixa2) at (2, 2) {$\Ix\ \tAST\ n$};
\node (inter) at (2, 1) {$\interpretIx$};
\node (ixt2) at (2, 0) {$\Ix\ t\ n$};
\draw[->] (ixt1) -- (extend) -- (ixa1);
\draw[->] (ixa1) -- (f) -- (ixa2);
\draw[->] (ixa2) -- (inter) -- (ixt2);
\draw[->] (extend) to[bend left=25] node[below] {$\var{env}$} (inter);
\end{tikzpicture}
\end{center}
That is to say: we generate variable names for the components of the input index, pass the resulting symbolic index through the symbolic mapping function, and then compute the output index by evaluating the output symbolic index with the generated variable names mapped to the components of the input index.

\section{AD as an Algebra Interpretation}\label{sec:algebra-ad}

\TODO{Need to cite prior work on functorial AD here}

It turns out that because of the nice, compositional nature of the AD code transformation that produces the forward pass (\cref{fig:ad-rules-dual-arrays,fig:naive-dn-ad-terms}), it can be written as an algebra interpretation on our core language.
The carrier data type here is a dual number (on arrays, of course, as described in \cref{sec:ad-dual-arrays}): a pair of an array and a $\tDelta$ term.
\[ \begin{array}{@{}l@{}}
    \mathbf{data}\ \tADVal\ k\ \tau = \tADVal\ (\Array\ k\ \tau)\ (\tDelta\ k) \\
    \mathbf{instance}\ \cTensor\ \tADVal\ \mathbf{where} \\
    \quad \begin{array}{@{}l@{\;}c@{\;}l@{}}
        \name{tconcrete}\ a &=& \tADVal\ a\ \del{Zero} \\
        \name{tlet}\ (\tADVal\ p\ d)\ f &=& f\ (\tADVal\ p\ d) \\
        \name{tshare}\ (\tADVal\ p\ d) &=& \tADVal\ p\ d \\
        \name{tindex}\ (\tADVal\ p\ d)\ \var{ix} &=& \tADVal\ (\name{tindex}\ p\ \var{ix})\ (\del{Share}\ \genid\ (\del{Index}\ d\ \var{ix})) \\
        \name{tgather}\ \var{sh}\ (\tADVal\ p\ d)\ f &=& \tADVal\ (\name{tgather}\ \var{sh}\ p\ f)\ (\del{Share}\ \genid\ (\del{Gather}\ d\ f)) \\
        \ccomment{etc.}
    \end{array}
\end{array} \]
The implementations of the $\cTensor$ methods for $\tADVal$ are adapted directly from the equations of the $D$ code transformation from \cref{sec:ad-dual-arrays} (\cref{fig:ad-rules-dual-arrays}), with the (meta-)pairs used there replaced by uses of the $\tADVal$ constructor.
The implementations of `$\name{tlet}$' and `$\name{tshare}$', however, require some justification.
As before in the $\Array$ instance, the sharing operations ($\name{tlet}$ and $\name{tshare}$) have no effect on the primal (left) half because Haskell does not recompute values when you use them multiple times.
For the dual (right) half, however, we can also ignore the sharing operations: because the $\tDelta$ terms created by the other methods in the instance are always wrapped inside a $\del{Share}$ node and thus freely duplicable in the meta-language, $d$ in the argument to `$\name{tlet}$' or `$\name{tshare}$' will be freely duplicable.
Hence, wrapping it in another $\del{Share}$ node does not achieve anything, and we choose to omit the redundant wrapper for efficiency.

Having the carrier be a \emph{pair} of two tensor-like things (an array and a $\tDelta$ term, in this case) means that any computation that is interpreted into this algebra, gets repeated twice: once on arrays and once on $\tDelta$ terms.
Furthermore, there is no dependency between these two computations in any of the method implementations in the instance.
Thus, when a $\cTensor$-polymorphic function, for example:
\[
    \var{dotprod} :: \cTensor\ t \Rightarrow t\ 1\ \R \to t\ 1\ \R \to t\ 0\ \R
\]
is interpreted in the dual-numbers algebra $\tADVal$, its result \emph{is a pair} of a primal result and a $\tDelta$ term that by construction are fully separated.
This is good, because it means that we no longer have the entangling of the primal and dual halves of the AD output that we started out with at the beginning of \cref{sec:restaging}.

\subsection{Generalisation}
Looking at the $\cTensor$ instance for $\tADVal$, we notice that the primal halves of the returned pairs simply mirror the $\cTensor$ methods they are implementing: `$\name{tindex}$' maps to `$\name{tindex}$', etc., as expected.
After all, the primal computation of a derivative program performs the same computations on scalars and arrays as the original program did.
This means that we can generalise $\tADVal$: its primal component need not be an $\Array$, and could instead be any $\cTensor$ type.
Thus, as a first attempt we can try to simply parametrise $\tADVal$ on the tensor algebra used for the primal operations like this:
\[
    \mathbf{data}\ \tADVal'\ t\ k\ a = \tADVal\ (t\ k\ a)\ (\tDelta\ k)
\]
but this does not quite work.
The reason is that primal tensors end up in the $\tDelta$ term as well.
Consider the rule for multiplication: (compare the original code transformation for scalars in \cref{fig:dn-ad-terms})
\[ \begin{array}{@{}l@{}}
    \textcolor{gray}{\mathbf{instance}\ \cTensor\ \tADVal\ \mathbf{where}} \\
    \quad \tADVal\ p_1\ d_1 \times_\R \tADVal\ p_2\ d_2
        = \tADVal\ \multilineT{
            (p_1 \times_\R p_2) \\
            (\del{Share}\ \genid\ (\del{Add}\ (\del{Scale}\ p_2\ d_1)\ (\del{Scale}\ p_1\ d_2)))
        }
\end{array} \]
The $\del{Scale}$ constructor of $\tDelta$ contains a tensor from the primal half of the dual-numbers pair, so if we generalise the primal tensor type, we must generalise the tensor type in $\del{Scale}$ as well.
Let us do so:
\[ \begin{array}{@{}l@{}}
    \mathbf{data}\ \tDelta\ \textcolor{red}{t}\ k\ \mathbf{where} \\
    \quad \multilineT{
        \ccomment{... other constructors ...} \\
        \del{Scale} :: \textcolor{red}{t}\ k\ \R \to \tDelta\ \textcolor{red}{t}\ k \to \tDelta\ \textcolor{red}{t}\ k \\
        \ccomment{... other constructors ...} \\
    } \vspace{0.5em} \\
    \mathbf{data}\ \tADVal\ t\ k\ a = \tADVal\ (t\ k\ a)\ (\tDelta\ \textcolor{red}{t}\ k)
\end{array} \]
Now it becomes straightforward to lift the previous $\cTensor$ instance for $\tADVal$ to the generalised, parametrised $\tADVal$; we just need to take care to use `$\name{tconcrete}$' explicitly for injecting constants into the primal computation, and to use explicit `$\name{tshare}$' on primal terms when duplicating them:
\[ \begin{array}{@{}l@{}}
    \mathbf{instance}\ \cTensor\ t \Rightarrow \cTensor\ (\tADVal\ t)\ \mathbf{where} \\
    \quad \begin{array}{@{}l@{\;}c@{\;}l@{}}
        \textcolor{gray}{\name{tconcrete}\ a} &\textcolor{gray}{=}& \textcolor{gray}{\tADVal}\ (\name{tconcrete}\ a)\ \textcolor{gray}{\del{Zero}} \\
        \textcolor{gray}{\name{tlet}\ (\tADVal\ p\ d)\ f} &\textcolor{gray}{=}& \textcolor{gray}{f\ (\tADVal}\ (\name{tshare}\ p)\ \textcolor{gray}{d)} \\
        \textcolor{gray}{\name{tshare}\ (\tADVal\ p\ d)} &\textcolor{gray}{=}& \textcolor{gray}{\tADVal}\ (\name{tshare}\ p)\ \textcolor{gray}{d} \\
        \textcolor{gray}{\name{tindex}\ (\tADVal\ p\ d)\ \var{ix}} &\textcolor{gray}{=}& \textcolor{gray}{\tADVal\ (\name{tindex}\ p\ \var{ix})\ (\del{Share}\ \genid\ (\del{Index}\ d\ \var{ix}))} \\
        \textcolor{gray}{\name{tgather}\ \var{sh}\ (\tADVal\ p\ d)\ f} &\textcolor{gray}{=}& \textcolor{gray}{\tADVal\ (\name{tgather}\ \var{sh}\ p\ f)\ (\del{Share}\ \genid\ (\del{Gather}\ d\ f))} \\
        \ccomment{... other methods ...}
    \end{array}
\end{array} \]
Most of the code remains unchanged; the result is a parametrised algebra interpretation\footnote{Formally, this defines a homomorphism between the algebras $t$ and $\tADVal\ t$, and is an example of a derived algebra morphism.} into dual arrays.

Notable in this $\cTensor$ instance is that let-bindings are interpreted into global sharing.
The reason for this choice is quite subtle: because primal tensors end up inside $\tDelta$ terms, and $\tDelta$ terms are not lexically-scoped subterms of our primal computation (after all, we want to \emph{disentangle} the two, not interleave both in the same computation!), we cannot use `$\name{tlet}$' to interpret the sharing in the primal computation.
Now that we use global sharing instead, we can choose to scope the namespace of $\tID$s of globally shared primal values over the whole computation, not just the primal half; that way, the dual computation depends on the primal computation, but not the other way round.
Hence, the two are still disentangled.

In effect, we thus create a single namespace of $\tID$s for primal tensors over all primal and dual values in an $\tADVal$ computation, and separately a namespace of $\tID$s just within the $\tDelta$ terms.
The primal tensor $\tID$s are referenced in the primal computation as well as in the embedded tensor values inside $\del{Scale}$ in $\tDelta$, using whatever method the tensor type $t$ uses to record global sharing; the $\tID$s referring to $\tDelta$ terms are just referenced using $\del{Share}$ constructors.
These namespaces are disjoint, because the former encodes sharing of primal tensors and the latter encodes sharing of $\tDelta$ terms, which are different types.

\subsection{Instantiation}
Now is the time when we can finally solve the problem of repeated re-differentiation that we set out to solve at the beginning of \cref{sec:restaging}.
The trick is that we can instantiate this parametrised algebra interpretation ($\tADVal$) to ASTs, yielding `$\tADVal\ \tAST\ k\ a$': a pair of an AST and a $\tDelta$ term containing ASTs.
\[
    \tADVal\ \tAST\ k\ a \approx (\tAST\ k\ a, \tDelta\ \tAST\ k)
\]
We call these pairs \emph{symbolic dual arrays}.
Because we have a \cTensor instance for $\tAST$ and a parametrised one for $\tADVal$ as shown above, this instantiated type is also an instance of \cTensor.
This means that we can interpret programs into it using our $\interpret$ function!
What does the result look like?
\begin{itemize}
\item
    Where the original program took arrays as input, the reinterpreted program takes symbolic dual arrays as input.
    In particular, the reinterpreted program can be run to completion (symbolically) without supplying concrete input arguments: free-variable $\tAST$ nodes suffice.

\item
    Because of the simple design of our core language, a program returns exactly one tensor as output.
    Thus, the program output will be one symbolic dual array.

\item
    The primal half of this dual array is an $\tAST$ that, when evaluated (i.e.\ interpreted into a concrete array algebra) computes the original value of the program.
    Note that this $\tAST$ uses global sharing, so before it can be interpreted, it needs the global sharing transformed into local sharing (`\name{tshare}' into `\name{tlet}'), analogous to the conversion in \cref{sec:restaging-unglobal}.

\item
    The dual half of the output dual array is a $\tDelta$ term containing $\tAST$s (that reference values computed in the primal half): this describes the (symbolic) forward derivative of the program evaluated at the given (symbolic) inputs.
\end{itemize}

In \cref{sec:ad-dual-arrays}, evaluation of a $\tDelta$ term proceeded by passing it to $\reversePass$, which took an incoming cotangent and a (non-symbolic) $\tDelta$ term and produced a sparse gradient, which could be materialised into a full gradient in the wrapper around the algorithm.
Surprisingly, the operations that $\eval$, and hence $\reversePass$, performs on the cotangents are precisely those that comprise the core language: this works because we have designed our core language to be closed under differentiation (assuming sufficient primitive arithmetic operators).
For example, we have not only $\cc{gather}$ but also $\cc{scatter}$, and not only $\cc{sumOuter}$ but also $\cc{replicate}$.

Hence, we can generalise $\reversePass$ to work on arbitrary tensor algebras.
Doing this, we end up with the following types:
\begin{align*}
    &\begin{array}{@{}l@{}}
        \mathbf{data}\ \EState\ t = \EState \\
        \quad \{ \begin{array}[t]{@{}l@{}}
            \ \name{grad} :: \tDMap\ \tDVarName\ (\lambda k.\ t\ k\ \R) \\
            \leadcomma\ \name{dfrag} :: \tDMap\ \tID\ (\tDelta\ t) \\
            \leadcomma\ \name{accum} :: \tDMap\ \tID\ (\lambda k.\ t\ k\ \R)\ \} \\
        \end{array}
    \end{array} \\
    &\begin{array}{@{}l@{\ }l@{}}
        \reversePass &:: \cTensor\ t \Rightarrow t\ k\ \R \to \tDelta\ t\ k \to \tDMap\ \tDVarName\ (\lambda k.\ t\ k\ \R) \\
        \eval &:: \cTensor\ t \Rightarrow t\ k\ \R \to \tDelta\ t\ k \to \EState\ t \to \EState\ t \\
        \backprop &:: \cTensor\ t \Rightarrow \EState\ t \to \EState\ t
    \end{array}
\end{align*}
with basically identical implementations to those given in \cref{sec:ad-dual-arrays}.

By this point, we have a fairly complete implementation of reverse AD for our language, designed and implemented in a compositional, modular manner.
In particular, as a result of the compositionality of the design:
\begin{itemize}
\item If we remove \BOT, the AD algorithm can be written directly as a \emph{shallow embedding}\footnote{This is also known as a \emph{final} encoding of the algorithm, as opposed to an \emph{initial} encoding which goes via an initial algebra, i.e.\ an algebraic data type (the AST).}, of course losing efficient differentiation of array indexing, but gaining expressiveness of the source language (more expressive dynamic control flow) because there is no staging any more: all control flow is traced away.
\item The AD algorithm is completely decoupled from \BOT, the interface being purely an AST of the core language (\cref{fig:core-grammar}).
\end{itemize}
What is missing for a comprehensive picture is an overview of the full pipeline, and the wrapper around the algorithm that makes it usable.
These are not too hard to derive by generalizing the pipelines and wrappers from the previous sections and their implementation can be inspected in the \hordead source code.

\bibliography{bibliography}

\end{document}